\documentclass[sigconf,natbib=true,sort&compress]{acmart}

\usepackage[latin1]{inputenc}
\usepackage[T1]{fontenc}
\usepackage{amsmath}
\usepackage{mleftright}
\usepackage{graphicx}
\usepackage[inline]{enumitem}
\usepackage{nicematrix}
\usepackage{xcolor}
\usepackage{nicefrac}
\usepackage{xspace}
\usepackage[capitalize]{cleveref}
\usepackage{multirow}
\usepackage{tikz}
\usetikzlibrary{trees,arrows}
\usepackage{bookmark}

\usepackage{algorithm}
\usepackage[noend]{algpseudocode}
\usepackage{algorithmicx}

\algrenewcommand{\alglinenumber}[1]{\color{black}\footnotesize#1:}

\usepackage{tikz}
\usetikzlibrary{shapes,arrows}
\usepackage{subcaption}

\usepackage{thmtools}
\usepackage{thm-restate}

\newtheorem{cor}{Corollary}
\newtheorem{problem}{Problem}

\setlist[itemize]{label=--}

\definecolor{red}{HTML}{BF2300}	
\definecolor{silver}{gray}{0.7}

\newcommand{\ourmethod}{\textsc{Gamine}\xspace}
\newcommand{\oururl}{\href{https://doi.org/10.5281/zenodo.7936816}{10.5281/zenodo.7936816}}

\newcommand{\NP}{\textnormal{NP}\xspace}
\newcommand{\graph}{G}
\newcommand{\nodes}{V}
\newcommand{\edges}{E}
\newcommand{\outdegree}[1]{\delta^+\!(#1)}
\newcommand{\outregulardegree}{d}
\newcommand{\indegree}[1]{\delta^-\!(#1)}
\newcommand{\nnodes}{n}
\newcommand{\nedges}{m}
\newcommand{\cardinality}[1]{|#1|}
\newcommand{\outneighbors}[1]{\Gamma^+\!(#1)}
\newcommand{\labeling}{\ensuremath{c}\xspace}
\newcommand{\labelingvector}{\ensuremath{\mathbf{c}}\xspace}
\newcommand{\range}{[0,1]}
\newcommand{\edge}[2]{(#1,#2)}
\newcommand{\probability}[1]{p_{#1}}
\newcommand{\pabsorption}{\alpha}
\newcommand{\transitionmatrix}{\ensuremath{\mathbf{P}}\xspace}
\newcommand{\fundamental}{\ensuremath{\mathbf{F}}\xspace}
\newcommand{\identity}{\ensuremath{\mathbf{I}}}
\newcommand{\norm}[1]{\lVert #1 \rVert}
\newcommand{\budget}{r}
\newcommand{\onevector}{\ensuremath{\mathbf{1}}\xspace}
\newcommand{\objective}[1]{f(#1)}
\newcommand{\objectivef}{f}
\newcommand{\ivector}{\ensuremath{\mathbf{u}}\xspace}
\newcommand{\jkvector}{\ensuremath{\mathbf{v}}\xspace}
\newcommand{\unitvector}{\ensuremath{\mathbf{e}}\xspace}
\newcommand{\ourproblem}{\textsc{REM}\xspace}
\newcommand{\safenodes}{S}
\newcommand{\safenode}{s}
\newcommand{\unsafenodes}{U}
\newcommand{\mcovered}{\gamma}
\newcommand{\analgorithm}{\mathcal{A}}
\newcommand{\maxoutdegree}{\Delta^+}
\newcommand{\maxunsafeoutdegree}{\Lambda^+}
\newcommand{\selection}{\ensuremath{X}\xspace}
\newcommand{\reals}{\mathbb{R}}
\newcommand{\naturals}{\mathbb{N}}

\newcommand{\bigoh}[1]{\mathcal{O}(#1)}
\newcommand{\niter}{\kappa}

\newcommand{\quality}[1]{\theta(#1)}
\newcommand{\qualityf}{\theta}
\newcommand{\qualitythreshold}{q}
\newcommand{\fabbrialg}{\textsc{MMS}\xspace}
\newcommand{\baselineone}{\textsc{BL1}\xspace}
\newcommand{\baselinetwo}{\textsc{BL2}\xspace}
\newcommand{\baselinethree}{\textsc{BL3}\xspace}
\newcommand{\baselinefour}{\textsc{BL4}\xspace}
\newcommand{\relevancematrix}{\ensuremath{\mathbf{R}}\xspace}
\newcommand{\rank}{\text{idx}}
\newcommand{\dcg}{\textsc{DCG}\xspace}
\newcommand{\ndcg}{\textsc{nDCG}\xspace}
\newcommand{\idcg}{\textsc{iDCG}\xspace}
\newcommand{\topranked}{T}
\newcommand{\outdeg}{\delta^+}
\newcommand{\ourproblemtwo}{\textsc{QREM}\xspace}
\newcommand{\somematrix}{\ensuremath{\mathbf{M}}\xspace}
\newcommand{\apxerror}{\varepsilon}
\newcommand{\probabilityshape}{\chi}

\newcommand{\maxobjectivef}{\ensuremath{\objectivef_\Delta}\xspace}
\newcommand{\maxsetobjectivef}{\ensuremath{\hat{\objectivef}_\Delta}\xspace}
\newcommand{\maxobjective}[1]{\ensuremath{\objectivef_\Delta(#1)}\xspace}
\newcommand{\maxsetobjective}[1]{\ensuremath{\hat{\objectivef}_\Delta(#1)}\xspace}
\newcommand{\rewiring}[3]{(#1,#2,#3)}
\newcommand{\outseq}[1]{\mathbf{r}_#1}
\newcommand{\nqpermissible}{\ell}
\newcommand{\qualityset}{Q}
\newcommand{\qualitycomplexity}{g}
\newcommand{\qualitycomplexitystart}{h}
\newcommand{\cover}{C}
\newcommand{\distance}{t}
\newcommand{\smallterm}{\varepsilon}
\newcommand{\neglectedterm}{\xi}

\newcommand{\greedyname}{1-REM\xspace}
\newcommand{\kcandidates}{K}
\newcommand{\heuristic}{\hat{\Delta}}

\newcommand{\synuni}{\textsc{SU}\xspace}
\newcommand{\synhom}{\textsc{SH}\xspace}
\newcommand{\yt}{\textsc{YT}\xspace}
\newcommand{\nf}{\textsc{NF}\xspace}
\newcommand{\ytone}{\textsc{YT-100k}\xspace}
\newcommand{\yttwo}{\textsc{YT-10k}\xspace}
\newcommand{\nelaone}{\textsc{NF-Jan06}\xspace}
\newcommand{\nelatwo}{\textsc{NF-Cov19}\xspace}
\newcommand{\nelathree}{\textsc{NF-All}\xspace}
\newcommand{\roundingthreshold}{\mu}
\newcommand{\fractionbad}{\beta}

\title{Reducing Exposure to Harmful Content via Graph Rewiring}
\date{}

\copyrightyear{2023} 
\acmYear{2023} 
\setcopyright{rightsretained} 
\acmConference[KDD '23]{Proceedings of the 29th ACM SIGKDD Conference on Knowledge Discovery and Data Mining}{August 6--10, 2023}{Long Beach, CA, USA}
\acmBooktitle{Proceedings of the 29th ACM SIGKDD Conference on Knowledge Discovery and Data Mining (KDD '23), August 6--10, 2023, Long Beach, CA, USA}
\acmDOI{10.1145/3580305.3599489}
\acmISBN{979-8-4007-0103-0/23/08}
\acmSubmissionID{rtfp0318}

\begin{document}
	\author{Corinna Coupette}
	\orcid{0000-0001-9151-2092}
	\affiliation{%
		\institution{MPI for Informatics}
		\country{Germany}
	}
	\author{Stefan Neumann}
	\orcid{0000-0002-3981-1500}
	\affiliation{%
		\institution{KTH Royal Institute of Technology}
		\country{Sweden}
	}
	\author{Aristides Gionis}
	\orcid{0000-0002-5211-112X}
	\affiliation{%
		\institution{KTH Royal Institute of Technology}
		\country{Sweden}
	}

\renewcommand{\shortauthors}{Coupette et al.}

\begin{abstract}
Most media content consumed today is provided by digital platforms that aggregate input from diverse sources, 
where access to information is mediated by recommendation algorithms.
One principal challenge in this context
is dealing with content that is considered harmful. 
Striking a balance between competing stakeholder interests,
rather than block harmful content altogether, 
one approach is to minimize the exposure to such content that is induced specifically by algorithmic recommendations.
Hence, 
modeling media items and recommendations as a directed graph,
we study the problem of reducing the exposure to harmful content via edge rewiring. 
We formalize this problem using absorbing random walks, 
and prove that it is \NP-hard and \NP-hard to approximate to within an additive error, 
while under realistic assumptions,
the greedy method yields a $(1-\nicefrac{1}{e})$-approximation.
Thus, we introduce \ourmethod, 
a fast greedy algorithm that can reduce the exposure to harmful content with or without quality constraints on recommendations.
By performing just $100$ rewirings on YouTube graphs with several hundred thousand edges, 
\ourmethod reduces the initial exposure by $50\%$, 
while ensuring that its recommendations are at most $5\%$ less relevant than the original recommendations.
Through extensive experiments on synthetic data 
and real-world data from video recommendation and news feed applications,
we confirm the effectiveness, robustness, and efficiency of \ourmethod in practice.
\end{abstract}



\maketitle

\captionsetup[subfigure]{skip=0pt}
\captionsetup{skip=3pt}

\section{Introduction}
\label{sec:introduction}

Recommendation algorithms mediate access to content on digital platforms, 
and as such, they critically influence how individuals and societies perceive the world and form their opinions \cite{spinelli2017closed,robertson2018auditing,papadamou2022just,hussein2020measuring,ferrara2022link}. 
In recent years, 
platforms have come under increasing scrutiny from researchers and regulators alike
due to concerns and evidence that their recommendation algorithms create filter bubbles \cite{ledwich2022radical,srba2022auditing,kirdemir2022exploring,chitra2020analyzing} 
and fuel radicalization \cite{hosseinmardi2021examining,whittaker2021recommender,ribeiro2020auditing,ledwich2020algorithmic,pescetelli2022bots}. 
One of the main challenges in this context is dealing with content that is considered harmful \cite{bandy2021problematic,yesilada2022systematic,costanza2022audits}.
To address this challenge while balancing the interests of creators, users, and platforms,
rather than block harmful content, 
one approach is to minimize the exposure to such content that is induced by algorithmic recommendations. 

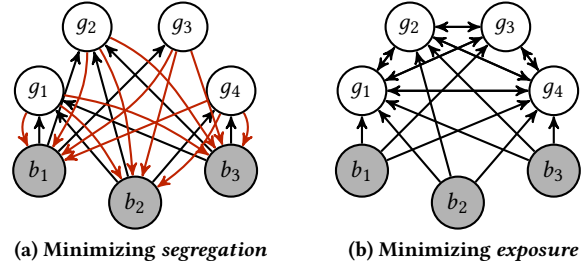
\begin{figure}[t]
	\centering
	\begin{subfigure}{0.5\linewidth}
		\centering
\begin{tikzpicture}[->,>=stealth',thick,scale=0.85]
	\node[minimum height = 2em,minimum width = 2em,draw,circle] at (0, 0)   (g1) {$g_1$};
	\node[minimum height = 2em,minimum width = 2em,draw,circle] at (0.75, 1)   (g2) {$g_2$};
	\node[minimum height = 2em,minimum width = 2em,draw,circle] at (2.25, 1)   (g3) {$g_3$};
	\node[minimum height = 2em,minimum width = 2em,draw,circle] at (3, 0)   (g4) {$g_4$};
	\node[minimum height = 2em,minimum width = 2em,draw,circle,fill=silver] at (0, -1.25)   (b1) {$b_1$};
	\node[minimum height = 2em,minimum width = 2em,draw,circle,fill=silver] at (1.5, -1.75)   (b2) {$b_2$};
	\node[minimum height = 2em,minimum width = 2em,draw,circle,fill=silver] at (3, -1.25)   (b3) {$b_3$};

	\draw (b1) -- (g1);
	\draw (b1) -- (g2);
	\draw (b1) -- (g3);
	\draw (b2) -- (g1);
	\draw (b2) -- (g2);
	\draw (b2) -- (g4);
	\draw (b3) -- (g1);
	\draw (b3) -- (g2);
	\draw (b3) -- (g4);
	\draw[red] (g3) to [out=245,in=35] (b1); 
	\draw[red] (g3) -- (b2); 
	\draw[red] (g3) -- (b3); 
	\draw[red] (g1) to [out=240,in=115] (b1);
	\draw[red] (g1) to [out=325,in=120] (b2); 
	\draw[red] (g1) to [out=350,in=145] (b3); 
	\draw[red] (g2) to [out=270,in=60] (b1); 
	\draw[red] (g2) to [out=300,in=95] (b2); 
	\draw[red] (g2) to [out=335,in=125] (b3); 
	\draw[red] (g4) -- (b1); 
	\draw[red] (g4) to [out=245,in=35] (b2); 
	\draw[red] (g4) to [out=300,in=65] (b3); 
\end{tikzpicture}
		\vspace*{3pt}
		\subcaption{Minimizing \emph{segregation}}
		\label{fig:setting:bad}
	\end{subfigure}~%
	\begin{subfigure}{0.5\linewidth}
		\centering

\begin{tikzpicture}[->,>=stealth',thick,scale=0.85]
	\node[minimum height = 2em,minimum width = 2em,draw,circle] at (0, 0)   (g1) {$g_1$};
	\node[minimum height = 2em,minimum width = 2em,draw,circle] at (0.75, 1)   (g2) {$g_2$};
	\node[minimum height = 2em,minimum width = 2em,draw,circle] at (2.25, 1)   (g3) {$g_3$};
	\node[minimum height = 2em,minimum width = 2em,draw,circle] at (3, 0)   (g4) {$g_4$};
	\node[minimum height = 2em,minimum width = 2em,draw,circle,fill=silver] at (0, -1.25)   (b1) {$b_1$};
	\node[minimum height = 2em,minimum width = 2em,draw,circle,fill=silver] at (1.5, -1.75)   (b2) {$b_2$};
	\node[minimum height = 2em,minimum width = 2em,draw,circle,fill=silver] at (3, -1.25)   (b3) {$b_3$};
	
	\draw (g1) -- (g2); 
	\draw (g1) -- (g3);
	\draw (g1) -- (g4);
	\draw (g2) -- (g1); 
	\draw (g2) -- (g3);
	\draw (g2) -- (g4);
	\draw (g3) -- (g1); 
	\draw (g3) -- (g2);
	\draw (g3) -- (g4);
	\draw (g4) -- (g1); 
	\draw (g4) -- (g2);
	\draw (g4) -- (g3);
	\draw (b1) -- (g1);
	\draw (b1) -- (g3); 
	\draw (b1) -- (g4);
	\draw (b2) -- (g1); 
	\draw (b2) -- (g2); 
	\draw (b2) -- (g4);
	\draw (b3) -- (g1); 
	\draw (b3) -- (g2);
	\draw (b3) -- (g4);
\end{tikzpicture}
		\vspace*{3pt}
		\subcaption{Minimizing \emph{exposure}}
		\label{fig:setting:good}
	\end{subfigure}
	\caption{%
		3-out-regular directed graphs with four \emph{good} nodes (white) and three \emph{bad} nodes (gray). 
		Edges running from good to bad nodes are drawn in red. 
		The left graph minimizes the \emph{segregation} objective from \citeauthor{fabbri2022rewiring} \cite{fabbri2022rewiring}, 
		but random walks oscillate between good nodes and bad nodes. 
		In contrast, only the right graph minimizes our \emph{exposure} objective.
	}\vspace*{-12pt}
	\label{fig:setting}
\end{figure}

In this paper, 
we study the problem of reducing the exposure to harmful content via \emph{edge rewiring}, 
i.e., replacing certain recommendations by others. 
This problem was recently introduced by \citeauthor{fabbri2022rewiring} \cite{fabbri2022rewiring},
who proposed to address it
by modeling harmfulness as a \emph{binary} node label and minimizing the \emph{maximum} \emph{segregation}, 
defined as the largest expected number of steps of a random walk starting at a harmful node until it visits a benign node. 
However, while \citeauthor{fabbri2022rewiring} \cite{fabbri2022rewiring} 
posed a theoretically interesting and practically important problem, 
their approach has some crucial limitations. 

First, treating harmfulness as dichotomous 
fails to capture the complexity of real-world harmfulness assessments.
Second, 
the segregation objective
ignores completely all random-walk continuations that return to harmful content after the first visit to a benign node, 
but \emph{benign nodes do not act as absorbing states} in practice.
The consequences are illustrated in \cref{fig:setting:bad}, 
where the segregation objective judges that the graph
provides minimal exposure to harmful content 
(the hitting time from any harmful node to a benign node is~1), 
while long random walks, 
which model user behavior more realistically, 
oscillate between harmful and benign~content.

In this paper, 
we remedy the above-mentioned limitations.
First, we more nuancedly model harmfulness as \emph{real-valued} node costs.
Second, we propose a novel minimization objective, 
the \emph{expected} \emph{total} \emph{exposure}, 
defined as the sum of the costs of absorbing random walks starting at any node.
Notably, in our model, 
no node is an absorbing state, 
but any node can lead to absorption, 
which represents more faithfully how users cease to interact with a platform.
Our exposure objective truly minimizes the exposure to harmful content. 
For example, it correctly 
identifies the graph in \cref{fig:setting:good} as significantly less harmful than that in \cref{fig:setting:bad}, 
while for the segregation objective by \citeauthor{fabbri2022rewiring} \cite{fabbri2022rewiring},
the two graphs are indistinguishable. 

On the algorithmic side, 
we show that although minimizing the expected total exposure is \NP-hard and \NP-hard to approximate to within an additive error, 
its maximization version is equivalent to a submodular maximization problem 
under the assumption that the input graph contains a small number of \emph{safe} nodes, 
i.e., nodes that cannot reach nodes with non-zero costs.
If these safe nodes are present---%
which holds in $80\%$ of the real-world graphs used in our experiments---%
the greedy method yields a $(1-\nicefrac{1}{e})$-approximation.  
Based on our theoretical insights, 
we introduce \ourmethod, 
a fast greedy algorithm for reducing exposure to harmful content via edge rewiring. 
\ourmethod leverages provable strategies for pruning unpromising rewiring candidates,
and it works both with and without quality constraints on recommendations.
With just $100$ rewirings on YouTube graphs containing hundred thousands of edges, 
\ourmethod reduces the exposure by $50\%$, 
while ensuring that its recommendations are at least $95\%$ as relevant as the originals.

In the following, 
we introduce our problems, \ourproblem and \ourproblemtwo (\cref{sec:problems}), 
and analyze them theoretically (\cref{sec:theory}).
Building on our theoretical insights, 
we develop \ourmethod as an efficient greedy algorithm for tackling our problems (\cref{sec:algorithm}). 
Having discussed related work (\cref{sec:related}),
we demonstrate the performance of \ourmethod through extensive experiments (\cref{sec:experiments}) 
before concluding with a discussion (\cref{sec:conclusion}).
All code, datasets, and results are publicly available,\!\footnote{\oururl}
and we provide further materials in \cref{apx:datasets,apx:edits,apx:experiments,apx:hardness,apx:pseudocode,apx:reproducibility}.

\section{Problems}
\label{sec:problems}

We consider a directed graph $\graph = (\nodes,\edges)$ of content items ($\nodes$) and what-to-consume-next recommendations ($\edges$),
with $\nnodes = \cardinality{\nodes}$ nodes and $\nedges = \cardinality{\edges}$ edges. 
Since we can typically make a fixed number of recommendations for a given content item, 
such \emph{recommendation graphs} are often $\outregulardegree$-out-regular, 
i.e., all nodes have $\outregulardegree = \nicefrac{\nedges}{\nnodes}$ out-neighbors, 
but we do not restrict ourselves to this setting.
Rather, each node $i$ has an out-degree $\outdegree{i} = \cardinality{\outneighbors{i}}$, 
where~$\outneighbors{i}$ is the set of out-neighbors of $i$, 
and a cost $\labeling_i\in\range$,
which quantifies the harmfulness of content item $i$, 
ranging from $0$ (not harmful at all) to $1$ (maximally harmful).
For convenience, we define $\maxoutdegree = \max\{\outdegree{i}\mid i \in \nodes\}$
and collect all costs into a vector $\labelingvector\in \range^{\nnodes}$. 
We model user behavior as a random-walk process on the recommendation graph~$\graph$. 
Each edge $\edge{i}{j}$ in the recommendation graph
is associated with a transition probability $\probability{ij}$ 
such that $\sum_{j\in\outneighbors{i}}\probability{ij} = 1 - \pabsorption_i$, 
where $\pabsorption_i$ is the absorption probability of a random walk at node $i$ 
(i.e., the probability that the walk ends~at~$i$). Intuitively, one can interpret $\alpha_i$ as the probability that a user stops using the service after consuming content~$i$.
For simplicity, we assume $\pabsorption_i = \pabsorption \in (0,1]$ for all $i\in\nodes$.
Thus, we can represent the random-walk process on $\graph$ 
by the transition matrix $\transitionmatrix\in[0,1-\pabsorption]^{\nnodes\times\nnodes}$, 
where 
\begin{align}
	\transitionmatrix[i,j] =\begin{cases}
		 \probability{ij}& \text{if }\edge{i}{j}\in \edges\;,\\
		 0&\text{otherwise}\;.
	\end{cases}
\end{align}
This is an absorbing Markov chain, and the expected number of visits from a node $i$ 
to a node $j$ before absorption is given by the entry $(i,j)$ 
of the \emph{fundamental matrix} 
$\fundamental\in\reals_{\geq 0}^{\nnodes\times\nnodes}$, defined as
\begin{align}\label{eq:fundamental}
	\fundamental = \sum_{i=0}^{\infty} \transitionmatrix^i = (\identity - \transitionmatrix)^{-1}\;,
\end{align}
where $\identity$ is the $\nnodes\times\nnodes$-dimensional identity matrix, 
and the series converges since $\norm{\transitionmatrix}_{\infty} = \max_i\sum_{j=0}^{\nnodes}\transitionmatrix[i,j] = 1 - \pabsorption < 1$.
Denoting the $i$-th unit vector as~$\unitvector_i$, 
observe that the row vector $\unitvector_i^T\fundamental$ gives the expected number of visits, before absorption, from $i$ to any node, 
and the column vector $\fundamental\unitvector_i$ gives the expected number of visits from any node to $i$.
Hence, $\unitvector_i^T\fundamental\labelingvector = \sum_{j\in\nodes}\fundamental[i,j]\labelingvector_j$ gives the expected exposure to harmful content of users starting their random walk at node $i$, 
referred to as the \emph{exposure} of $i$. 
The \emph{expected total exposure} to harm in the graph $\graph$, then, 
is given by the non-negative function
\begin{align}
	\objective{\graph} = \onevector^T\fundamental\labelingvector\;,
	\label{eq:harm}
\end{align}
where $\onevector$ is the vector with each entry equal to~$1$.

We would like to \emph{minimize} the exposure function given in \cref{eq:harm} 
by making $\budget$ edits to the graph $\graph$, 
i.e., we seek an effective \emph{post-processing strategy} for harm reduction. 
In line with our motivating application, 
we restrict edits to \emph{edge rewirings} denoted as $\rewiring{i}{j}{k}$, 
in which we replace an edge $\edge{i}{j}\in\edges$ by an edge $\edge{i}{k}\notin\edges$ with $i\neq k$, 
setting $\probability{ik} = \probability{ij}$ 
(other edits are discussed in \cref{apx:edits}). 
Seeking edge rewirings to minimize the expected total exposure yields the following problem definition.
\begin{problem}[$\budget$-rewiring exposure minimization {[}\ourproblem{]}]
	\label{problem}
	Given a graph $\graph$, 
	its random-walk transition matrix $\transitionmatrix$, 
	a node cost vector $\labelingvector$, 
	and a budget $\budget$, 
	minimize $\objective{\graph_\budget}$, 
	where $\graph_\budget$ is $\graph$ after $\budget$ rewirings.
\end{problem}
Equivalently, we can \emph{maximize} the \emph{reduction} in the expected total exposure to harmful content,
\begin{align}
	\maxobjective{\graph,\graph_\budget} =\objective{\graph}-\objective{\graph_\budget}\;.
	\label{eq:maxdelta}
\end{align}
Note that while any set of rewirings minimizing $\objective{\graph_\budget}$ also maximizes $\maxobjective{\graph,\graph_\budget}$, 
the approximabilities of $\objectivef$ and $\maxobjectivef$ can differ widely.

As \cref{problem} does not impose any constraints on the rewiring operations, 
the optimal solution might contain rewirings $\rewiring{i}{j}{k}$ such that node $k$ is unrelated to $i$. 
To guarantee high-quality recommendations, 
we need additional relevance information, 
which we assume to be given as a \emph{relevance matrix} 
$\relevancematrix\in\reals_{\geq 0}^{\nnodes\times\nnodes}$, 
where $\relevancematrix[i,j]$ denotes the relevance of node~$j$ in the context of node~$i$. 
Given such relevance information, 
and assuming that the out-neighbors of a node $i$ are ordered as
$\outseq{i}\in \nodes^{\outdegree{i}}$, 
we can define a \emph{relevance function}~$\qualityf$ with range $\range$
to judge the quality of the recommendation sequence at node~$i$,
depending on the relevance and ordering of recommended nodes, 
and demand that any rewiring retain $\quality{\outseq{i}} \geq \qualitythreshold$ for all $i\in\nodes$ and some \emph{quality threshold} $\qualitythreshold\in\range$. 
One potential choice for $\qualityf$
is the normalized discounted cumulative gain ($\ndcg$), 
a popular ranking quality measure,
which we use in our experiments and define in \cref{apx:reproducibility:qualityf}.
Introducing $\qualityf$ allows us to consider a variant of \ourproblem with relevance constraints.

\begin{problem}[$\qualitythreshold$-relevant $\budget$-rewiring exposure minimization {[}\ourproblemtwo{]}]\label{problem2}
	Given a graph $\graph$, 
	its random-walk transition matrix $\transitionmatrix$, 
	a node cost vector $\labelingvector$, 
	a budget $\budget$, 
	a relevance matrix $\relevancematrix$,
	a relevance function $\qualityf$,
	and a quality threshold $\qualitythreshold$, 
	minimize $\objective{\graph_\budget}$
	under the condition that $\quality{\outseq{i}} \geq \qualitythreshold$ for all $i\in\nodes$.
\end{problem}

For $\qualitythreshold = 0$, \ourproblemtwo is equivalent to \ourproblem. 
Collecting our notation in Appendix~\cref{tab:notation}, 
we now seek to address both problems.

\section{Theory}
\label{sec:theory}

To start with, we establish some theoretical properties of our problems, the functions $\objectivef$ and $\maxobjectivef$, 
and potential solution approaches.

\paragraph{Hardness}
We begin by proving that \ourproblem (and hence, also \ourproblemtwo) is an \NP-hard problem.
\begin{restatable}[\NP-Hardness of \ourproblem]{thm}{hardness}\label{thm:hardness}
	The $\budget$-rewiring exposure minimization problem is \NP-hard, even on 3-out-regular input graphs with binary costs $\labelingvector\in\{0,1\}^{\nnodes}$.
\end{restatable}
\begin{proof}
	We obtain this result by reduction from minimum vertex cover for cubic, 
	i.e., 3-regular graphs (MVC-3), 
	which is known to be \NP-hard \cite{greenlaw1995cubic}.
	A full, illustrated proof is given in \cref{apx:hardness:np}.
\end{proof}

Next, we further show that \ourproblem is hard to approximate under the Unique Games Conjecture (UGC) \cite{khot2002power}, 
an influential conjecture in hardness-of-approximation theory.
\begin{restatable}{thm}{apxhardness}\label{thm:apxhardness}
	Assuming the UGC, 
	\ourproblem is hard to approximate to within an additive error of both $\Theta(\nnodes)$ and $\Theta(\budget)$.
\end{restatable}
\begin{proof}
	We obtain this result via the hardness of approximation of MVC 
	under the UGC. 
	A full proof is given in \cref{apx:hardness:apx}.
\end{proof}
Both \cref{thm:hardness} and \cref{thm:apxhardness} extend from $\objectivef$ to \maxobjectivef (\cref{eq:delta}). 

\paragraph{Approximability}
Although we cannot approximate $\objectivef$ directly,  
we \emph{can} approximate $\maxobjectivef$ \emph{with guarantees} 
under mild assumptions, detailed below. 
To formulate this result and its assumptions, 
we start by calling a node \emph{safe} 
if  $\unitvector_i^T\fundamental\labelingvector = 0$, 
i.e., no node $j$ with $\labeling_j > 0$ is reachable from $i$, 
and \emph{unsafe} otherwise. 
Note that the existence of a safe node in a graph $\graph$ containing at least one unsafe node (i.e.,~$\labeling_i > 0$ for some $i\in \nodes$) 
implies that $\graph$ is not strongly connected. 
The node safety property partitions $\nodes$ into two sets of safe resp. unsafe nodes, 
$\safenodes = \{i\in\nodes\mid \unitvector_i^T\fundamental\labelingvector = 0\}$~and~$\unsafenodes = \{i\in\nodes\mid \unitvector_i^T\fundamental\labelingvector > 0\}$, 
and $\edges$ into four sets, 
$\edges_{\safenodes\safenodes}$, $\edges_{\safenodes\unsafenodes}$, $\edges_{\unsafenodes\safenodes}$, and $\edges_{\unsafenodes\unsafenodes}$, 
where $\edges_{AB} = \{(i,j)\in\edges\mid i\in A, j\in B\}$, 
and $\edges_{\safenodes\unsafenodes}=\emptyset$ by construction. 
Further, observe that if $\safenodes\neq\emptyset$,  
then $\objectivef$ is minimized, 
and $\maxobjectivef$ is maximized, 
once $\edges_{\unsafenodes\unsafenodes} = \emptyset$.
This allows us to state the following result.
\begin{restatable}{lem}{submodularitywithduplicates}\label{lem:submoddupes}
	If there exists a safe node in $\graph$ and we allow multi-edges,  
	maximizing $\maxobjectivef$ is equivalent to maximizing a monotone, submodular set function over $\edges_{\unsafenodes\unsafenodes}$.
\end{restatable}
\begin{proof}
	Leveraging the terminology introduced above, we obtain this result by applying the definitions of monotonicity and submodularity. 
	A full proof is given in \cref{apx:hardness:submod}.
\end{proof}
Our motivating application, however, ideally prevents multi-edges.
To get a similar result without multi-edges, 
denote by $\maxunsafeoutdegree = \max\{\outdegree{i}\mid i \in \unsafenodes\}$
the maximum out-degree of any \emph{unsafe} node in~$\graph$, 
and assume that $\cardinality{\safenodes} \geq \maxunsafeoutdegree$. 
Now, we obtain the following.

\begin{restatable}{thm}{submodularity}\label{thm:submodularity}
	If $\cardinality{\safenodes} \geq \maxunsafeoutdegree$,
	then maximizing $\maxobjectivef$ is equivalent to maximizing a monotone and submodular set function over $\edges_{\unsafenodes\unsafenodes}$.
\end{restatable}
\begin{proof}
	Following the reasoning provided for \cref{lem:submoddupes}, 
	with the modification that we need $\cardinality{\safenodes} \geq \maxunsafeoutdegree$ to ensure that safe targets are always available for rewiring without creating multi-edges.
\end{proof}

Observe that the larger the number of zero-cost nodes, 
the smaller the number of edges, 
or the more homophilous the linking, 
the higher the probability that safe nodes exist in a graph. 
Notably, the precondition of \cref{thm:submodularity} holds for the graph constructed to prove \cref{thm:hardness} (\cref{apx:hardness:np}, \cref{fig:hardness}) 
as well as for most of the real-world graphs used in our experiments (\cref{apx:datasets}, \cref{fig:safety}).
However,
\cref{thm:submodularity} only applies to the maximization version of \ourproblem (\cref{eq:delta}) and not to the maximization version of \ourproblemtwo, 
since in the quality-constrained setting, 
some safe nodes might not be available as rewiring targets for edges emanating from unsafe nodes. 
Still, for the maximization version of \ourproblem,
due to \cref{thm:submodularity}, 
using a greedy approach to optimize $\maxobjectivef$ 
provides an approximation guarantee 
with respect to the optimal solution~\cite{nemhauser1978analysis}.
\begin{cor}\label{cor:apxguarantee}
	If the precondition of \cref{thm:submodularity} holds, 
	then the greedy algorithm, 
	which always picks the rewiring $\rewiring{i}{j}{k}$ that maximizes $\maxobjective{\graph,\graph_1}$ for the current $\graph$,
	yields a $(1-\nicefrac{1}{e})$-approximation for $\maxobjectivef$.
\end{cor}
Note that \cref{cor:apxguarantee} only applies to the \emph{maximization} version of \ourproblem, 
not to its \emph{minimization} version, 
as supermodular minimization is less well-behaved than submodular maximization \cite{zhang2022fast,ilev2001approximation}.

\paragraph{Greedy Rewiring}
Given the quality assurance of a greedy approach at least for \ourproblem, 
we seek to design an efficient greedy algorithm to tackle both \ourproblem and \ourproblemtwo.
To this end, we analyze the mechanics of individual rewirings to understand how we can identify and perform greedily optimal rewirings efficiently.
As each greedy step constitutes a rank-one update of the 
transition matrix~$\transitionmatrix$, 
we can express the new transition matrix $\transitionmatrix'$ as 
\begin{align}
	\transitionmatrix' = \transitionmatrix + \ivector(-\jkvector)^T \;,
\end{align}
where $\ivector = \probability{ij}\unitvector_i$ and $\jkvector = \unitvector_j - \unitvector_k$, 
and we omit the dependence on $i$, $j$, and $k$ for notational conciseness.
This corresponds to a rank-one update of $\fundamental$, such that we obtain the new fundamental matrix $\fundamental'$ as
\begin{align}
	\fundamental' 
	= (\identity - (\transitionmatrix + \ivector(-\jkvector)^T))^{-1}
	= (\identity - \transitionmatrix + \ivector\jkvector^T)^{-1} \;.
\end{align}

The rank-one update allows us to use the 
Sherman-Morrison formula \cite{sherman1950adjustment} to compute 
the updated fundamental matrix as
\begin{align}\label{eq:sherman}
	\fundamental'  = 
		\fundamental - 
		\frac{\fundamental \ivector \jkvector^T\fundamental}
			 {1 + \jkvector^T\fundamental\ivector} \;.
\end{align}

The mechanics of an individual edge rewiring are summarized in \cref{tab:rewiring}.
They will help us \emph{perform} greedy updates efficiently.

To also \emph{identify} greedily optimal rewirings efficiently, 
leveraging \cref{eq:sherman},
we assess the impact 
of a rewiring on the value of our objective function, 
which will help us prune weak rewiring candidates. 
For a rewiring $\rewiring{i}{j}{k}$ represented by $\ivector$ and $\jkvector$, 
the value of the exposure function~$\objectivef$
for the new graph $\graph'$ is
\begin{align}
	\objective{\graph'} 
	&	= \onevector^T\fundamental'\labelingvector
		= \onevector^T\left(\fundamental - \frac{\fundamental \ivector \jkvector^T\fundamental}{1 + \jkvector^T\fundamental\ivector}\right)\labelingvector
		= \onevector^T\fundamental\labelingvector - \onevector^T\left(\frac{\fundamental \ivector \jkvector^T\fundamental}{1 + \jkvector^T\fundamental\ivector}\right)\labelingvector \nonumber \\
	&	= \objective{\graph} -  \frac{(\onevector^T\fundamental \ivector)(\jkvector^T\fundamental\labelingvector)}{1 + \jkvector^T\fundamental\ivector}
		= \objective{\graph} - \frac{\sigma\tau}{\rho}
		= \objective{\graph} - \Delta \;,  
	\label{eq:fnew=fold-delta}
\end{align}
with 
$\sigma = \onevector^T\fundamental \ivector$, 
$\tau = \jkvector^T\fundamental\labelingvector$,
$\rho = 1 + \jkvector^T\fundamental\ivector$,
and 
\begin{align}\label{eq:delta}
	\Delta = \maxobjective{\graph,\graph'} = \frac{\sigma\tau}{\rho} = \frac{(\onevector^T\fundamental \ivector)(\jkvector^T\fundamental\labelingvector)}{1 + \jkvector^T\fundamental\ivector} \;.
\end{align}

\begin{table}[t]
	\centering
	\caption{%
		Summary of an edge rewiring $\rewiring{i}{j}{k}$
		in a~graph $\graph = (\nodes,\edges)$
		with random-walk transition matrix $\transitionmatrix$
		and fundamental matrix $\fundamental = (\identity - \transitionmatrix)^{-1}$.
	}
	\label{tab:rewiring}
	\begin{tabular}{p{0.9625\columnwidth}}
		\toprule
		$\graph'=(\nodes,\edges')$, for $E' = (E \setminus \{\edge{i}{j}\}) \cup \{\edge{i}{k}\}$, $\edge{i}{j}\in E$, $\edge{i}{k}\notin E$\\
		\midrule
		$\transitionmatrix'[x,y] = \begin{cases}
			0&\text{if } x = i \text{ and } y = j\;,\\
			\transitionmatrix[i,j]&\text{if } x = i \text{ and } y = k\;,\\
			\transitionmatrix[x,y]&\text{otherwise}\;.
		\end{cases}$\\
		\midrule
		$\fundamental'
		= \fundamental - \frac{\fundamental\ivector\jkvector^T\fundamental}{1 + \jkvector^T\fundamental\ivector}$, with $\ivector = \probability{ij} \unitvector_i$, $\jkvector = \unitvector_j-\unitvector_k$, cf.~\cref{eq:sherman}\\
		\bottomrule
	\end{tabular}
\end{table}

The interpretation of the above quantities is as follows:
$\sigma$~is the $\probability{ij}$-scaled $i$-th column sum of $\fundamental$ (expected number of visits to $i$),
$\tau$~is the cost-scaled sum of the differences between the $j$-th row and the $k$-th row of $\fundamental$ (expected number of visits from $j$ resp. $k$), 
and $\rho$~is a normalization factor scaling the update by $1$ plus the $\probability{ij}$-scaled difference in the expected number of visits from $j$ to $i$ and from $k$ to $i$, 
ensuring that $\fundamental'\onevector = \fundamental\onevector$.
Scrutinizing \cref{eq:delta}, we observe:
\begin{restatable}{lem}{deltacomponents}\label{lem:deltacomponents}
	\label{lemma:signs}
	For a rewiring $(i,j,k)$ represented by $\ivector$ and $\jkvector$,
	\begin{enumerate*}[label=(\roman*)]
		\item $\rho$ is always positive, 
		\item $\sigma$ is always positive, and
		\item $\tau$ can have any sign.
	\end{enumerate*}
\end{restatable}
\begin{proof}
	We obtain this result by analyzing the definitions of $\rho$, $\sigma$, and $\tau$. 
	The full proof is given in \cref{apx:hardness:deltaanalysis}.
\end{proof}

To express when we can safely prune rewiring candidates,
we call a rewiring $\rewiring{i}{j}{k}$ \emph{greedily permissible} if $\Delta > 0$, 
i.e., if it reduces our objective, 
and \emph{greedily optimal} if it maximizes $\Delta$. 
For \ourproblemtwo, we further call a rewiring $\rewiring{i}{j}{k}$ \emph{greedily $\qualitythreshold$-permissible} if it ensures that $\quality{\outseq{i}} \geq \qualitythreshold$ under the given relevance function $\qualityf$.
With this terminology, 
we can confirm our intuition about rewirings 
as a corollary of \cref{eq:fnew=fold-delta,eq:delta}, 
combined with \cref{lemma:signs}.
\begin{cor}
	A rewiring $\rewiring{i}{j}{k}$ is greedily permissible if and only if $\tau > 0$, i.e., 
	if $j$ is more exposed to harm than $k$.
\end{cor}
 
For the greedily optimal rewiring, 
that is, to maximize $\Delta$, 
we would like $\sigma\tau$ to be as \emph{large} as possible, 
and $\rho$ to be as \emph{small} as possible.
Inspecting \cref{eq:delta},
we find that to accomplish this objective, it helps if (in expectation)
$i$~is visited more often (from~$\sigma$),
$j$~is more exposed and $k$ is less exposed to harm (from~$\tau$),
and $i$~is harder to reach from $j$ and easier to reach from $k$ (from~$\rho$).

In the next section, we leverage these insights
to guide our efficient implementation of the greedy method
for \ourproblem and \ourproblemtwo.

\section{Algorithm}
\label{sec:algorithm}

In the previous section, we identified useful structure
in the fundamental matrix $\fundamental$, the exposure function $\objectivef$, and our maximization objective \maxobjectivef.
Now, we leverage this structure to design an efficient greedy algorithm
for \ourproblem and \ourproblemtwo.
We develop this algorithm in three steps, 
focusing on \ourproblem in the first two steps, 
and integrating the capability to handle \ourproblemtwo in the third step.

\paragraph{Na\"ive implementation}

Given a graph $\graph$, its transition matrix~$\transitionmatrix$, a cost vector $\labelingvector$,
and a budget $\budget$,
a na\"ive greedy implementation for \ourproblem 
computes the fundamental matrix 
and gradually fills up an initially empty set of rewirings 
by performing $\budget$ greedy steps 
before returning the selected rewirings (\cref{apx:pseudocode}, \cref{alg:naivegreedy}). 
In each greedy step, 
we identify the triple $\rewiring{i}{j}{k}$ that maximizes \cref{eq:delta} 
by going through all edges $\edge{i}{j}\in\edges$ 
and computing $\Delta$ for rewirings to all potential targets $k$. 
We then update $\edges$, $\transitionmatrix$, and $\fundamental$ to reflect a rewiring replacing $\edge{i}{j}$ by $\edge{i}{k}$ (cf.~\cref{tab:rewiring}), 
and add the triple $\rewiring{i}{j}{k}$ to our set of rewirings.
Computing the fundamental matrix na\"ively takes time $\bigoh{\nnodes^3}$,
computing $\Delta$ takes time $\bigoh{\nnodes}$ and is done $\bigoh{\nedges\nnodes}$ times, and 
updating $\fundamental$ takes time $\bigoh{\nnodes^2}$.
Hence, we arrive at a time complexity of $\bigoh{\budget\nnodes^2(\nnodes+\nedges)}$.
But we can do better.

\paragraph{Forgoing matrix inversion}

When identifying the greedy rewiring, 
we never need access to $\fundamental$ directly.
Rather, in \cref{eq:delta}, we work with $\onevector^T\fundamental$, 
corresponding to the column sums of $\fundamental$, 
and with $\fundamental\labelingvector$, 
corresponding to the cost-scaled row sums of $\fundamental$.
We can approximate both via power iteration:
\begin{align}
	\onevector^T\fundamental 
	=~& \onevector^T\sum_{i=0}^{\infty}\transitionmatrix^i 
	= \onevector^T 
	+ \onevector^T\transitionmatrix 
	+  (\onevector^T\transitionmatrix)\transitionmatrix 
	+ ((\onevector^T\transitionmatrix)\transitionmatrix)\transitionmatrix + \dots\\
	\fundamental\labelingvector 
	=~& \left(\sum_{i=0}^{\infty}\transitionmatrix^i\right)\labelingvector
	= \labelingvector + \transitionmatrix\labelingvector + \transitionmatrix(\transitionmatrix\labelingvector) + \transitionmatrix(\transitionmatrix(\transitionmatrix\labelingvector)) + \dots
\end{align}
For each term in these sums, we need to perform $\bigoh{\nedges}$ multiplications, 
such that we can compute $\onevector^T\fundamental$ and $\fundamental\labelingvector$ in time $\bigoh{\niter\nedges}$, 
where $\niter$ is the number of power iterations.
This allows us to compute 
$\onevector^T\fundamental\ivector$ for all $\edge{i}{j}\in\edges$ in time $\bigoh{\nedges}$ and $\jkvector^T\fundamental\labelingvector$ for all $j\neq k\in\nodes$ in time $\bigoh{\nnodes^2}$.
To compute $\Delta$ in time $\bigoh{1}$,
as $\fundamental$ is now unknown, 
we need to compute $\fundamental\ivector$ for all $(i,j)\in\edges$ via power iteration, 
which is doable in time $\bigoh{\niter\nnodes^2}$.
This changes the running time from $\bigoh{\budget\nnodes^2(\nnodes+\nedges)}$ to $\bigoh{\budget\niter\nnodes(\nnodes+\nedges)}$ (\cref{apx:pseudocode}, \cref{alg:greedy}).
But we can do better.

\paragraph{Reducing the number of candidate rewirings}
Observe that to further improve the time complexity of our algorithm,
we need to reduce the number of rewiring candidates considered. 
To this end, note that the quantity $\tau$ is maximized for the nodes $j$ and $k$ with the largest difference in cost-scaled row sums.
How exactly we leverage this fact depends on our problem.

If we solve \ourproblem, 
instead of considering all possible rewiring targets, 
we focus on the $\maxoutdegree+2$ candidate targets $\kcandidates$ with the smallest exposure, 
which we can identify in time $\bigoh{\nnodes}$ without sorting $\fundamental\labelingvector$.
This ensures that for each $\edge{i}{j}\in\edges$, 
there is at least one $k\in\kcandidates$ such that $k\neq i$ and $k\neq j$, 
which ascertains that despite restricting to $\kcandidates$, 
for each $i\in\nodes$, 
we still consider the rewiring $\rewiring{i}{j}{k}$ maximizing $\tau$.
With this modification, we reduce the number of candidate targets from $\bigoh{\nnodes}$ to $\bigoh{\maxoutdegree}$
and the time to compute all relevant $\jkvector^T\fundamental\labelingvector$ values from $\bigoh{\nnodes^2}$ to $\bigoh{\maxoutdegree\nnodes}$.
To obtain a subquadratic complexity, however, we still need to eliminate the computation of $\fundamental\ivector$ for all $(i,j)\in\edges$.
This also means that we can no longer afford to compute~$\rho$ for each of the now $\bigoh{\nedges\maxoutdegree}$ rewiring candidates under consideration, 
as this can only be done in constant time if $\fundamental\ivector$ is already precomputed for the relevant edge $\edge{i}{j}$.
However, $\rho$ is driven by the difference between two \emph{entries} of $\fundamental$, 
whereas $\tau$ is driven by the difference between two \emph{row sums} of $\fundamental$, 
and $\sigma$ is driven by a single \emph{column sum} of $\fundamental$. 
Thus, although $\sigma\tau > \sigma\tau'$ does not generally imply $\nicefrac{\sigma\tau}{\rho} > \nicefrac{\sigma\tau'}{\rho'}$,
the variation in $\sigma\tau$ is typically much larger than that in $\rho$, 
and large $\sigma\tau$ values mostly dominate small values of $\rho$.
Consequently, as demonstrated in \cref{apx:exp:heuristic}, 
the correlation between $\heuristic = \Delta\rho = \sigma\tau$ and  $\Delta = \nicefrac{\sigma\tau}{\rho}$ is almost perfect.
Thus, instead of $\Delta$, we opt to compute $\heuristic$ as a heuristic, 
and we further hedge against small fluctuations without increasing the time complexity of our algorithm by computing $\Delta$ 
for the rewirings associated with the $\bigoh{1}$ \emph{largest} values of $\heuristic$,
rather than selecting the rewiring with the \emph{best} $\heuristic$ value directly. 
Using $\heuristic$ instead of $\Delta$, 
we obtain a running time of $\bigoh{\budget\niter\maxoutdegree(\nnodes+\nedges)}$ when solving \ourproblem.

When solving \ourproblemtwo, 
we are given a relevance matrix $\relevancematrix$, 
a relevance function $\qualityf$,
and a relevance threshold $\qualitythreshold$ as additional inputs. 
Instead of considering the $\maxoutdegree+2$ nodes $\kcandidates$ with the smallest exposure as candidate targets for \emph{all} edges, 
for \emph{each} edge $\edge{i}{j}$, 
we first identify the set of rewiring candidates $\rewiring{i}{j}{k}$ 
such that $\rewiring{i}{j}{k}$ is $\qualitythreshold$-permissible, 
i.e., $\quality{\outseq{i}} \geq \qualitythreshold$ after replacing $\edge{i}{j}$ by $\edge{i}{k}$,
and then select the node $k_{ij}$ with the smallest exposure to construct our most promising rewiring candidate $\rewiring{i}{j}{k_{ij}}$ for edge $\edge{i}{j}$.
This ensures that we can still identify the rewiring $\rewiring{i}{j}{k}$ that maximizes $\sigma\tau$ \emph{and} satisfies our quality constraints,
and it leaves us to consider $\bigoh{\nedges}$ rewiring candidates. 
Again using $\heuristic$ instead of  $\Delta$, 
we can now solve \ourproblemtwo in time 
$\bigoh{\budget\niter\nqpermissible\qualitycomplexity\nedges+\qualitycomplexitystart}$, 
where 
$\nqpermissible$ is the maximum number of targets $k$ such that $\rewiring{i}{j}{k}$ is $\qualitythreshold$-permissible,
$\qualitycomplexity$ is the complexity of evaluating $\qualityf$, 
and $\qualitycomplexitystart$ is the complexity of determining the initial set $\qualityset$ of $\qualitythreshold$-permissible rewirings.

\smallskip
Thus, we have arrived at our efficient greedy algorithm, 
called \ourmethod (\textsc{G}reedy \textsc{a}pproximate \textsc{min}imi\-zation of \textsc{e}xposure),
whose pseudocode we state as 
\cref{alg:greedyrelevant:rem,alg:greedyrelevant} in \cref{apx:pseudocode}.
\ourmethod solves \ourproblem in time $\bigoh{\budget\niter\maxoutdegree(\nnodes+\nedges)}$
and \ourproblemtwo in time $\bigoh{\budget\niter\nqpermissible\qualitycomplexity\nedges+\qualitycomplexitystart}$.
In realistic recommendation settings, 
the graph $\graph$ is $\outregulardegree$-out-regular for $\outregulardegree\in\bigoh{1}$, 
such that
$\maxoutdegree\in\bigoh{1}$ and $\nedges=\outregulardegree\nnodes\in\bigoh{\nnodes}$.
Further, for \ourproblemtwo,
we can expect that $\qualityf$ is evaluable in time $\bigoh{1}$,
and that only the $\bigoh{1}$ nodes most relevant for $i$ will be considered as potential rewiring targets of any edge $\edge{i}{j}$, 
such that $\nqpermissible\in\bigoh{1}$ and $\qualitycomplexitystart \in \bigoh{\nedges} = \bigoh{\nnodes}$. 
As we can also safely work with a number of power iterations $\niter\in\bigoh{1}$ (\cref{apx:reproducibility:poweriteration}), 
in realistic settings, 
\ourmethod solves both \ourproblem and \ourproblemtwo in time $\bigoh{\budget\nnodes}$, 
which, for $\budget\in\bigoh{1}$, is linear in the order of the input graph~$\graph$.

\section{Related Work}
\label{sec:related}

Our work methodically relates to research on \emph{graph edits} with distinct goals, 
such as improving robustness, reducing distances, or increasing centralities \cite{chan2016optimizing,parotsidis2015selecting,medya2018group},
and research leveraging \emph{random walks} to rank nodes \cite{wkas2019random,oettershagen2022temporal,mavroforakis2015absorbing}
or recommend links \cite{yin2010unified,paudel2021random}.
The agenda of our work, however, 
aligns most closely with the literature studying harm reduction, bias mitigation, and conflict prevention in graphs.
Here, the large body of research on shaping opinions or mitigating negative phenomena
in \emph{graphs of user interactions} (especially on social media)
\cite{gionis2013opinion,das2014modeling,abebe2021opinion,amelkin2019fighting,garimella2017reducing,garimella2018quantifying,tsioutsiouliklis2022link,zhu2021minimizing,zhu2022nearly,vendeville2023opening,minici2022cascade}
pursues goals \emph{similar} to ours in graphs capturing \emph{different} digital contexts.

As our research is motivated by recent work demonstrating how recommendations on digital media platforms like YouTube can fuel radicalization \cite{ribeiro2020auditing,mamie2021anti}, 
the comparatively scarce literature on harm reduction in \emph{graphs of content items} is even more closely related.
Our contribution is inspired by \citeauthor{fabbri2022rewiring} \cite{fabbri2022rewiring}, 
who study how edge rewiring can reduce \emph{radicalization pathways} in recommendation graphs. 
\citeauthor{fabbri2022rewiring} \cite{fabbri2022rewiring} encode harmfulness in binary node labels, 
model benign nodes as absorbing states, 
and aim to minimize the \emph{maximum} \emph{segregation} of any node, 
defined as the largest expected length of a random walk starting at a harmful node before it visits a benign node.
In contrast, we encode harmfulness in more nuanced, real-valued node attributes,
use an absorbing Markov chain model that more naturally reflects user behavior, 
and aim to minimize the \emph{expected} \emph{total} \emph{exposure} to harm in random walks  starting at any node. 
Thus, our work not only eliminates several limitations of the work by \citeauthor{fabbri2022rewiring} \cite{fabbri2022rewiring}, 
but it also provides a different perspective on harm mitigation in recommendation graphs. 

While \citeauthor{fabbri2022rewiring} \cite{fabbri2022rewiring}, 
like us, 
consider \emph{recommendation graphs}, 
\citeauthor{haddadan2022reducing} \cite{haddadan2022reducing}
focus on polarization mitigation via \emph{edge insertions}. 
Their setting was recently reconsidered by \citeauthor{adriaens2023minimizing}~\cite{adriaens2023minimizing}, 
who tackle the minimization objective directly instead of using the maximization objective as a proxy,
providing approximation bounds as well as speed-ups for the standard greedy method. 
Both \citeauthor{fabbri2022rewiring} \cite{fabbri2022rewiring} and the works on edge insertion employ with random-walk objectives 
that---%
unlike our exposure function---%
do not depend on random walks starting from \emph{all} nodes.
In our experiments, 
we compare with the algorithm introduced by \citeauthor{fabbri2022rewiring} \cite{fabbri2022rewiring}, which we call \fabbrialg. 
We refrain from comparing with edge insertion strategies because they consider a different graph edit operation and are already outperformed by \fabbrialg.

\section{Experimental Evaluation}
\label{sec:experiments}

In our experiments, we seek to 
\begin{enumerate}
	\item establish the impact of modeling choices and input parameters
	on the performance of \ourmethod;
	\item demonstrate the effectiveness of \ourmethod in reducing exposure to harm compared to existing methods and baselines;
	\item ensure that \ourmethod is scalable in theory \emph{and practice}; 
	\item understand what features make reducing exposure to harm easier resp. harder on different datasets; and
	\item derive general guidelines for reducing exposure to harm in recommendation graphs under budget constraints.
\end{enumerate}
Further experimental results are provided in \cref{apx:experiments}.

\subsection{Setup}

\subsubsection{Datasets}
To achieve our experimental goals, 
we work with both synthetic and real-world data, 
as summarized in \cref{tab:data}. 
Below, we briefly introduce these datasets.
Further details, 
including on data generation and preprocessing, are provided in \cref{apx:datasets}.

\begin{table}[t]
	\centering
	\caption{%
		Overview of the datasets used in our experiments. 
		For each graph $\graph$, 
		we report the regular out-degree $\outregulardegree$, 
		the number of nodes $\nnodes$, and the number of edges $\nedges$, 
		as well as the range of the expected exposure $\nicefrac{\objective{\graph}}{\nnodes}$ under our various cost functions, edge wirings, and edge transition probabilities.
		Datasets with identical statistics are pooled in the same row.
	}
	\label{tab:data}
	\begin{tabular}{lrrrr}
	\toprule
	Dataset&$\outregulardegree$&$\nnodes$&$\nedges$&$\nicefrac{\objective{\graph}}{\nnodes}$\\
	\midrule
	\multirow{2}{0.23\linewidth}{\synuni, \synhom \quad\quad ($2\cdot 4\cdot 36$~graphs)}&\multirow{2}{*}{5}&\multirow{1}{*}{10$^i$}&\multirow{1}{*}{5$\times$10$^i$}&\multirow{2}{*}{[1.291, 15.231]}\\
	&&\multicolumn{2}{r}{\multirow{1}{*}{for $i\in\{2,3,4,5\}$}}&\\
	\midrule
	\multirow{3}{0.2\linewidth}{\ytone ($3\cdot6$~graphs)}&5&\multirow{3}{*}{40\,415}&202\,075&[0.900, 8.475]\\
	&10&&404\,150&[0.938, 8.701]\\
	&20&&808\,300&[0.989, 9.444]\\
	\midrule
	\multirow{3}{0.2\linewidth}{\yttwo ($3\cdot6$~graphs)}&5&\multirow{3}{*}{150\,572}&752\,860&[0.806, 5.785]\\
	&10&&1\,505\,720&[0.883, 7.576]\\
	&20&&3\,011\,440&[0.949, 8.987]\\
	\midrule
	\multirow{3}{0.2\linewidth}{\nelaone ($3\cdot6$~graphs)}&5&\multirow{3}{*}{11\,931}&59\,655&[4.217, 9.533]\\
	&10&&119\,310&[4.248, 9.567]\\
	&20&&238\,620&[4.217, 9.533]\\
	\midrule
	\multirow{3}{0.2\linewidth}{\nelatwo ($3\cdot6$~graphs)}&5&\multirow{3}{*}{57\,447}&287\,235&[4.609, 11.068]\\
	&10&&574\,470&[4.392, 10.769]\\
	&20&&1\,148\,940&[4.329, 10.741]\\
	\midrule
	\multirow{3}{0.2\linewidth}{\nelathree ($3\cdot6$~graphs)}&5&\multirow{3}{*}{93\,455}&467\,275&[5.565, 11.896]\\
	&10&&934\,550&[5.315, 11.660]\\
	&20&&1\,869\,100&[5.138, 11.517]\\
	\bottomrule
\end{tabular}

\end{table}

\paragraph{Synthetic data}
As our synthetic data, 
we generate a total of $288$ synthetic graphs of four different sizes using two different edge placement models and various parametrizations.
The first model, \synuni, 
chooses out-edges \emph{uniformly} at random, 
similar to a directed Erd\H{o}s-R\'enyi model~\cite{erdHos1959random}. 
In contrast, the second model, \synhom, 
chooses edges \emph{preferentially} to favor small distances between the costs of the source and the target node, 
implementing the concept of \emph{homophily} \cite{mcpherson2001birds}.
We use these graphs primarily to analyze the behavior of our objective function, 
and to understand the impact of using $\heuristic$ instead of $\Delta$ to select the greedily optimal rewiring (\cref{apx:exp:heuristic}).

\paragraph{Real-world data}
We work with real-world data from two domains, 
video recommendations (\yt) and news feeds (\nf).
For our \emph{video application}, 
we use the YouTube data by \citeauthor{ribeiro2020auditing} \cite{mamie2021anti,ribeiro2020auditing}, 
which contains identifiers and ``Up Next''-recommendations for videos from selected channels categorized 
to reflect different degrees and directions of radicalization. 
For our \emph{news application}, 
we use subsets of the NELA-GT-2021 dataset \cite{gruppi2020nelagt2021}, 
which contains 1.8 million news articles published in 2021 from 367 outlets, 
along with veracity labels from Media Bias/Fact Check. 
Prior versions of both datasets are used in the experiments reported by \citeauthor{fabbri2022rewiring} \cite{fabbri2022rewiring}. 

\paragraph{Parametrizations}
To comprehensively assess the effect of modeling assumptions regarding the input graph and its associated random-walk process on our measure of exposure as well as on the performance of \ourmethod and its competitors, 
we experiment with a variety of parametrizations expressing these assumptions.
For all datasets, 
we distinguish three random-walk absorption probabilities $\alpha\in\{0.05,0.1,0.2\}$ 
and two probability shapes $\probabilityshape\in\{\mathbf{U},\mathbf{S}\}$ over the  out-edges of each node (\textbf{U}ni\-form and \textbf{S}kewed).
For our synthetic datasets, 
we further experiment with three fractions of latently harmful nodes $\fractionbad\in\{0.3,0.5,0.7\}$ and 
two cost functions $\labeling \in \{\labeling_{B},\labeling_{R}\}$, 
one binary and one real-valued.
Lastly, for our real-world datasets,   
we distinguish three regular out-degrees $\outregulardegree\in\{5,10,20\}$,
five quality thresholds $\qualitythreshold\in\{0.0,0.5,0.9,0.95,0.99\}$ 
and four cost functions,
two binary ($\labeling_B1$, $\labeling_B2$) and two real-valued ($\labeling_{R1}, \labeling_{R2}$),
based on labels provided with the original datasets, 
as detailed in \cref{apx:data:costfunctions}.

\subsubsection{Algorithms}
\label{exp:alg}
We compare \ourmethod, our algorithm for \ourproblem and \ourproblemtwo, 
with four baselines (\baselineone-\baselinefour) 
and the algorithm by \citeauthor{fabbri2022rewiring} \cite{fabbri2022rewiring} for minimizing the maximum segregation, which we call \fabbrialg. 
In all \ourproblemtwo experiments, 
we use the $\bigoh{1}$-computable normalized discounted cumulative gain (\ndcg),
defined in \cref{apx:reproducibility:qualityf} and also used by \fabbrialg,
as a relevance function $\qualityf$, 
and consider the $100$ most relevant nodes as potential rewiring targets.

As \fabbrialg can only handle binary costs, 
we transform nonbinary costs $\labeling$ into binary costs $\labeling'$ by thresholding to ensure $\labeling_i \geq \roundingthreshold \Leftrightarrow \labeling'_i = 1$ for some rounding threshold $\roundingthreshold\in(0,1]$ 
(cf. \cref{apx:reproducibility:binarization}).
Since \fabbrialg requires access to relevance information, 
we restrict our comparisons with \fabbrialg to data where this information is available.

Our baselines \baselineone-\baselinefour are ablations of \ourmethod, 
such that outperforming them shows how each component of our approach is beneficial. 
We order the baselines by the competition we expect from them, 
from no competition at all (\baselineone) to strong competition (\baselinefour). 
Intuitively, \baselineone does not consider our objective at all,
\baselinetwo is a heuristic focusing on the $\tau$ component of our objective, 
\baselinethree is a heuristic focusing on the $\sigma$ component of our objective, 
and \baselinefour is a heuristic eliminating the iterative element of our approach. 
\baselineone--\baselinethree each run in $\budget$ rounds, while \baselinefour runs in one round. 
In each round,
\baselineone randomly selects a permissible rewiring via rejection sampling.
\baselinetwo selects the rewiring $(i,j,k)$ with the node $j$ maximizing $\unitvector_j^T\fundamental\labelingvector$ as its old target, 
the node $i$ with $j\in\outneighbors{i}$ maximizing $\onevector^T\fundamental\unitvector_i$ as its source,
and the available node $k$ minimizing $\unitvector_k^T\fundamental\labelingvector$ as its new target. 
\baselinethree selects the rewiring $(i,j,k)$ with the node $i$ maximizing $\onevector^T\fundamental\unitvector_i$ as its source, 
the node $j$ with $j\in\outneighbors{i}$ maximizing $\unitvector_j^T\fundamental\labelingvector$ as its old target, and the available node $k$ minimizing $\unitvector_k^T\fundamental\labelingvector$ as its new target.
\baselinefour selects the $\budget$ rewirings with the largest initial values of $\heuristic$, 
while ensuring each edge is rewired at most once.

\subsubsection{Implementation and reproducibility}
All algorithms, including \ourmethod, the baselines, and \fabbrialg, 
are implemented in Python 3.10.
We run our experiments on a 2.9 GHz 6-Core Intel Core i9 with 32 GB RAM and report wall-clock time.
All code, datasets, and results are publicly available,\!\footnote{\oururl}
and we provide further reproducibility information in \cref{apx:reproducibility}.

\subsection{Results}

\subsubsection{Impact of modeling choices}
To understand the impact of a particular modeling choice on the performance of \ourmethod and its competitors, 
we analyze groups of experimental settings that vary only the parameter of interest while keeping the other parameters constant,
focusing on the \ytone datasets. 
We primarily report the evolution of the ratio $\nicefrac{\objective{\graph_{\budget}}}{\objective{\graph}}=\nicefrac{\big(\objective{\graph}-\maxobjective{\graph,\graph_{\budget}}\big)}{\objective{\graph}}$, 
which indicates what fraction of the initial expected total exposure is left after $\budget$ rewirings, 
and hence is comparable across \ourproblem instances with different starting values.
Overall, we observe that \ourmethod robustly reduces the expected total exposure to harm, 
and that it changes its behavior predictably under parameter variations. 
Due to space constraints, we defer the results showing this for variations in the regular out-degree $\outregulardegree$, 
the random-walk absorption probability $\pabsorption$, 
the probability shape $\probabilityshape$,
and the cost function $\labeling$ 
to \cref{apx:exp:modelingchoices}.

\paragraph{Impact of quality threshold $\qualitythreshold$}
The higher the quality threshold $\qualitythreshold$, 
the more constrained our rewiring options. 
Thus, under a given budget $\budget$, 
we expect \ourmethod to reduce our objective more strongly for smaller $\qualitythreshold$. 
As illustrated in \cref{fig:quality_threshold}, 
our experiments confirm this intuition, 
and the effect is more pronounced if the out-edge probability distribution is skewed.
We further observe that \ourmethod can guarantee $\qualitythreshold = 0.5$ with little performance impact, 
and it can strongly reduce the exposure to harm even under a strict $\qualitythreshold = 0.95$:
With just $100$ edge rewirings, 
it reduces the expected total exposure to harm by $50\%$, 
while ensuring that its recommendations are at most $5\%$ less relevant than the original recommendations.

\subsubsection{Performance comparisons}
Having ensured that \ourmethod robustly and predictably reduces the total exposure across the entire spectrum of modeling choices, 
we now compare it with its competitors.
Overall, we find that \ourmethod offers more reliable performance and achieves stronger harm reduction than its contenders.

\paragraph{Comparison with baselines \baselineone--\baselinefour}
First, we compare \ourmethod with our four baselines, 
each representing a different ablation of our algorithm.  
As depicted in \cref{fig:baselines}, 
the general pattern we observe matches our performance expectations (from weak performance of \baselineone to strong performance of \baselinefour), 
but we are struck by the strong performance of \baselinethree 
(selecting based on $\sigma$), 
especially in contrast to the weak performance of \baselinetwo 
(selecting based on $\tau$). 
This suggests that whereas the most \emph{exposed} node does not necessarily have a highly visited node as an in-neighbor, 
the most \emph{visited} node tends to have a highly exposed node as an out-neighbor.
In other words, for some highly \emph{prominent} videos, 
the YouTube algorithm
problematically
appears to recommend highly \emph{harm-inducing} content to watch next. 
Despite the competitive performance of \baselinethree and \baselinefour,
\ourmethod consistently outperforms these baselines, too, 
and unlike the baselines, it \emph{smoothly} reduces the exposure function.
This lends additional support to our reliance on $\sigma\tau$ (rewiring a highly visited $i$ away from a highly exposed $j$) as an \emph{iteratively} evaluated heuristic.

\paragraph{Comparison with \fabbrialg}
Having established that all components of \ourmethod are needed to achieve its performance, 
we now compare our algorithm with \fabbrialg, the method proposed by \citeauthor{fabbri2022rewiring} \cite{fabbri2022rewiring}. 
To this end, we run both \ourmethod and \fabbrialg using their respective objective functions,
i.e., the expected total exposure to harm of random walks starting at any node (\emph{total exposure}, \ourmethod)
and the maximum expected number of random-walk steps from a harmful node to a benign node (\emph{maximum segregation}, \fabbrialg).
Reporting their performance under the objectives of \emph{both} algorithms (as well as the \emph{total segregation}, which sums the segregation scores of all harmful nodes) in \cref{fig:gamine_vs_mms}, 
we find that under strict quality control ($\qualitythreshold\in\{0.9,0.95,0.99\}$), 
\ourmethod outperforms \fabbrialg on \emph{all} objectives, 
and \fabbrialg stops early as it can no longer reduce its objective function.
For $\qualitythreshold = 0.5$, \fabbrialg outperforms \ourmethod on the segregation-based objectives, 
but \ourmethod still outperforms \fabbrialg on our exposure-based objective, sometimes at twice the margin (\cref{fig:largemargin}).
Further, while \ourmethod delivers consistent and predictable performance that is strong on exposure-based and segregation-based objectives, 
we observe much less consistency in the performance of \fabbrialg.
For example, it is counterintuitive that \fabbrialg identifies $100$ rewirings on the smaller \ytone data but stops early on the larger \yttwo data. 
Moreover, \fabbrialg delivers the results shown in \cref{fig:gamine_vs_mms} under $\labeling_{B1}$, 
but it cannot decrease its objective at all on the same data under $\labeling_{B2}$, 
which differs from $\labeling_{B1}$ only in that it also assigns harm to anti-feminist content (\cref{apx:datasets}, \cref{tab:yt-costs}).
We attribute this brittleness to the reliance on the \emph{maximum}-based \emph{segregation} objective, 
which, by design, 
is less robust than our \emph{sum}-based \emph{exposure} objective.

\begin{figure}[t]
	\centering
	\begin{subfigure}{0.5\linewidth}
		\centering
		\includegraphics[width=\linewidth]{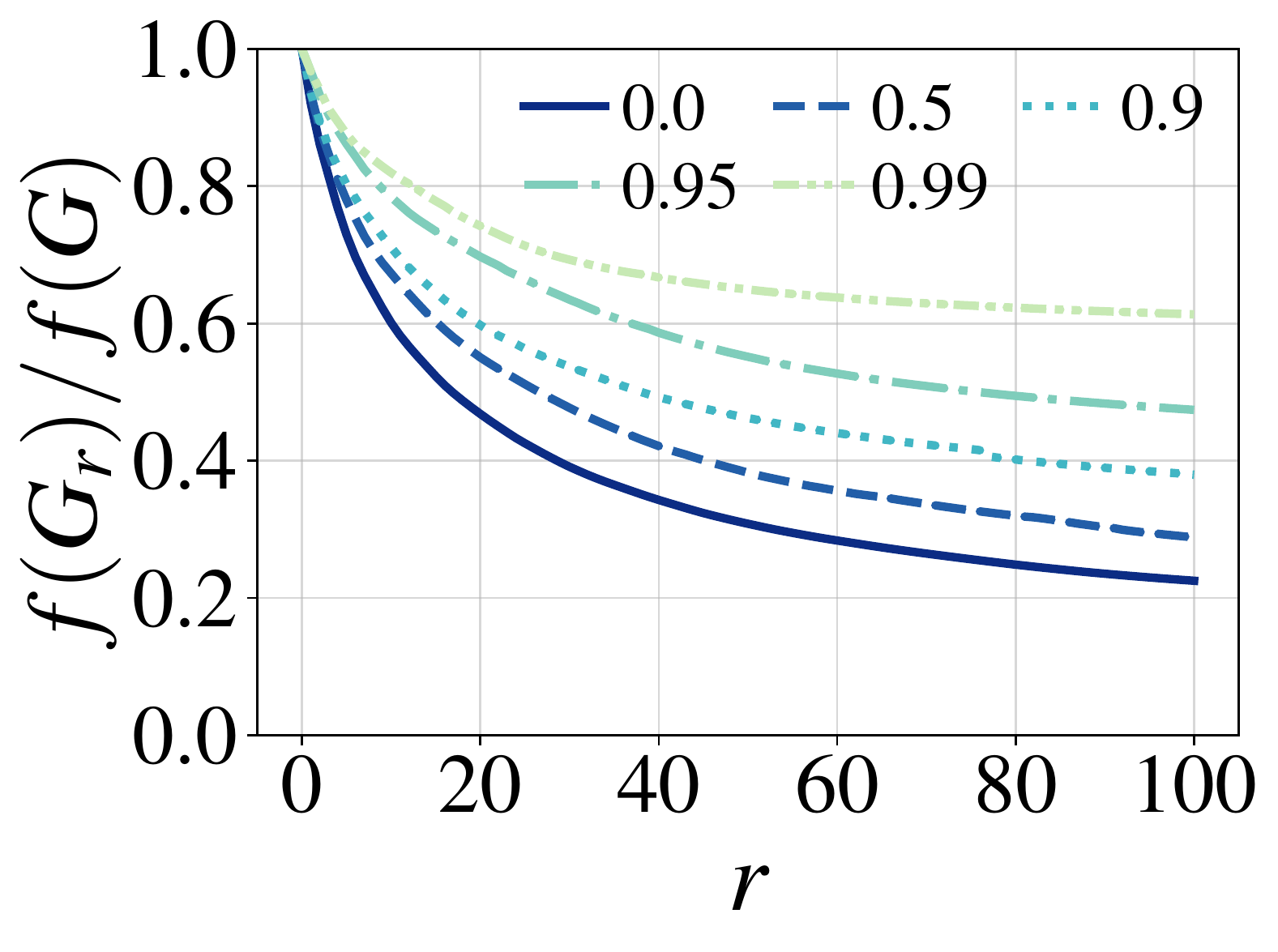}
		\subcaption{Probability shape $\probabilityshape = \mathbf{U}$}
	\end{subfigure}~%
	\begin{subfigure}{0.5\linewidth}
		\centering
		\includegraphics[width=\linewidth]{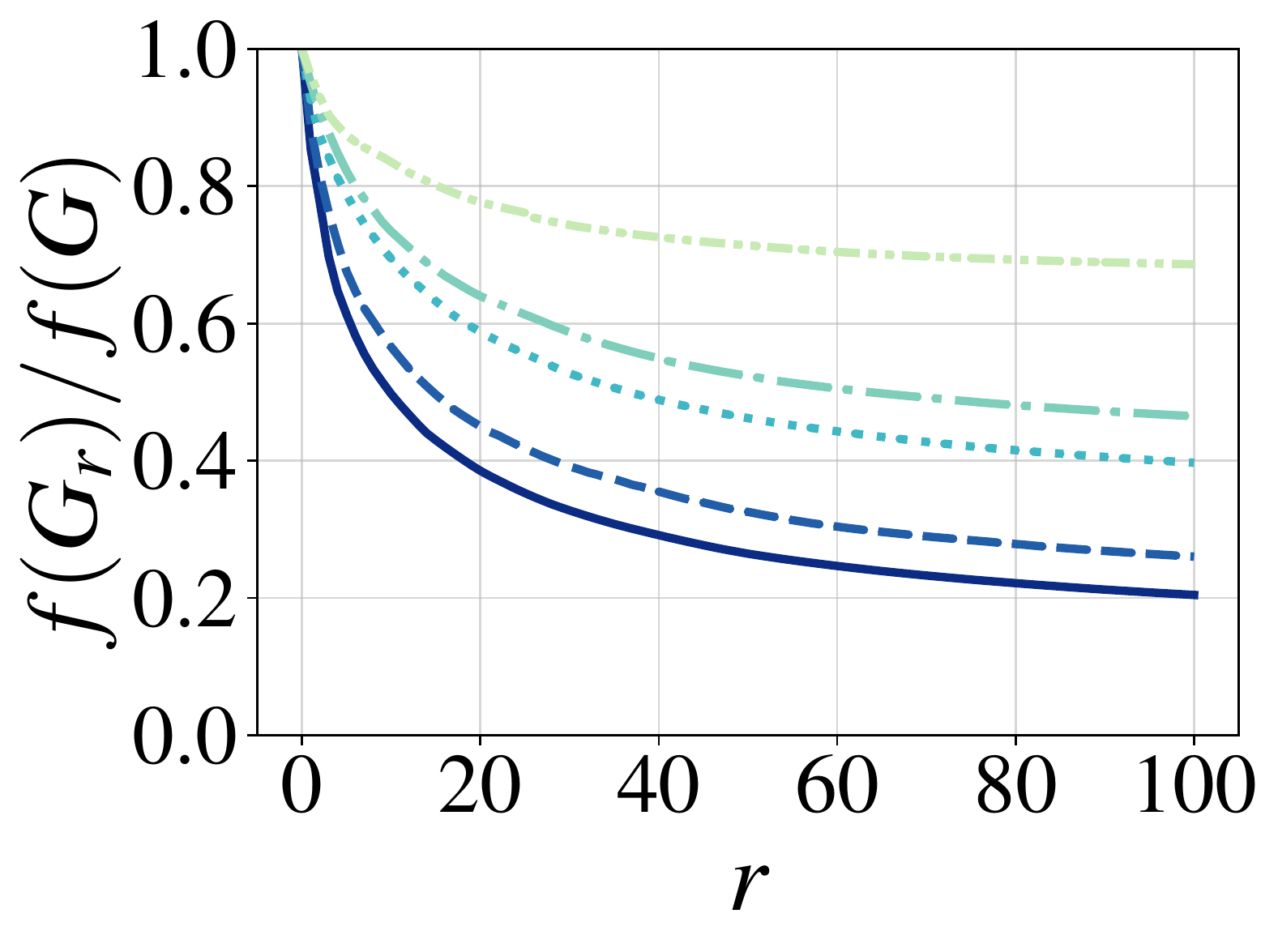}
		\subcaption{Probability shape $\probabilityshape = \mathbf{S}$}
	\end{subfigure}
	\caption{%
		Performance of \ourmethod for quality thresholds $\qualitythreshold\in\{0.0,0.5,0.9,0.95,0.99\}$ as measured by $c_{B2}$,
		run on \ytone with $\outregulardegree=5$ and $\pabsorption=0.05$.
		\ourmethod can ensure $\qualitythreshold = 0.5$ with little loss in performance, 
		and it can reduce our objective considerably even under a strict $\qualitythreshold = 0.95$.
	}\label{fig:quality_threshold}
\end{figure}
\begin{figure}[t] 
	\begin{subfigure}{0.5\linewidth}
		\centering
		\includegraphics[width=\linewidth]{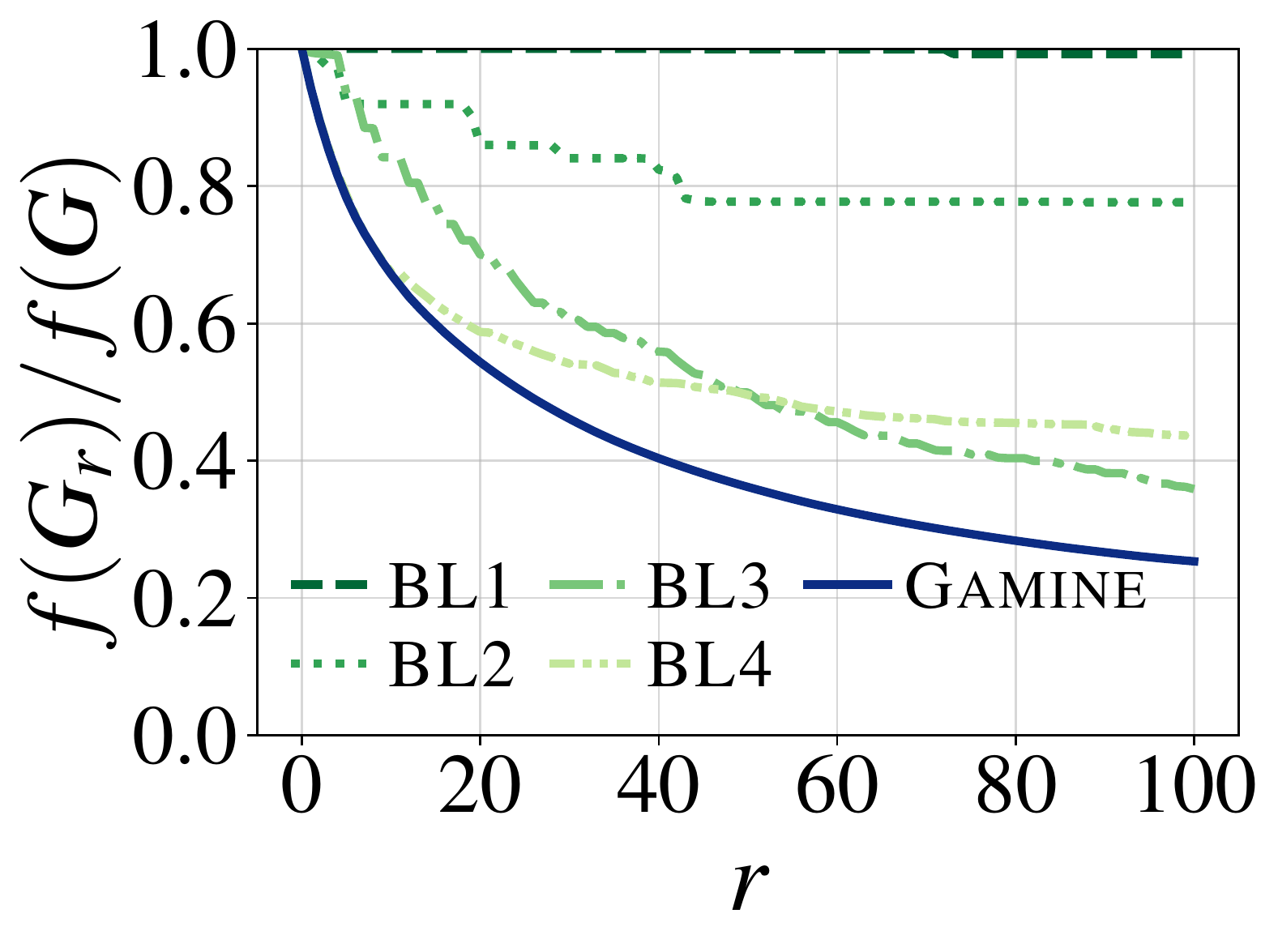}
		\subcaption{Regular out-degree $\outregulardegree = 10$}
	\end{subfigure}~%
	\begin{subfigure}{0.5\linewidth}
		\centering
		\includegraphics[width=\linewidth]{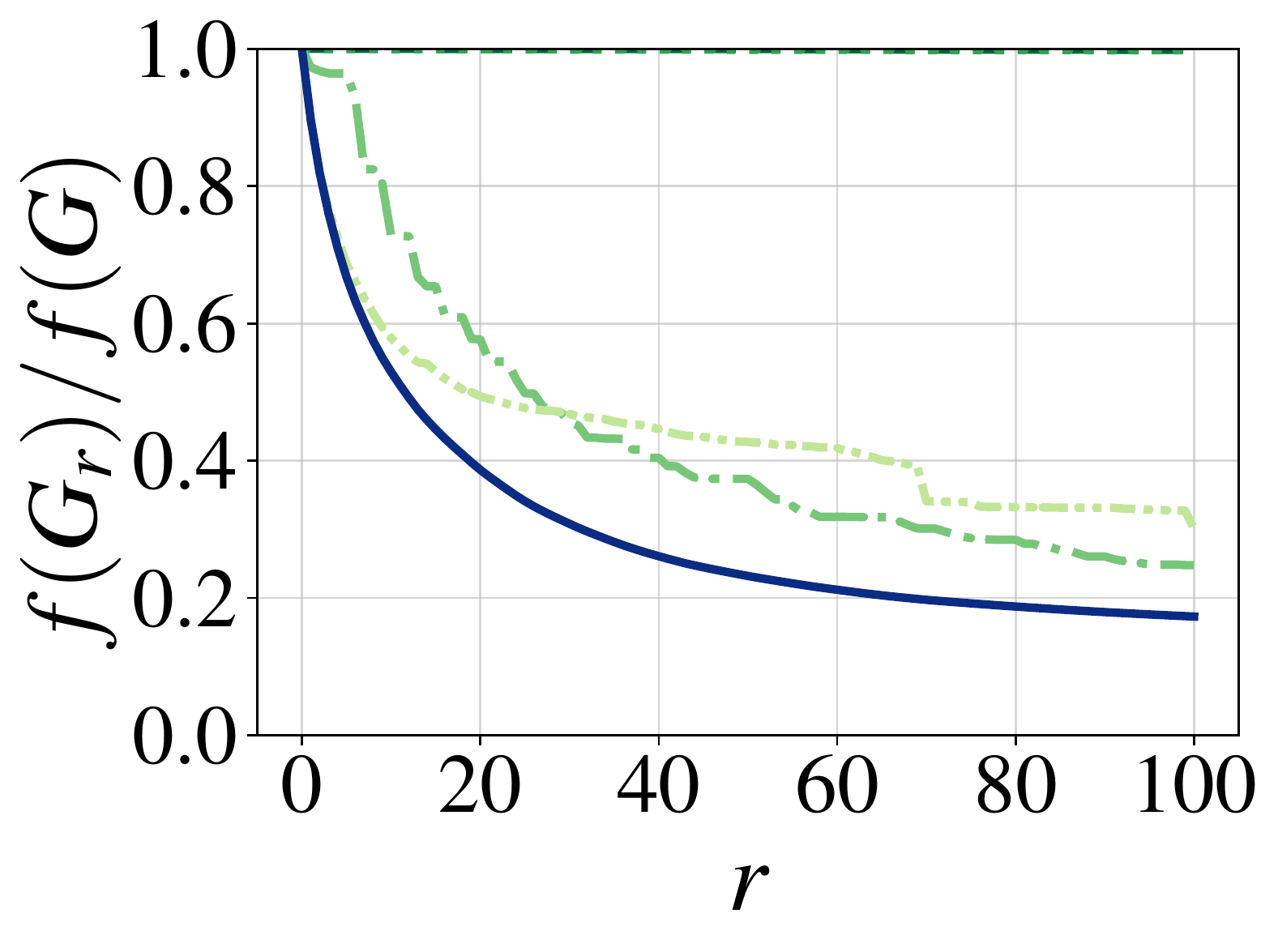}
		\subcaption{Regular out-degree $\outregulardegree = 5$}
	\end{subfigure}
	\caption{%
		Performance of \ourmethod with $\qualitythreshold = 0.0$, compared with the four baselines \baselineone, \baselinetwo, \baselinethree, and \baselinefour under $c_{B1}$,
		run on \ytone with $\pabsorption = 0.05$ and $\probabilityshape = \mathbf{U}$. 
		As \baselinefour is roundless, we apply its rewirings in decreasing order of $\Delta$ to depict its performance as a function of $\budget$.
		\ourmethod outcompetes all baselines, 
		but \baselinethree and \baselinefour also show strong performance. 
	}\label{fig:baselines}
\end{figure}
\begin{figure}[t]
	\centering
	\begin{subfigure}{0.5\linewidth}
		\centering
		\includegraphics[width=\linewidth]{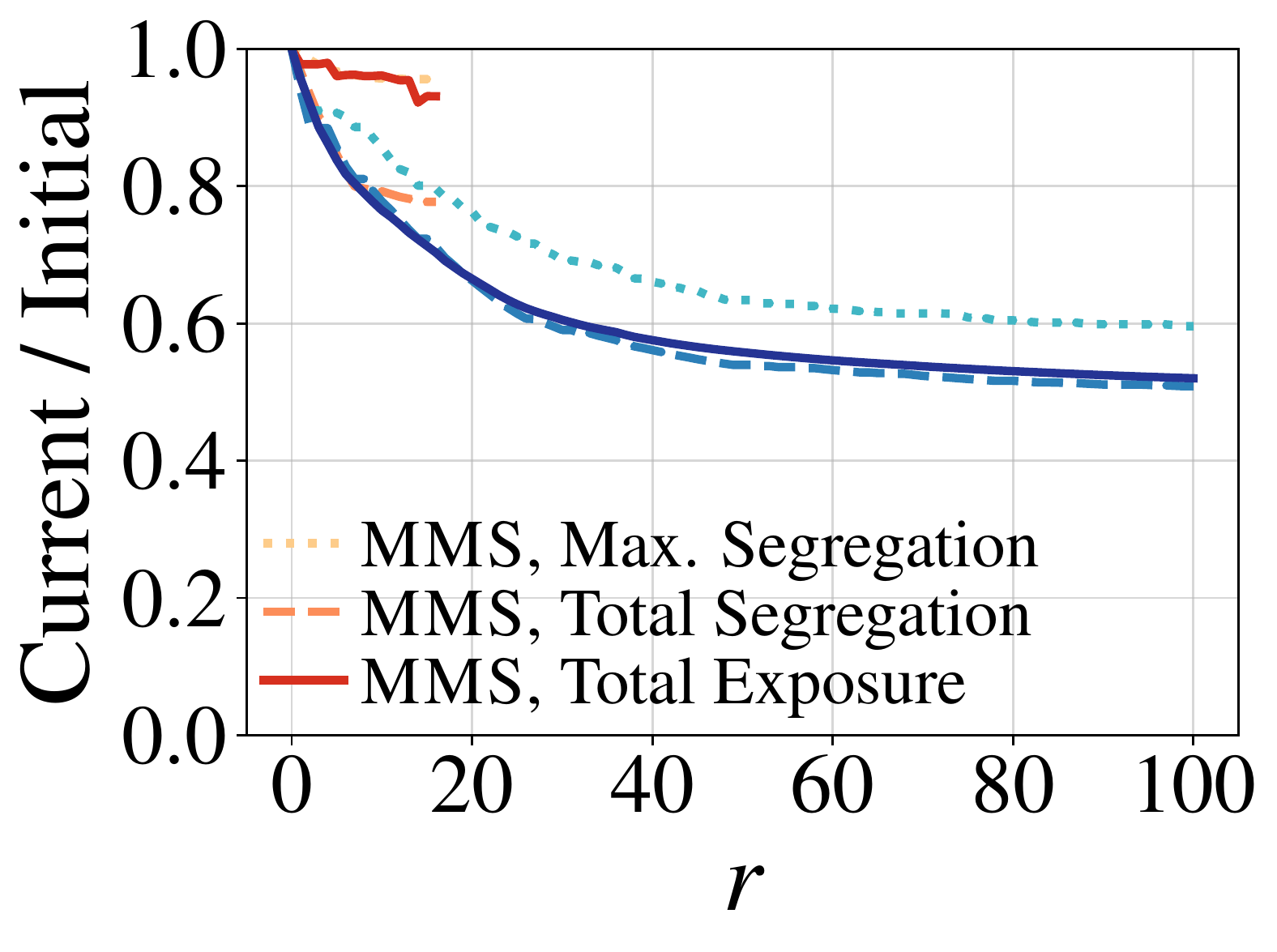}
		\subcaption{\ytone, $\qualitythreshold=0.99$}
	\end{subfigure}~%
	\begin{subfigure}{0.5\linewidth}
		\centering
		\includegraphics[width=\linewidth]{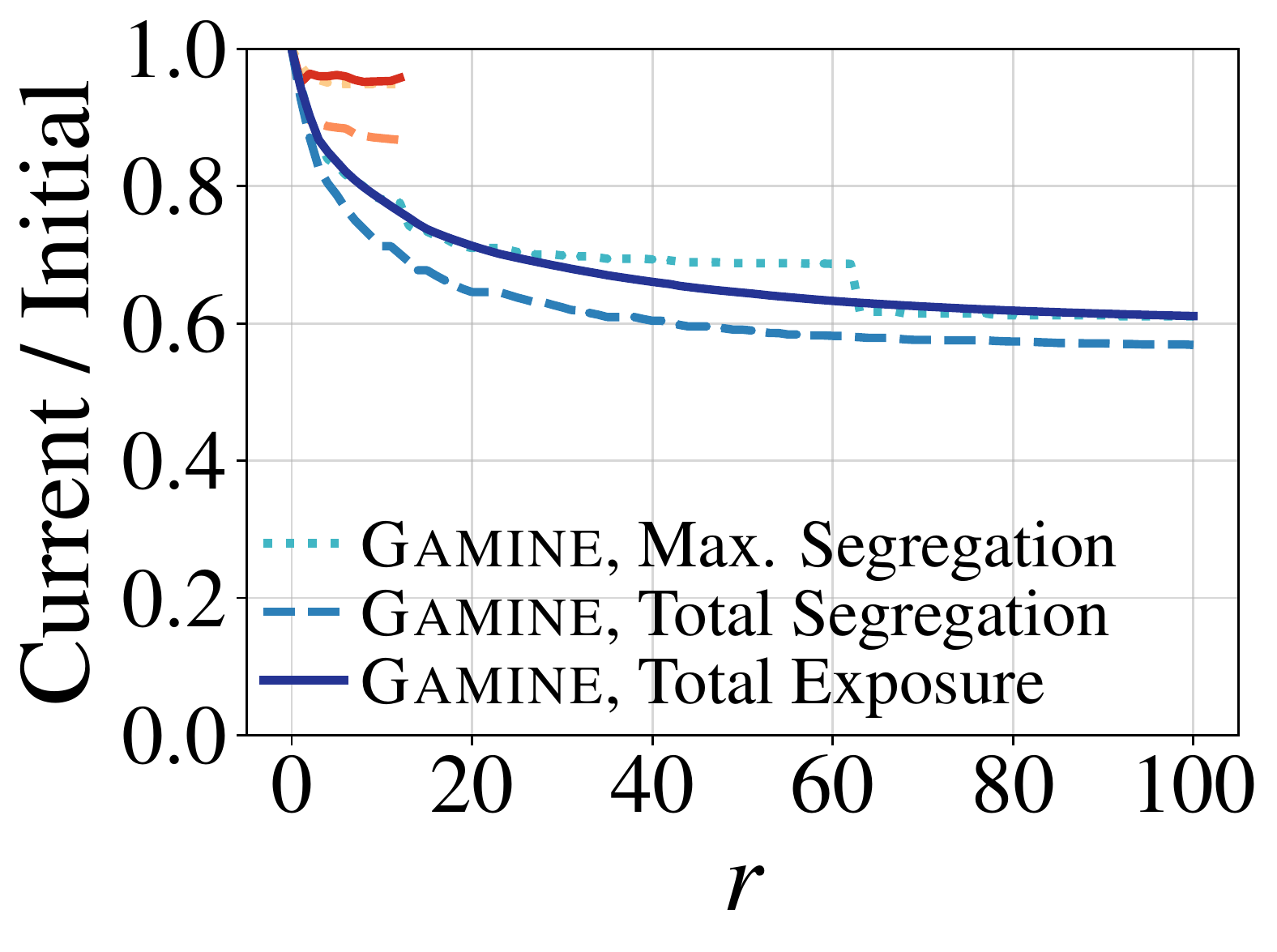}
		\subcaption{\yttwo, $\qualitythreshold=0.99$}
	\end{subfigure}\vspace*{3pt}
	\begin{subfigure}{0.5\linewidth}
		\centering
		\includegraphics[width=\linewidth]{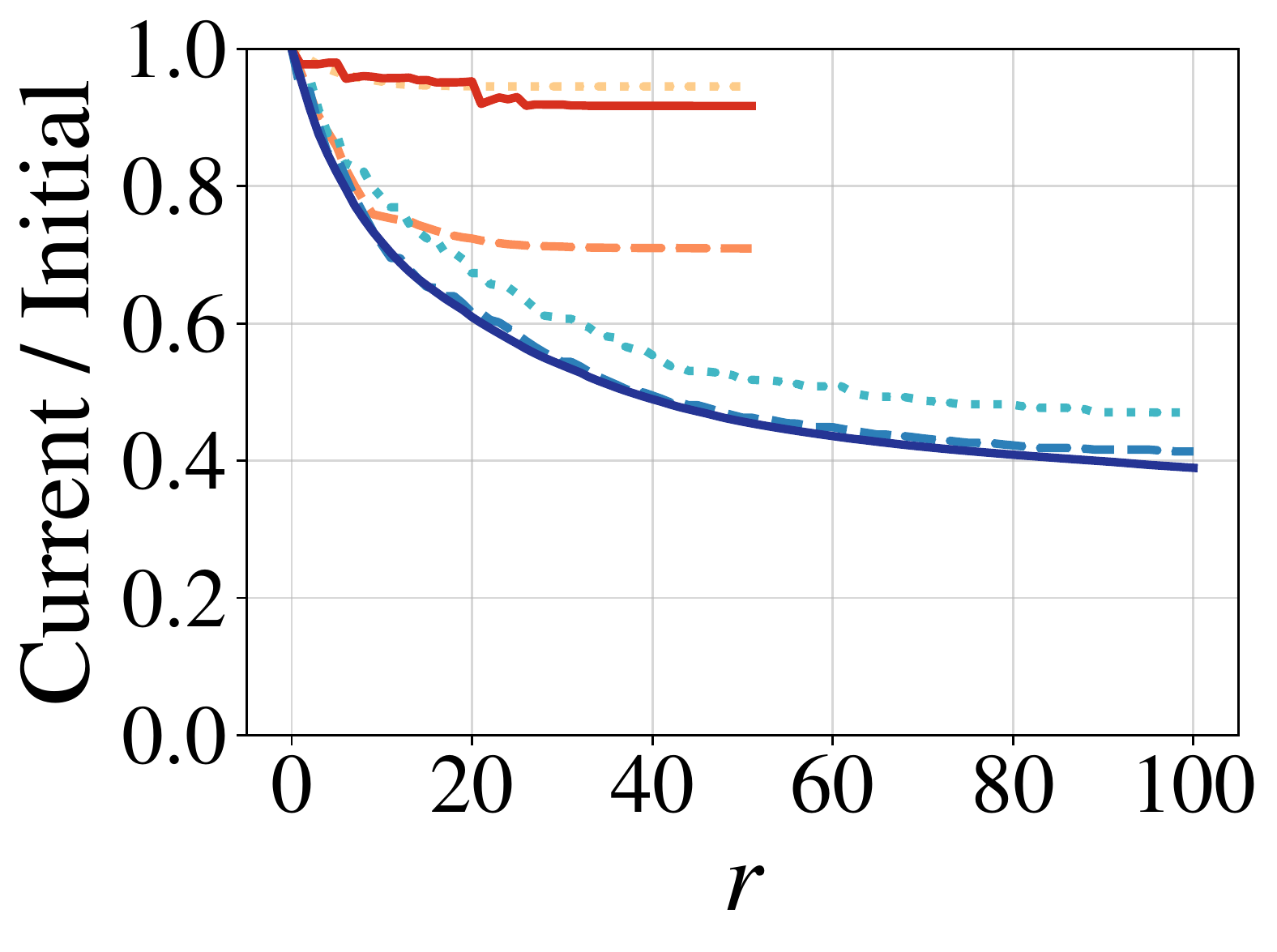}
		\subcaption{\ytone, $\qualitythreshold=0.95$}
	\end{subfigure}~%
	\begin{subfigure}{0.5\linewidth}
		\centering
		\includegraphics[width=\linewidth]{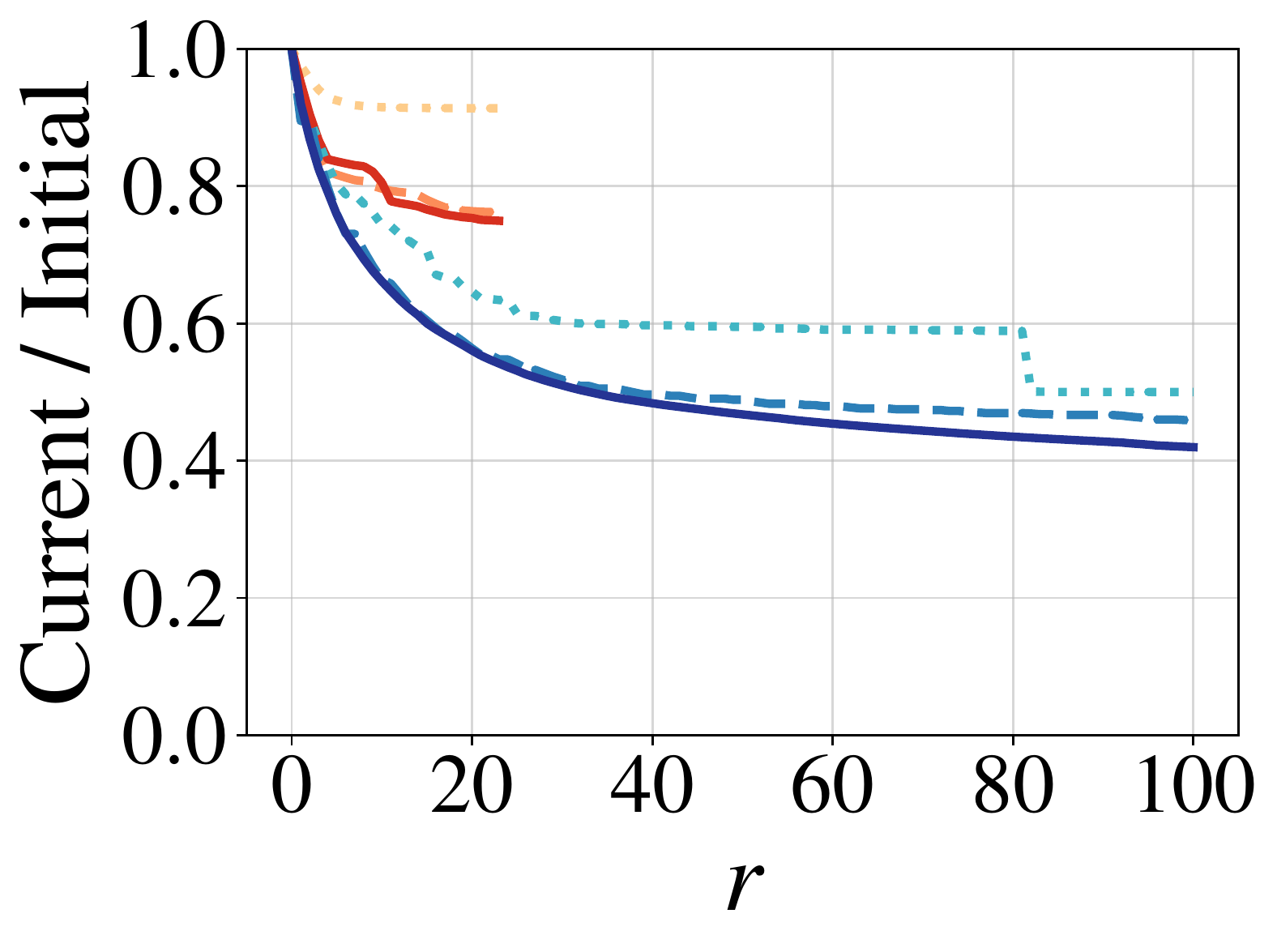}
		\subcaption{\yttwo, $\qualitythreshold=0.95$}
	\end{subfigure}\vspace*{3pt}
	\begin{subfigure}{0.5\linewidth}
		\centering
		\includegraphics[width=\linewidth]{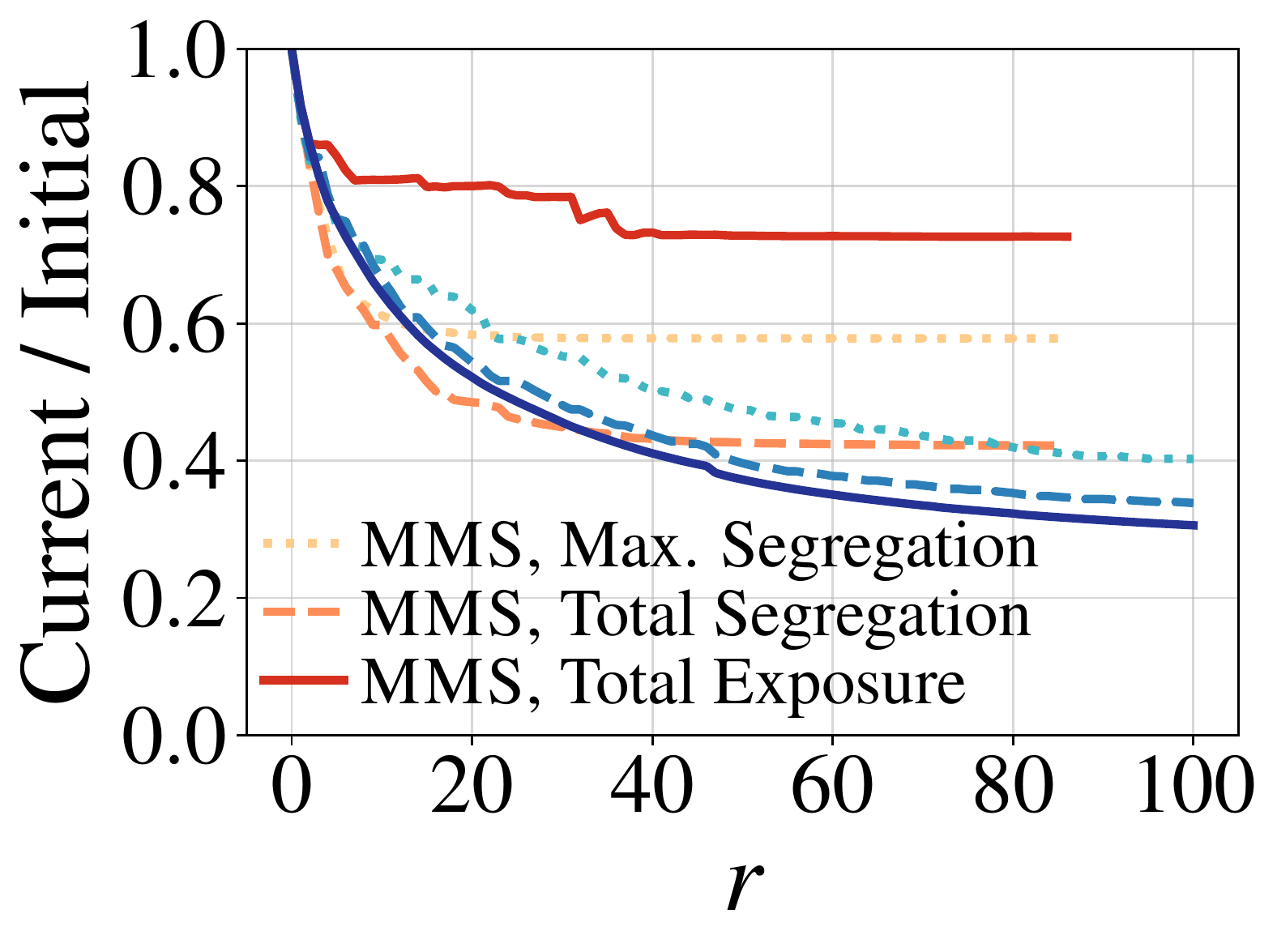}
		\subcaption{\ytone, $\qualitythreshold=0.9$}
	\end{subfigure}~%
	\begin{subfigure}{0.5\linewidth}
		\centering
		\includegraphics[width=\linewidth]{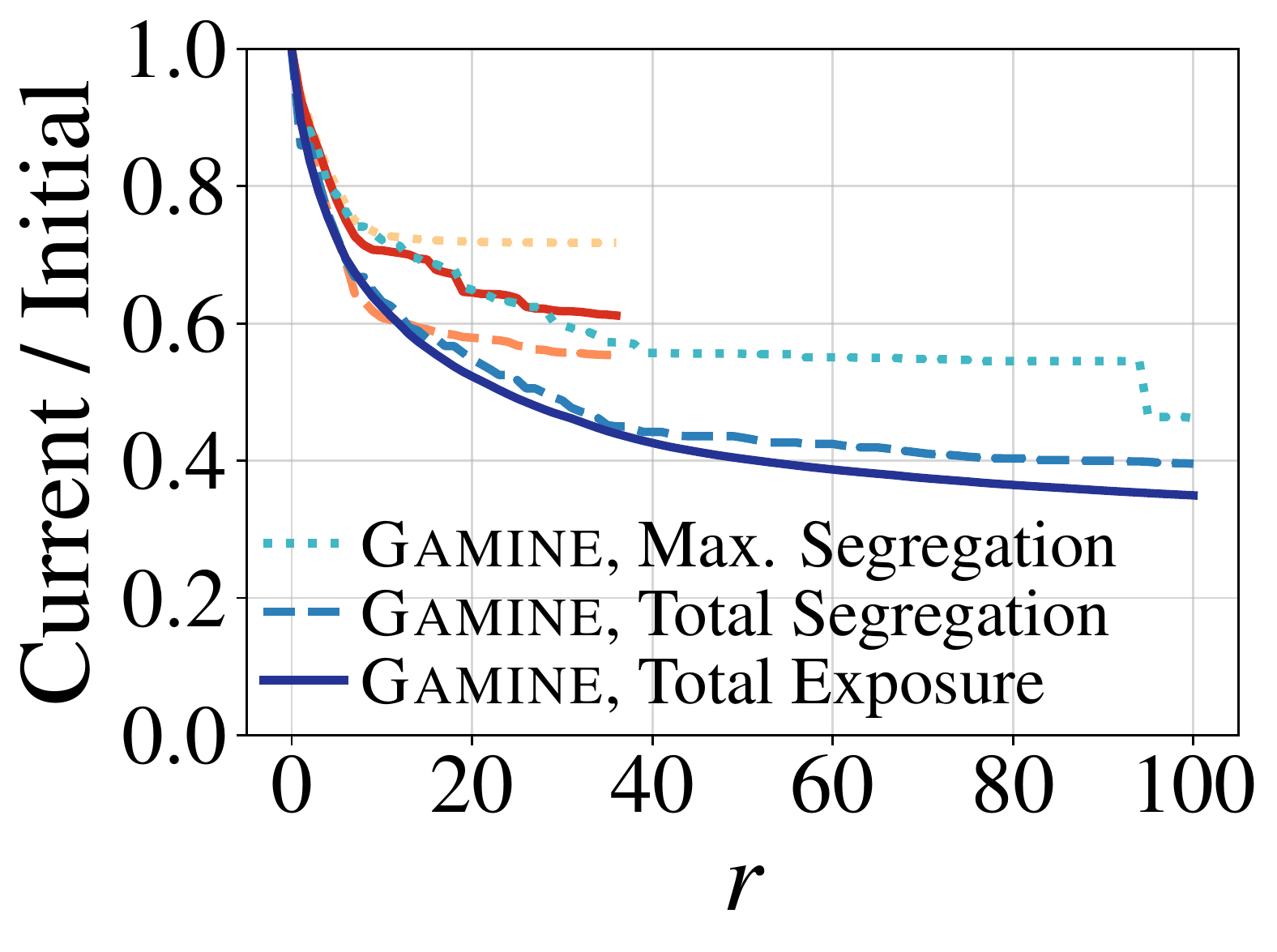}
		\subcaption{\yttwo, $\qualitythreshold=0.9$}
	\end{subfigure}\vspace*{3pt}
	\begin{subfigure}{0.5\linewidth}
		\centering
		\includegraphics[width=\linewidth]{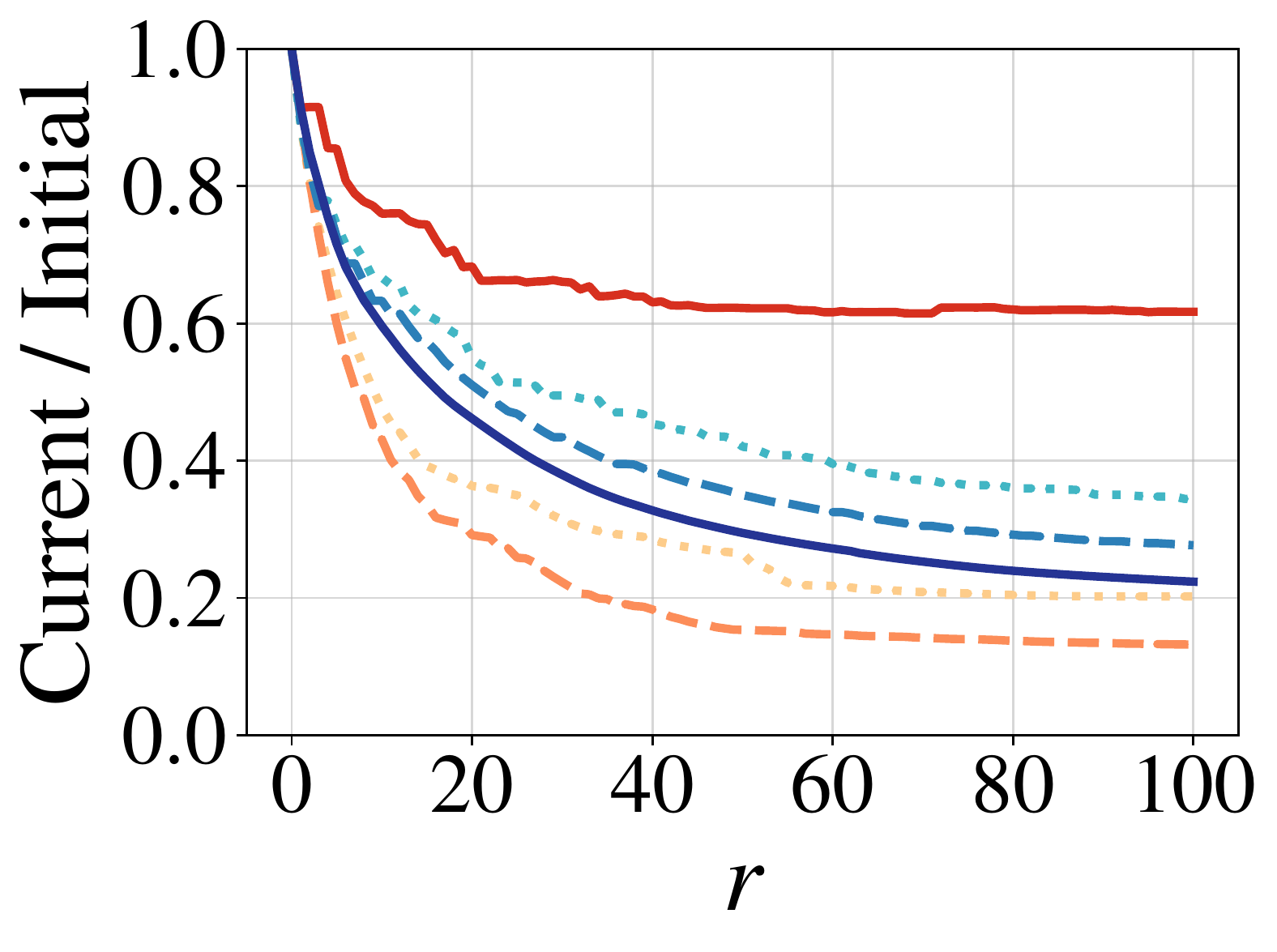}
		\subcaption{\ytone, $\qualitythreshold=0.5$}\label{fig:largemargin}
	\end{subfigure}~%
	\begin{subfigure}{0.5\linewidth}
		\centering
		\includegraphics[width=\linewidth]{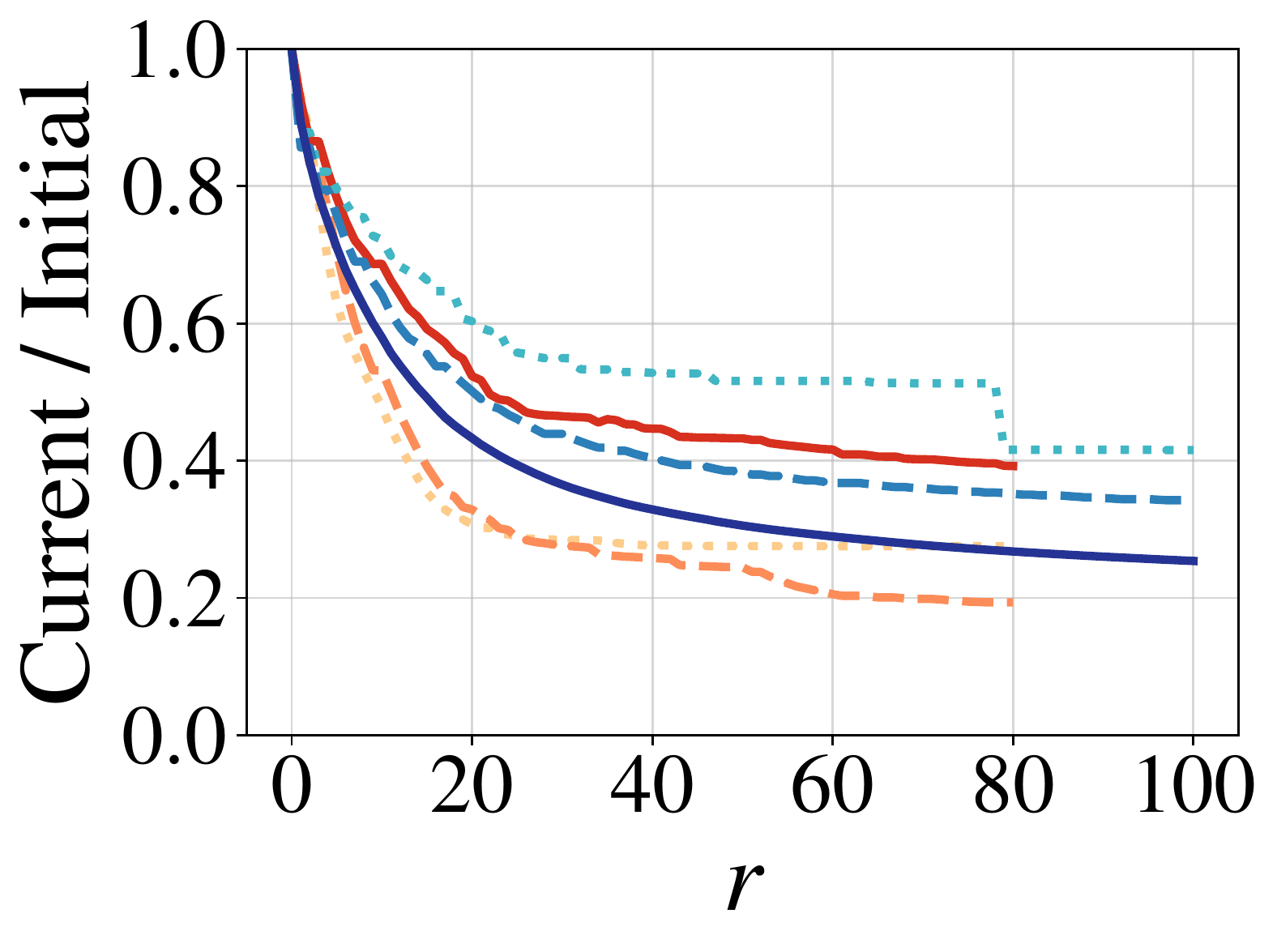}
		\subcaption{\yttwo, $\qualitythreshold=0.5$}
	\end{subfigure}
	\caption{%
		Performance of \ourmethod and \fabbrialg 
		when measured under $c_{B1}$ by the maximum segregation or the total segregation from \citeauthor{fabbri2022rewiring} \cite{fabbri2022rewiring}, 
		or by the total exposure as defined in \cref{eq:harm},
		run on \ytone (left) and \yttwo (right) with $\outregulardegree=5$, $\pabsorption=0.05$, and $\probabilityshape=\mathbf{U}$.
		For all but $\qualitythreshold = 0.5$, 
		\ourmethod outperforms \fabbrialg on \emph{all} objectives, 
		and \fabbrialg stops early because it can no longer reduce the maximum segregation. 
	}\label{fig:gamine_vs_mms}
\end{figure}
\begin{figure}[t]
	\begin{subfigure}{0.5\linewidth}
		\centering
		\includegraphics[width=\linewidth]{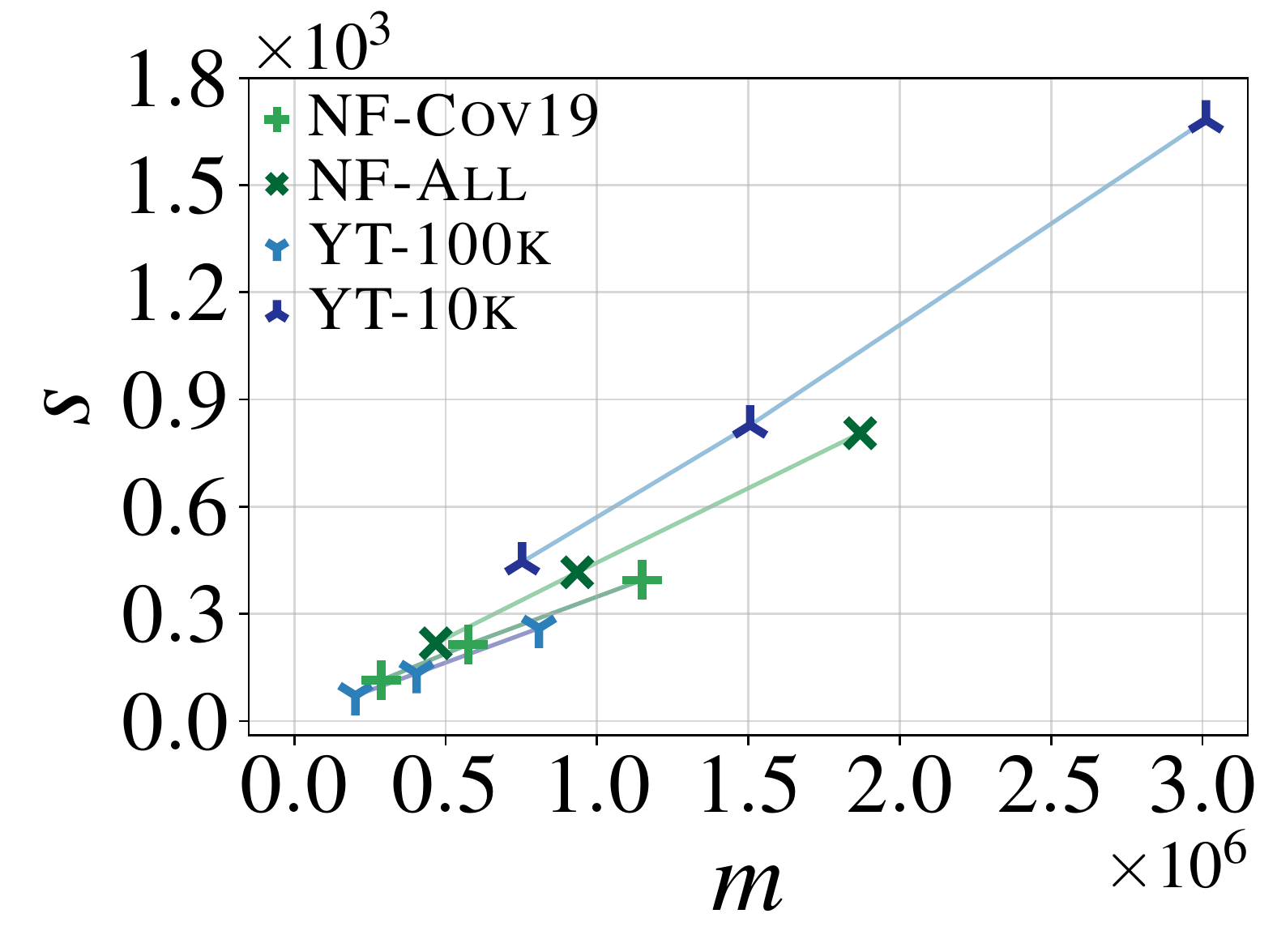}
		\subcaption{\ourmethod on \ourproblem}
	\end{subfigure}~%
	\begin{subfigure}{0.5\linewidth}
		\centering
		\includegraphics[width=\linewidth]{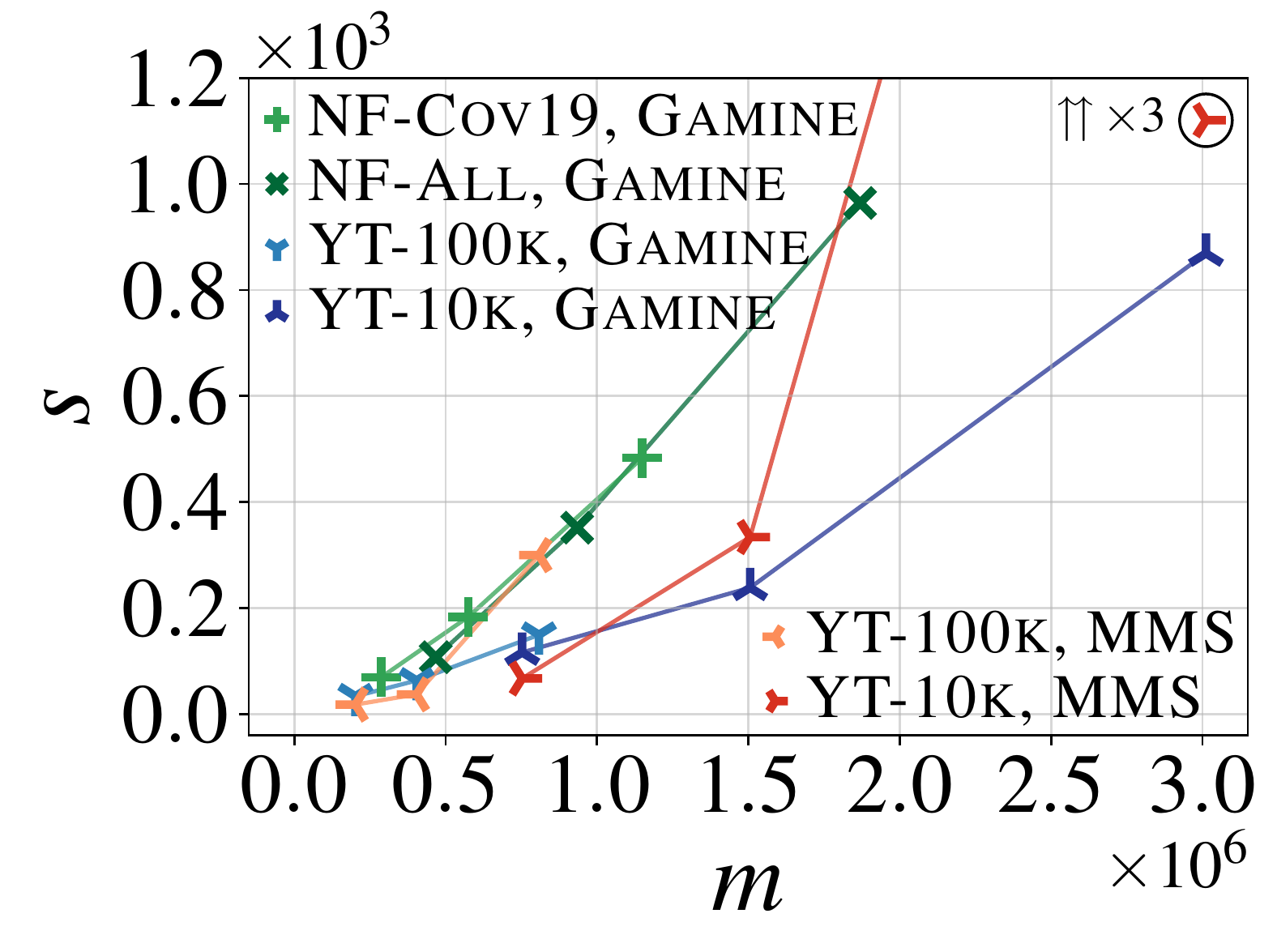}
		\subcaption{\ourmethod and \fabbrialg on \ourproblemtwo}
	\end{subfigure}
	\caption{%
		Scaling of \ourmethod and \fabbrialg
		under $\labeling_{B1}$ with 
		$\pabsorption = 0.05$,
		$\probabilityshape = \mathbf{U}$, and 
		$\qualitythreshold = 0.0$ (\ourproblem) resp.\ $0.99$ (\ourproblemtwo).
		We report the seconds $s$ to compute a single rewiring 
		as a function of $\nedges = \outregulardegree\nnodes$ 
		(\fabbrialg does not identify any rewirings on \nelatwo and \nelathree). 
		\ourmethod scales more favorably than \fabbrialg. 
	}\label{fig:scalability}
\end{figure}

\subsubsection{Empirical scalability of \ourmethod}

In our previous experiments, we found that \ourmethod robustly and reliably reduces the expected total exposure to harm. 
Now, we seek to ascertain that its practical scaling behavior matches our theoretical predictions, 
i.e., that under realistic assumptions on the input, 
\ourmethod scales linearly in $\nnodes$ and $\nedges$.
We are also interested in comparing \ourmethod's scalability to that of \fabbrialg.
To this end, we measure the time taken to compute a single rewiring and report, 
in \cref{fig:scalability},
the \emph{average} over ten
rewirings for each of our datasets.
This corresponds to the time taken by \greedyname in \ourmethod and by \textsc{1-Rewiring} in \fabbrialg, 
which drives the overall scaling behavior of both algorithms. 
We find that \ourmethod scales approximately linearly, 
whereas \fabbrialg scales approximately quadratically 
(contrasting with the empirical time complexity of $\bigoh{\nnodes\log\nnodes}$ claimed in \cite{fabbri2022rewiring}).
This is because our implementation of \fabbrialg follows the original authors', 
whose evaluation of the segregation objective takes time  $\bigoh{\nnodes}$ and is performed $\bigoh{\nedges}$ times. 
The speed of precomputations depends on the problem variant (\ourproblem vs. \ourproblemtwo), 
and for \ourproblemtwo, also on the quality function $\qualityf$. 
In our experiments, precomputations add linear overhead for \ourmethod and volatile overhead for \fabbrialg, 
as we report in \cref{apx:exp:scalability}.

\subsubsection{Data complexity}
Given that \ourmethod strongly reduces the expected total exposure to harm with few rewirings on the YouTube data, 
as evidenced in \cref{fig:quality_threshold,fig:baselines,fig:gamine_vs_mms}, 
one might be surprised to learn that its performance seems much weaker on the NELA-GT data (\cref{apx:exp:nela}): 
While it still reduces the expected total exposure and outperforms \fabbrialg (which struggles to reduce its objective at all on the \nf data), 
the impact of individual rewirings is much smaller than on the YouTube datasets, 
and the value of the quality threshold $\qualitythreshold$ barely makes a difference. 
This motivates us to investigate how \emph{data complexity} impacts our ability to reduce the expected total exposure to harm via edge rewiring: 
Could reducing exposure to harm be \emph{intrinsically} harder on \nf data than on \yt data? 
The answer is yes. 
First, 
the in-degree distributions of the \yt graphs are an order of magnitude more skewed than those of the \nf graphs (\cref{apx:datasets:statistics}, \cref{fig:real-indegrees}). 
This is unsurprising given the different origins of their edges (user interactions vs. cosine similarities), 
but it creates opportunities for high-impact rewirings involving highly prominent nodes in \yt graphs (which \ourmethod seizes in practice, see below).
Second, as depicted in \cref{fig:channelclasses},
harmful and benign nodes are much more strongly interwoven in the \nf data than in the \yt data. 
This means that harmful content is less siloed in the \nf graphs, 
but it also impedes strong reductions of the expected total exposure. 
Third, 
as a result of the two previous properties, 
the initial node exposures are much more concentrated in the \nf graphs than in the \yt graphs, 
as illustrated in \cref{fig:initialexposure}, 
with a median sometimes twice as large as the median of the identically parametrized \yt graphs, 
and a much higher average exposure (cf.~$\nicefrac{\objective{\graph}}{\nnodes}$ in \cref{tab:data}). 
Finally, the relevance scores are much more skewed in the \yt data than in the \nf data (\cref{apx:datasets:statistics}, \cref{fig:relevancescores}). 
Hence, while we are strongly constrained by $\qualitythreshold$ on the \yt data even when considering only the $100$ highest-ranked nodes as potential rewiring targets, 
we are almost unconstrained in the same setting on the \nf data, 
which explains the comparative irrelevance of $\qualitythreshold$ on the \nf data.
Thus, the performance differences we observe between the \nf data and the \yt data are due to intrinsic dataset properties: 
\ourproblem and \ourproblemtwo are simply more complex on the news data than on the video data.

\begin{figure}[t]
	\centering
	\hspace*{-5pt}\begin{subfigure}[b]{0.5\linewidth}
		\centering
		\includegraphics[width=0.94\linewidth]{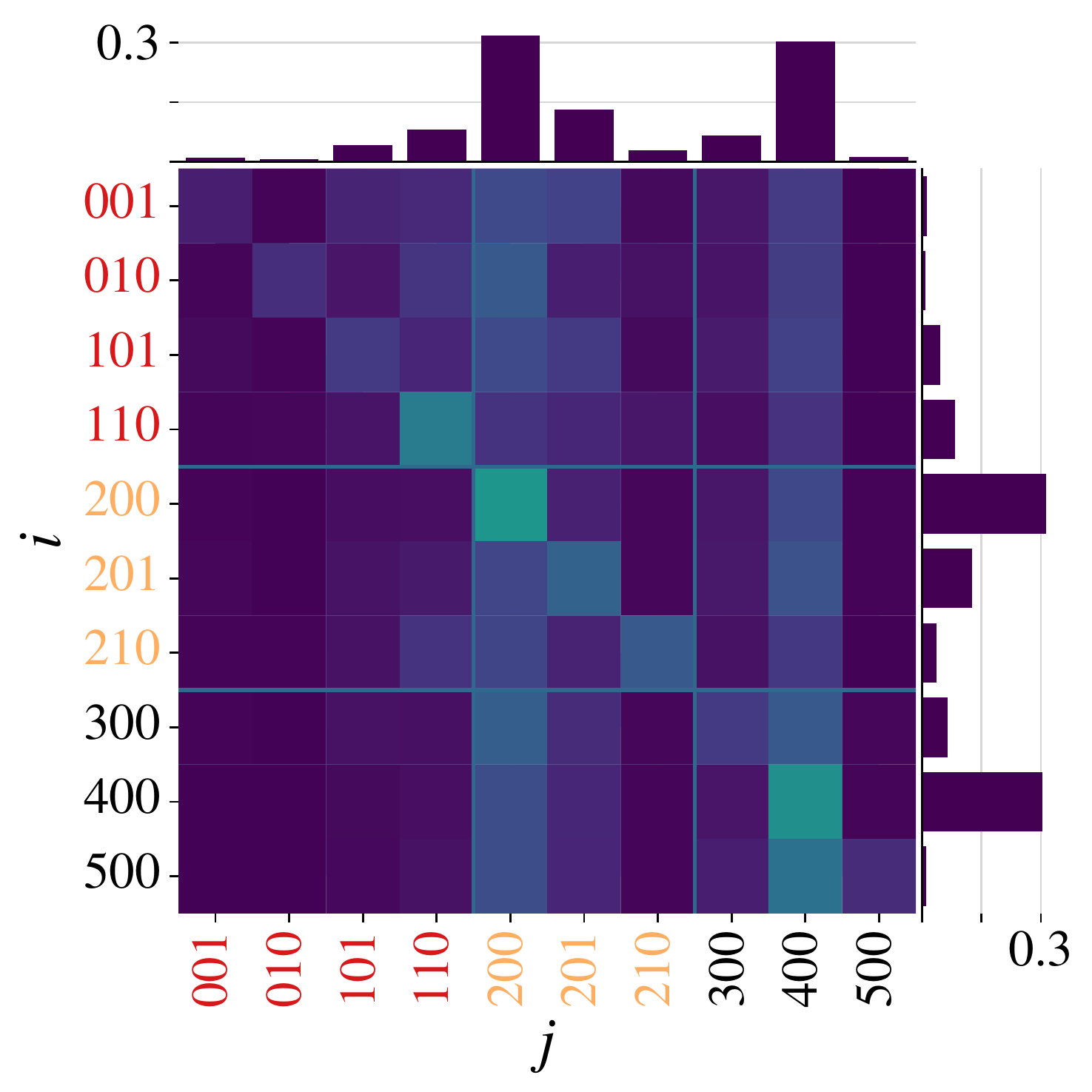}\vspace*{7pt}
		\subcaption{\nelathree}
	\end{subfigure}~%
	\hspace*{-11pt}\begin{subfigure}{0.5\linewidth}
		\centering
		\includegraphics[width=\linewidth]{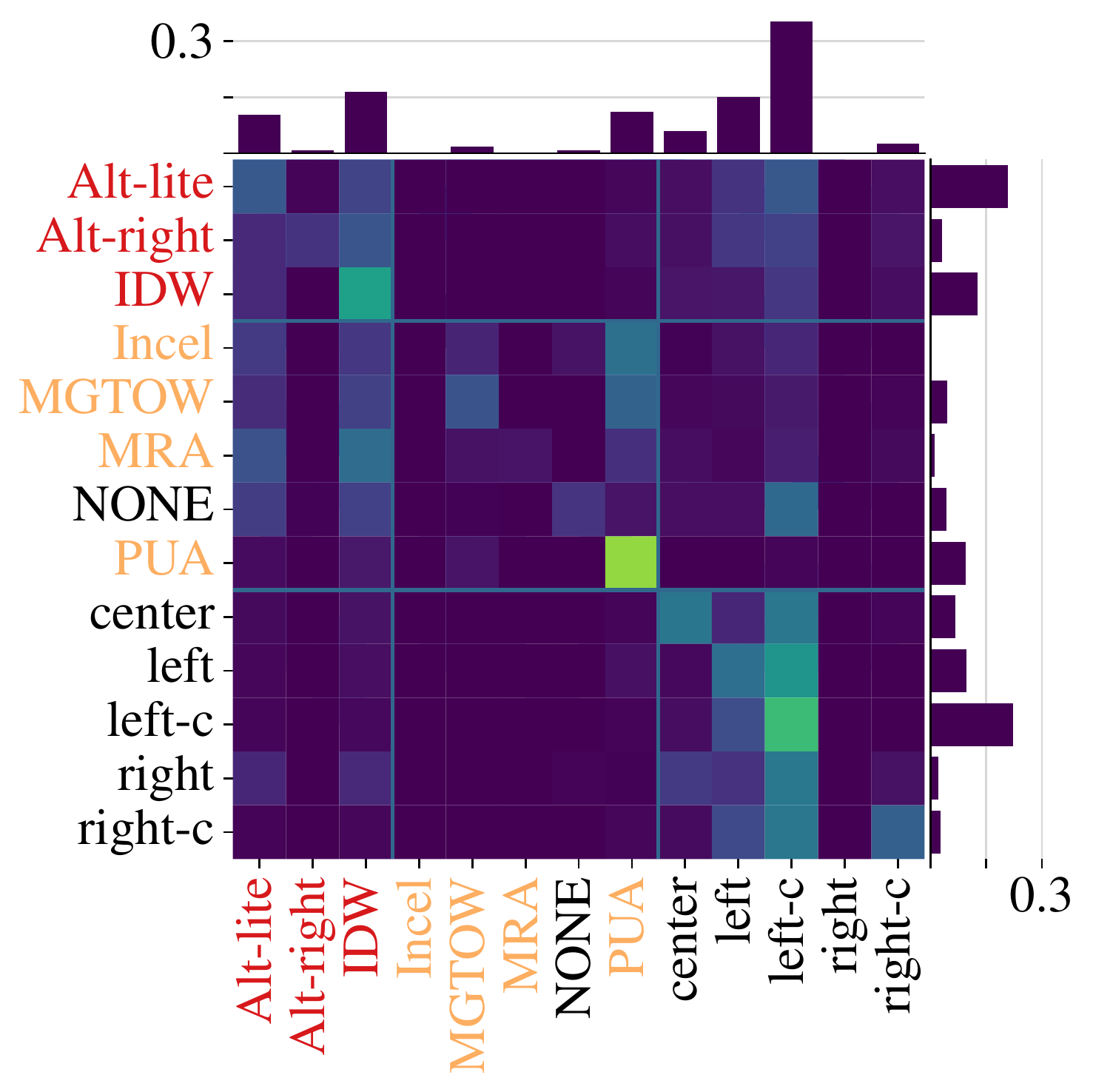}
		\subcaption{\yttwo}
	\end{subfigure}~%
	\hspace*{-7pt}\begin{subfigure}{0.078\linewidth}
		\centering
		\includegraphics[height=3cm]{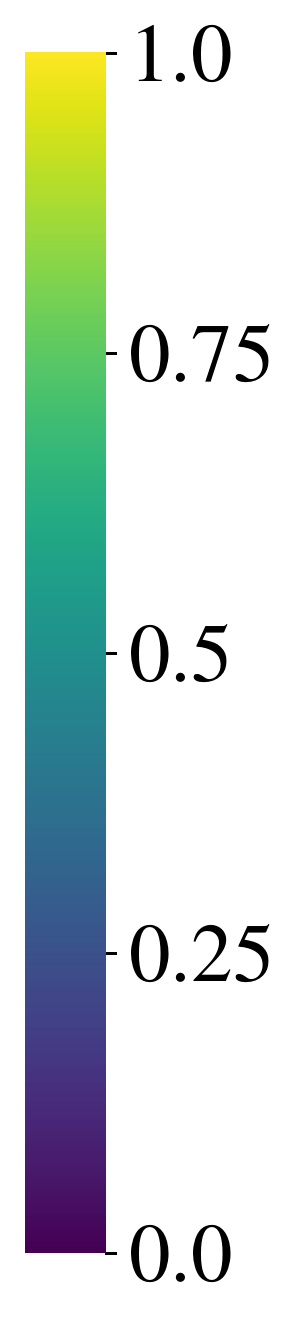}\vspace*{31.2pt}
	\end{subfigure}
	\caption{%
		Fractions of edges running between news outlet resp. video channel categories for real-world graphs with $\outregulardegree = 5$, 
		with marginals indicating the fraction of sources (right) resp. targets (top) in each category. 
		News outlet categories are denoted as triples (veracity score,
		conspiracy-pseudoscience flag, questionable-source flag); 
		for video channel categories, $\{\text{left},\text{right}\}$-center is abbreviated as $\{\text{left},\text{right}\}$-c;
		and label colors are coarse indicators of harm. 
		In \nelathree, harmful and benign nodes are more interconnected than in \yttwo.
	}
	\label{fig:channelclasses}
\end{figure}
\begin{figure}[t]
	\centering
	\begin{subfigure}{0.5\linewidth}
		\centering 
		\includegraphics[width=\linewidth]{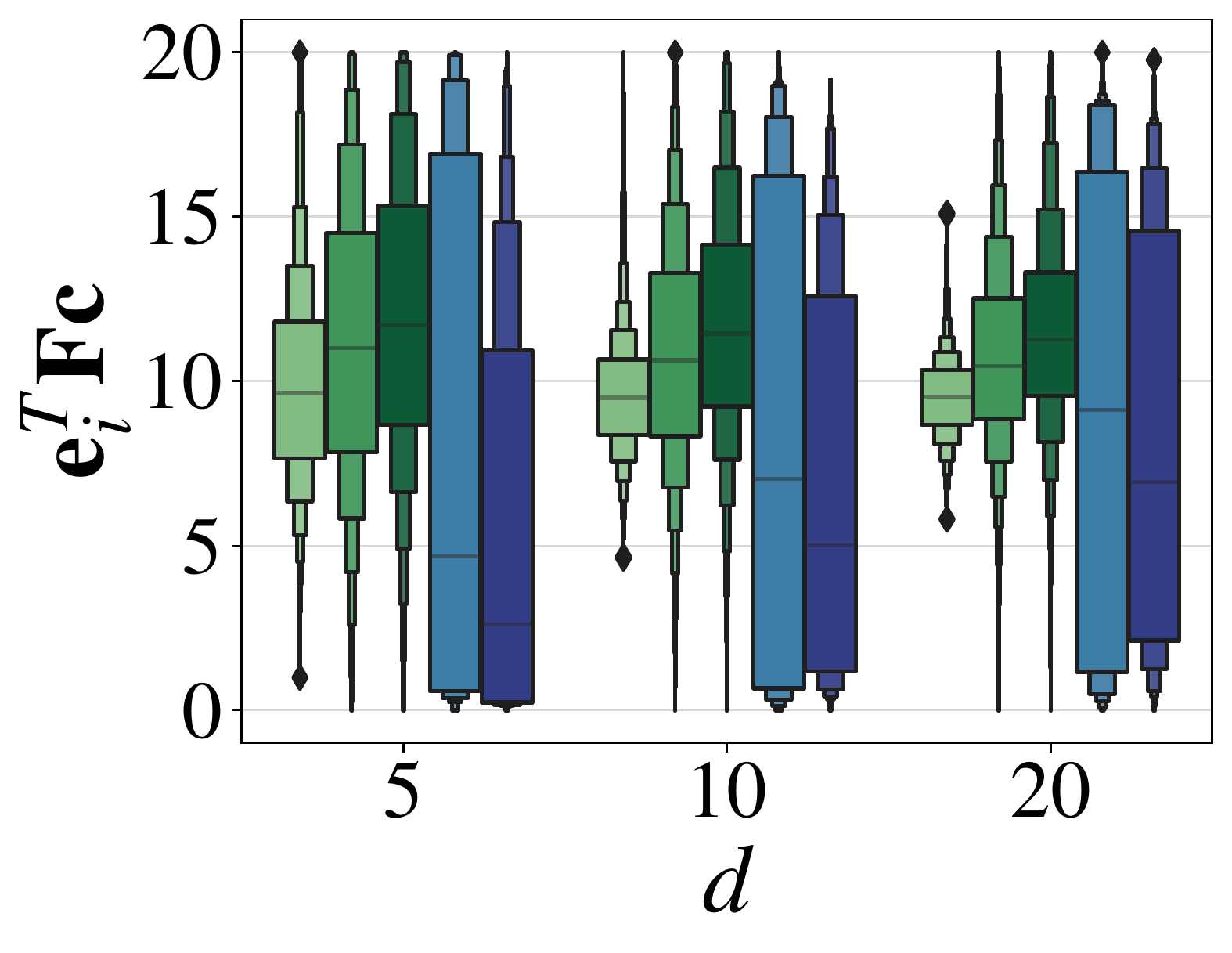}
		\subcaption{$\labeling_{B2}$, $\probabilityshape=\mathbf{S}$}
	\end{subfigure}~%
	\begin{subfigure}{0.5\linewidth}
		\centering
		\includegraphics[width=\linewidth]{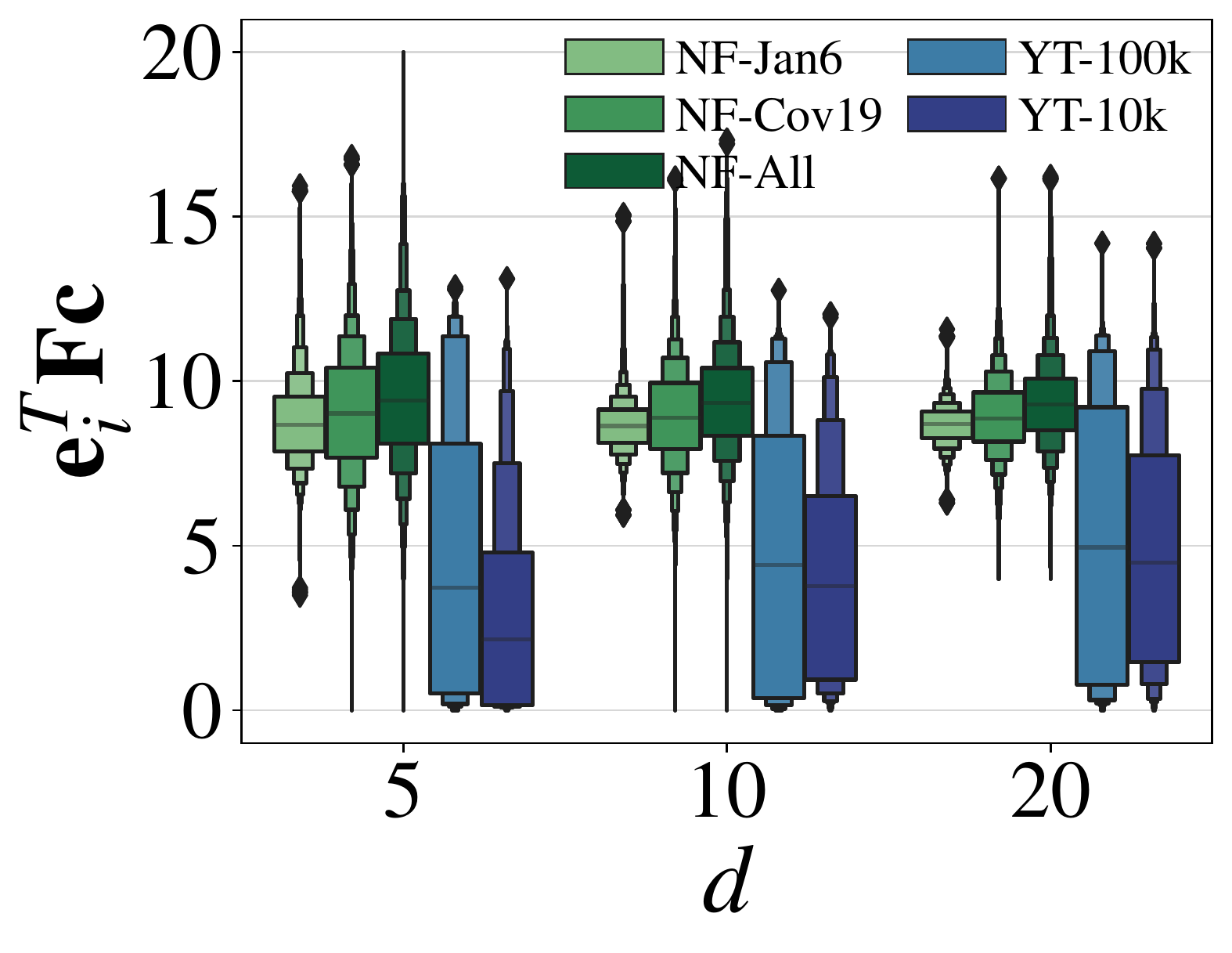}
		\subcaption{$\labeling_{R1}$, $\probabilityshape=\mathbf{U}$}
	\end{subfigure}
	\caption{%
		Distributions of initial node exposures $\unitvector^T_i\fundamental\labelingvector$ in our real-world datasets, computed with $\pabsorption=0.05$. 
		Note that cost functions sharing a name are defined differently for the \yt and \nf datasets (based on their semantics).
		The \nf datasets generally exhibit more concentrated exposure distributions than the \yt datasets and higher median exposures.
	}
	\label{fig:initialexposure}
\end{figure}
\begin{figure}[t]
	\centering
	\hspace*{-10pt}\begin{subfigure}{0.5\linewidth}
		\centering
		\includegraphics[width=1.1\linewidth]{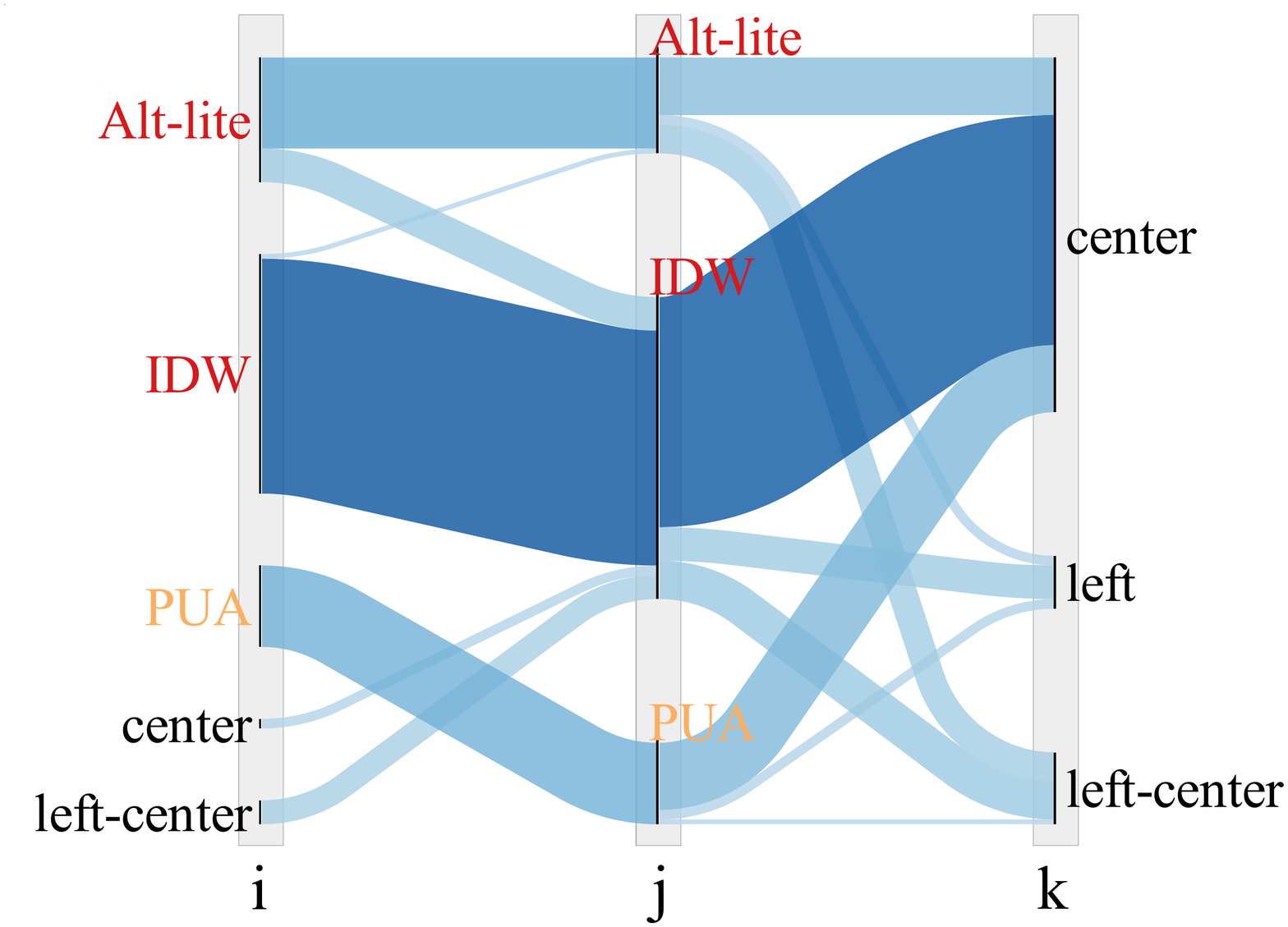
		} 
		\subcaption{Quality threshold $\qualitythreshold=0.0$}
	\end{subfigure}~%
	\begin{subfigure}{0.5\linewidth}
		\centering
		\hspace*{-3pt}\includegraphics[width=1.1\linewidth]{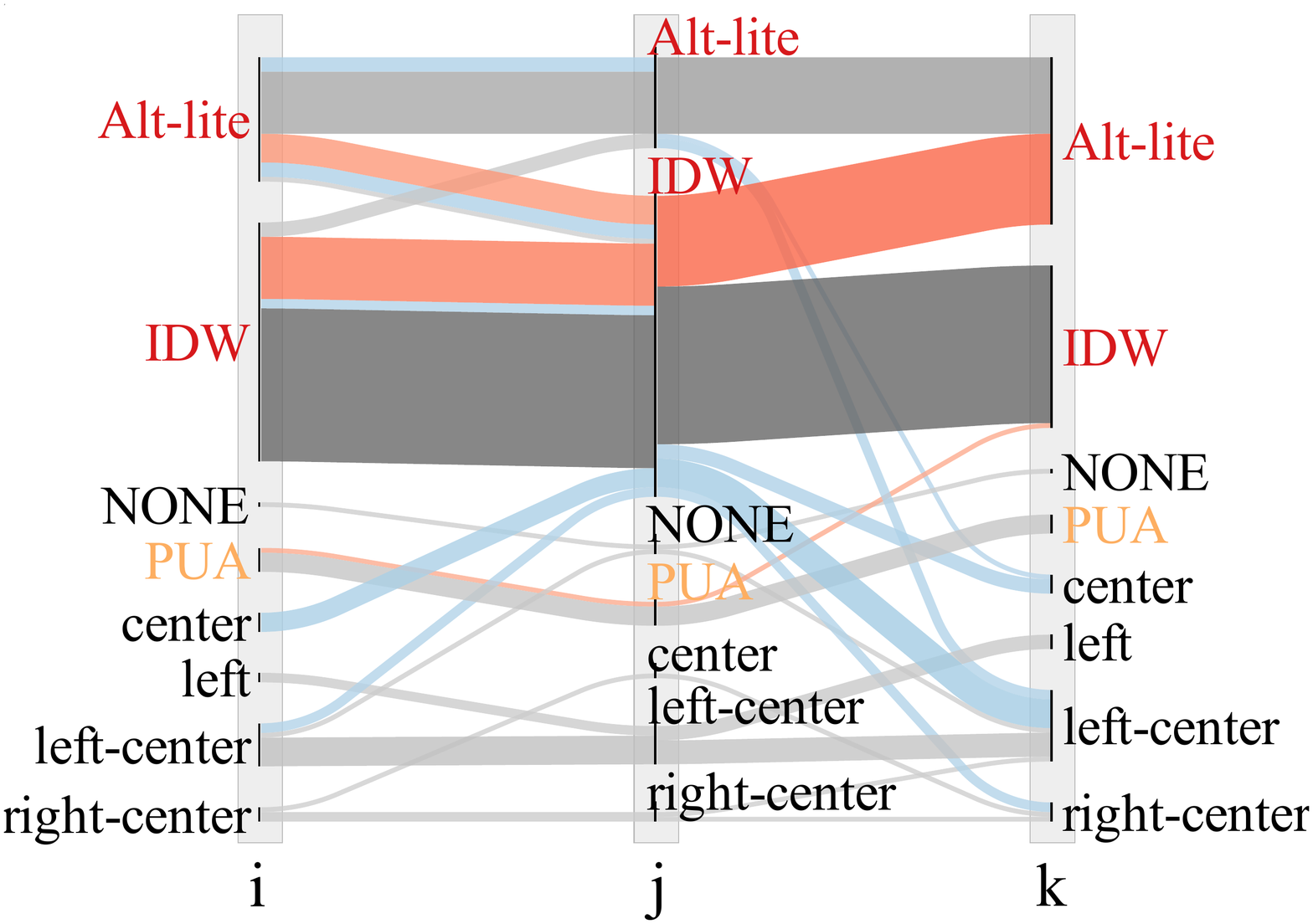
		}
		\subcaption{Quality threshold $\qualitythreshold=0.99$}
	\end{subfigure}
	\caption{%
		Channel class of videos in rewirings $\rewiring{i}{j}{k}$ on \ytone with $\outregulardegree=5$, 
		$\pabsorption=0.05$, 
		and $\probabilityshape=\mathbf{U}$, 
		computed using $\labeling_{R1}$, for different quality thresholds. 
		Rewirings between classes are color-scaled by their count, 
		using blues if $c_{R1}(k) < c_{R1}(j)$, reds if $c_{R1}(k) > c_{R1}(j)$, and grays otherwise. 
		For $\qualitythreshold = 0.0$, 
		we only replace costly targets $j$ by less costly targets $k$,
		as expected, 
		but for $\qualitythreshold = 0.99$, 
		we see many rewirings with $c_{R1}(k) \geq c_{R1}(j)$.
	}
	\label{fig:yt-rewiringchannels-R1}
\end{figure}
\begin{figure}[t]
	\centering
	\includegraphics[width=\linewidth]{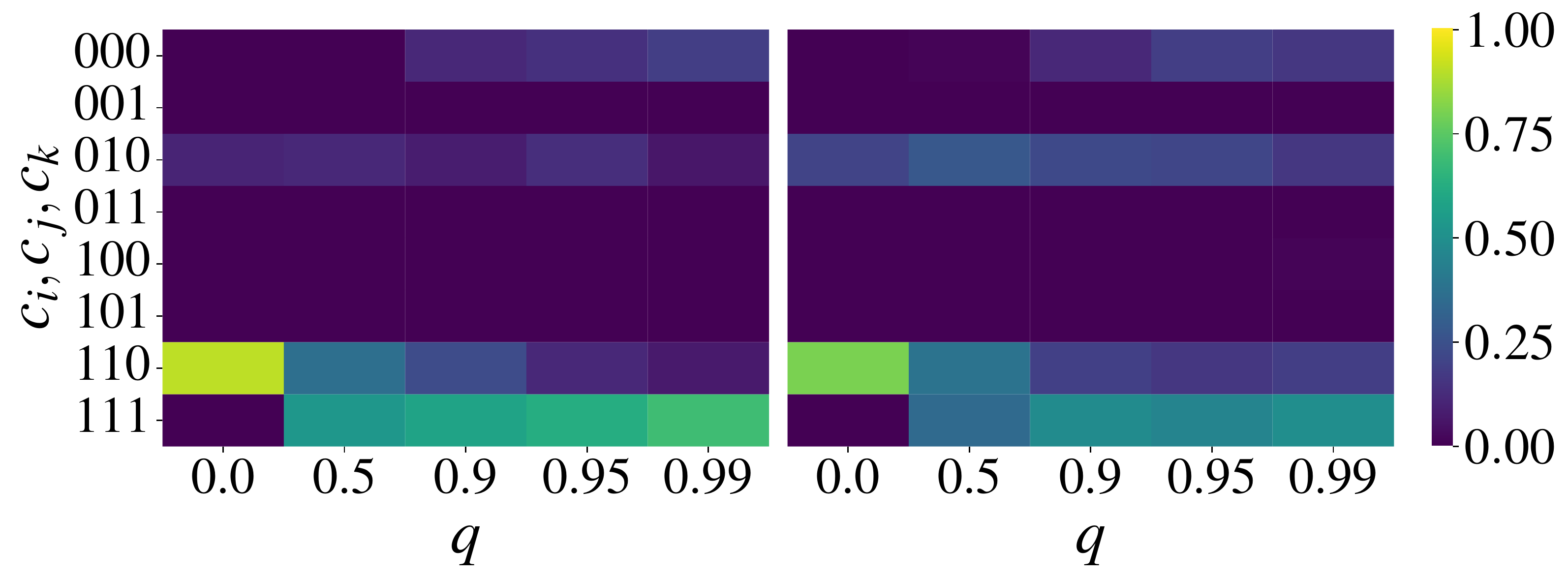}
	\begin{subfigure}{0.5\linewidth}
		\centering
		\subcaption{\ytone}
	\end{subfigure}~%
	\begin{subfigure}{0.5\linewidth}
		\centering
		\subcaption{\yttwo}
	\end{subfigure}
	\caption{%
		Mapping the nodes in each rewiring $\rewiring{i}{j}{k}$ to their costs $(\labeling_i, \labeling_j, \labeling_k)$, 
		we report the fraction of rewirings in each cost class 
		under $\labeling_{B1}$ and $\qualitythreshold\in\{0.0,0.5,0.9,0.95,0.99\}$,
		for \yt graphs with 
		$\outregulardegree=5$, 
		$\pabsorption = 0.05$, 
		and $\probabilityshape = \textbf{U}$.
		While most intuitively suboptimal classes occur rarely (e.g., 001, 011, 101), 
		under quality constraints,
		we often rewire \emph{among} harmful nodes.
	}
	\label{fig:yt-rewiringtypes-B1}
\end{figure}

\subsubsection{General guidelines}
Finally, we would like to abstract the findings from our experiments into general guidelines for reducing exposure to harm in recommendation graphs, 
especially under quality constraints.
To this end, we analyze the metadata associated with our rewirings. 
In particular, for each set of rewirings $(i,j,k)$ obtained in our experiments, 
we are interested in the channel resp. news outlet classes involved, 
as well as in the distributions of 
cost triples $(\labeling_i,\labeling_j,\labeling_k)$ 
and in-degree tuples $(\indegree{i},\indegree{j})$.
As exemplified in \cref{fig:yt-rewiringchannels-R1}, 
while we consistently rewire edges from harmful to benign targets in the quality-unconstrained setting ($\qualitythreshold=0.0$), 
under strict quality control ($\qualitythreshold=0.99$),
we frequently see rewirings from harmful to equally or more harmful targets.
More generally, as illustrated in \cref{fig:yt-rewiringtypes-B1},
the larger the threshold $\qualitythreshold$, 
the more we rewire \emph{among} harmful, resp.\ benign, nodes ($\labeling_i=\labeling_j=\labeling_k=1$, resp.\ $0$)---%
which \fabbrialg does not even allow.
Furthermore, the edges we rewire typically connect nodes with large in-degrees (\cref{apx:experiments:edgestats}, \cref{fig:indegree-sums}).
We conclude that a simplified strategy for reducing exposure to harm under quality constraints 
is to identify edges that connect high-cost nodes with large in-degrees, 
and rewire them to the node with the lowest exposure among all nodes meeting the quality constraints.

\section{Discussion and Conclusion}
\label{sec:conclusion}

We studied the problem of reducing the exposure to harmful content in recommendation graphs by edge rewiring. 
Modeling this exposure via absorbing random walks, 
we introduced \ourproblemtwo and \ourproblem as formalizations of the problem with and without quality constraints on recommendations. 
We proved that both problems are \NP-hard and \NP-hard to approximate to within an additive error, 
but that under mild assumptions, 
the greedy method provides a $(1-\nicefrac{1}{e})$-approximation for the \ourproblem problem.
Hence, we introduced \ourmethod, 
a greedy algorithm for \ourproblem and \ourproblemtwo running in linear time under realistic assumptions on the input, 
and we confirmed its effectiveness, robustness, and efficiency through extensive experiments on synthetic data 
as well as on real-world data from video recommendation and news feed applications.

Our work improves over the state of the art (\fabbrialg by \citeauthor{fabbri2022rewiring} \cite{fabbri2022rewiring}) in terms of performance, 
and it eliminates several limitations of prior work.
While \citeauthor{fabbri2022rewiring} \cite{fabbri2022rewiring} model benign nodes as \emph{absorbing} states and consider a brittle \emph{max}-objective that is minimized even by highly harm-exposing recommendation graphs, 
we model benign nodes as \emph{transient} states and consider a robust \emph{sum}-objective that captures the overall consumption of harmful content by users starting at any node in the graph. 
Whereas \fabbrialg can only handle \emph{binary} node labels, 
\ourmethod works with \emph{real-valued} node attributes, 
which permits a more nuanced encoding of harmfulness.

We see potential for future work in several directions.
For example, 
it would be interesting to adapt our objective to mitigate \emph{polarization}, 
i.e., the separation of content with opposing views, 
with positions modeled as positive and negative node costs.
Moreover, we currently assume that all nodes are equally likely as starting points of random walks, 
which is unrealistic in many applications. 
Finally, 
we observe that harm reduction in recommendation graphs has largely been studied in separation from harm reduction in other graphs representing consumption phenomena, 
such as user interaction graphs.
A framework for optimizing functions under budget constraints that includes edge rewirings, insertions, and deletions
could unify these research lines and facilitate future progress.

\vspace*{18pt}

\begin{acks}
This research is supported by 
the Academy of Finland project MLDB (325117), 
the ERC Advanced Grant REBOUND (834862), 
the EC H2020 RIA project SoBigData++ (871042), and 
the Wallenberg AI, Autonomous Systems and Software Program (WASP) funded by the Knut and Alice Wallenberg Foundation. 
\end{acks}

\clearpage

\bibliographystyle{ACM-Reference-Format}
\bibliography{bibliography.bib}

\appendix

\pdfbookmark[section]{Ethics Statement}{ethicsmark}
\section*{Ethics Statement}
\label{apx:ethics}

In this work, we introduce \ourmethod, 
a method to reduce the exposure to harm induced by recommendation algorithms on digital media platforms via edge rewiring, 
i.e., replacing certain recommendations by others. 
While removing harm-inducing recommendations constitutes a milder intervention than censoring content directly, 
it still steers attention away from certain content to other content, 
which, if pushed to the extreme, can have censorship-like effects.
Although in its intended usage, 
\ourmethod primarily counteracts the tendency of recommendation algorithms to overexpose harmful content as similar to other harmful content, 
when fed with a contrived cost function, 
it could also be used to discriminate against content considered undesirable for problematic reasons 
(e.g., due to political biases or stereotypes against minorities).
However, 
as the changes to recommendations suggested by \ourmethod could also be made by amending recommendation algorithms directly,
the risk of \emph{intentional} abuse is no greater than that inherent in the recommendation algorithms themselves, 
and \emph{unintentional} abuse can be prevented by rigorous impact assessments and cost function audits before and during deployment. 
Thus, we are confident that overall, \ourmethod can contribute to the health of digital platforms.

\pdfbookmark[section]{Appendix}{appendixmark}
\section*{Appendix}

In addition to \cref{tab:notation}, included below,
the written appendix to this work contains the following sections:
\begin{enumerate}[label=\Alph*,align=left,labelwidth=7.5pt,leftmargin=15pt]
	\item Omitted proofs
	\item Other graph edits
	\item Omitted pseudocode
	\item Reproducibility information
	\item Dataset information 
	\item Further experiments
\end{enumerate}
This appendix, along with the main paper, 
is available on arXiv and also deposited at the following DOI: 
\href{https://doi.org/10.5281/zenodo.8002980}{10.5281/zenodo.8002980}. 
To facilitate reproducibility, 
all code, data, and results are made available at the following DOI: \oururl.
\balance

\begin{table*}[t]	
	\centering\small
	\caption{Most important notation used in this work.}\label{tab:notation}\pdfbookmark[section]{Notation}{notationmark}
\begin{tabular}{r@{\hskip 3pt}c@{\hskip 3pt}ll}
	\toprule
	\bfseries Symbol & &\bfseries Definition & \bfseries Description\\
	\midrule
	\multicolumn{4}{c}{\textsc{Graph Notation}}\\
	\midrule
	$\graph$&$=$&$(\nodes,\edges)$&Graph\\
	$\nnodes$&$=$&$\cardinality{\nodes}$&Number of nodes\\
	$\nedges$&$=$&$\cardinality{\edges}$&Number of edges\\
	$\indegree{i}$&$=$&$\cardinality{\{j\mid (j,i)\in \edges\}}$&In-degree of node $i$\\
	$\outneighbors{i}$&$=$&$\{j\mid(i,j)\in\edges\}$&Set of out-neighbors of node $i$\\
	$\outdegree{i}$&$=$&$\cardinality{\outneighbors{i}}$&Out-degree of node $i$\\
	$\outregulardegree$&&&Regular out-degree of an out-regular graph\\
	$\maxoutdegree$&$=$&$\max\{\outdegree{i}\mid i\in\nodes\}$&Maximum out-degree\\
	$\safenodes$&$=$&$\{i\in\nodes \mid \unitvector_i^T\fundamental\labelingvector = 0\}$&Set of safe nodes\\
	$\unsafenodes$&$=$&$\{i\in\nodes \mid \unitvector_i^T\fundamental\labelingvector > 0\}$&Set of unsafe nodes\\
	$\maxunsafeoutdegree$&$=$&$\max\{\outdegree{i}\mid i\in\unsafenodes\}$&Maximum out-degree of an unsafe node\\
	\midrule
	\multicolumn{4}{c}{\textsc{Matrix Notation}}\\
	\midrule
	$\somematrix[i,j]$&&&Element in row $i$, column $j$ of $\somematrix$\\
	$\somematrix[i,:]$&&&Row $i$ of $\somematrix$\\
	$\somematrix[:,j]$&&&Column $j$ of $\somematrix$\\
	$\unitvector_i$&&&$i$-th unit vector\\
	$\onevector$&&&All-ones vector\\
	$\identity$&&&Identity matrix\\
	$\norm{\somematrix}_\infty$&$=$&$\max_i\sum_{j=0}^{\nnodes}\somematrix[i,j]$&Infinity norm\\
	\midrule
	\multicolumn{4}{c}{\textsc{Notation for \ourproblem and \ourproblemtwo}}\\
	\midrule
	$\rewiring{i}{j}{k}$&&&Rewiring replacing $\edge{i}{j}\in\edges$ by $\edge{i}{k}\notin\edges$ with $\probability{ik} = \probability{ij}$, cf.~\cref{tab:rewiring}\\
	$\budget$&$\in$&$\naturals$&Rewiring budget\\
	$\pabsorption$&$\in$&$(0,1]$&Random-walk absorption probability\\
	$\probability{ij}$&$\in$&$(0,1-\pabsorption]$&Probability of traversing $(i,j)$ from $i$\\
	$\transitionmatrix$&$\in$&$[0,1-\pabsorption]^{\nnodes\times\nnodes}$&Random-walk transition matrix\\
	$\fundamental$&$=$&$\sum_{i=0}^{\infty}\transitionmatrix^i = (\identity-\transitionmatrix)^{-1}$&Fundamental matrix\\
	$\labeling$&&&Cost function with range $\range$\\
	$\labeling_i$&$\in$&$\range$&Cost associated with node $i$\\
	$\labelingvector$&$\in$&$\range^{\nnodes}$&Vector of node costs\\
	$\niter$&$\in$&$\naturals$&Number of power iterations\\
	$\probabilityshape$&$\in$&$\{\mathbf{U},\mathbf{S}\}$&Shape of probability distribution over the out-edges of a node\\
	\midrule
	\multicolumn{4}{c}{\textsc{Notation for \ourproblemtwo Only}}\\
	\midrule
	$\relevancematrix$&$\in$&$\reals_{\geq 0}^{\nnodes\times\nnodes}$&Relevance matrix\\
	$\qualityf$&&&Relevance function with range $\range$\\
	$\qualitythreshold$&$\in$&$\range$&Quality threshold\\
	$\outseq{i}$&$\in$&$\nodes^{\outdegree{i}}$&Relevance-ordered targets of out-edges of $i$\\
	$\rank_i(j)$&&&Relevance rank of node $j$ for node $i$\\
	$\topranked_{\outdeg}\!(i)$&$=$&$\{j\mid \rank_i(j)\leq \outdegree{i}\}$&Set of the $\outdegree{i}$ nodes most relevant for node $i$\\
	$\dcg$&$=$&$\sum_{j\in\outneighbors{i}} \frac{\relevancematrix[i,j]}{\log_2(1 + \rank_i(j))}$&Discounted Cumulative Gain\\
	$\idcg$&$=$&$\sum_{j\in\topranked_{\outdeg}\!(i)} \frac{\relevancematrix[i,j]}{\log_2(1 + \rank_i(j))}$&Ideal Discounted Cumulative Gain\\
	$\ndcg$&$=$&$\frac{\dcg(i)}{\idcg(i)}$&Normalized Discounted Cumulative Gain\\
	\midrule
	\multicolumn{4}{c}{\textsc{Notation Related to the Exposure Function $\objectivef$ and its Analysis}}\\
	\midrule
	$\objective{\graph}$&$=$&$\onevector^T\fundamental\labeling$&Exposure function (minimization objective)\\
	$\maxobjective{\graph,\graph_\budget}$&$=$&$\objective{\graph}-\objective{\graph_\budget}$&Reduction-in-exposure function (equivalent maximization objective)\\
	$\graph'$, $\transitionmatrix'$, $\fundamental'$&&&Graph $\graph$, transition matrix $\transitionmatrix$, fundamental matrix $\fundamental$, as updated by rewiring $\rewiring{i}{j}{k}$, cf.~\cref{tab:rewiring}\\
	$\ivector$&$=$&$\probability{ij}\unitvector_i$&Vector capturing the source $i$ of a rewiring $\rewiring{i}{j}{k}$ and the traversal probability of $\edge{i}{j}$\\
	$\jkvector$&$=$&$\unitvector_j -\unitvector_k$&Vector capturing the old target $j$ and the new target $k$ of a rewiring $\rewiring{i}{j}{k}$\\
	$\sigma$&$=$&$\onevector^T\fundamental\ivector$&$\probability{ij}$-scaled $i$-th column sum\\
	$\tau$&$=$&$\jkvector^T\fundamental\labelingvector$&$\labelingvector$-scaled sum of differences between the $j$-th row sum and the $k$-th row sum\\
	$\rho$&$=$&$1+\jkvector^T\fundamental\ivector$&Normalization factor ensuring that $\fundamental'\onevector=\fundamental\onevector$\\
	$\Delta$&$=$&$\maxobjective{\graph,\graph'} = \nicefrac{\sigma\tau}{\rho}$&Reduction of $\objectivef$ obtained by a single rewiring $\rewiring{i}{j}{k}$\\
	$\heuristic$&$=$&$\Delta\rho = \sigma\tau$&Heuristic for $\Delta$\\
	\bottomrule
\end{tabular}\vspace*{-20em}
\end{table*}

\balance

\clearpage

\section*{\huge Appendix}

In this appendix, 
we present the proofs omitted in the main paper (\cref{apx:hardness}), 
discuss alternative graph edit operations (\cref{apx:edits}), 
and state the pseudocode for \ourmethod as well as the algorithms leading up to it (\cref{apx:pseudocode}).  
We also provide further reproducibility information (\cref{apx:reproducibility}), 
more details on our datasets (\cref{apx:datasets}), 
and additional experimental results (\cref{apx:experiments}).

\section{Omitted Proofs}
\label{apx:hardness}

In this section, we provide the full proofs of our hardness results for \ourproblem, which carry over to \ourproblemtwo: 
NP-hardness (\cref{thm:hardness}) and hardness of approximation (\cref{thm:apxhardness}). 
We further give the complete proofs of the submodularity of \maxobjectivef (\cref{lem:submoddupes,thm:submodularity}) 
and of the mathematical structure in $\Delta$ (\cref{lem:deltacomponents}).

\subsection{\NP-Hardness  of \ourproblem}
\label{apx:hardness:np}

\hardness*
\begin{proof}
	We reduce from minimum vertex cover for undirected cubic, i.e., 3-regular graphs (MVC-3), 
	which is known to be NP-hard \cite{greenlaw1995cubic}. 
	To this end, we transform an instance of MVC-3 into an instance of \ourproblem with a directed, 3-out-regular input graph (\ourproblem-3 instance) as follows.
	From a cubic undirected graph $\graph' = (\nodes', \edges')$ with $\nnodes' = \cardinality{\nodes'}$ and $\nedges' = \cardinality{\edges'} = \nicefrac{3\nnodes'}{2}$, 
	we construct our directed \ourproblem-3 instance 
	by defining a graph $\graph = (\nodes, \edges)$ with $\nnodes = \cardinality{\nodes} = 2\nnodes' + 4$ nodes and $\nedges = \cardinality{\edges} = 6\nnodes' + 12$ edges such that
	\begin{align*}
		\nodes 
		=~& \nodes' 
		\cup \overline{\nodes'} 
		\cup \safenodes\;,~\text{for}~\overline{\nodes'} = \{b_i\mid i\in\nodes'\}\;,~\safenodes = \{g_1, g_2, g_3, g_4\}\;,\\
		\edges 
		=~&\{(i,b_i) \mid i \in \nodes'\}
		\cup \{(b_i,j)\mid \{i,j\}\in \edges'\} \\
		&
		\cup \{(g_i,g_j)\in \safenodes\times \safenodes\mid i\neq j\}
		\cup \{(i,g_x)\mid i \in \nodes', x\in \{1,2\}\}
		\;,\\
		\transitionmatrix[x,y] =~& \frac{1 - \pabsorption}{\outdegree{x}} = \frac{1-\pabsorption}{3}\;,~\text{and}~\labeling_x = \begin{cases}
			1&x\in \overline{\nodes'} \\
			0&\text{otherwise}\;.
		\end{cases}
	\end{align*}
	That is, for each node $i\in\nodes'$, we introduce a node $i\in\nodes$ with $\labeling_i = 0$, a companion node $b_i\in\nodes$ with $\labeling_{b_i} = 1$, and an edge $(i,b_i)\in\edges$ in~$\graph$. 
	We then encode the original edge set implicitly by defining two edges $(b_i,j)\in\edges$ and $(b_j,i)\in \edges$ for each edge $\{i,j\} \in \edges'$. 
	Finally, we add a complete 3-out-regular graph of zero-cost nodes and connect each node representing a node from $\nodes'$ to the first two nodes of that graph. 
	
	Intuitively,
	the edges $\{(i,b_i)\mid i\in\nodes'\}$ will be our prime candidates for rewiring---%
	and rewiring an edge $(i,b_i)$ in \ourproblem-3 will correspond to selecting node $i$ into the vertex cover of the original MVC-3 instance.
	The implicit encoding of the original edge set introduces the asymmetry necessary to tell from the value of our objective function 
	if an optimal $\budget$-rewiring of $\graph$ corresponds to a vertex cover of cardinality $\budget$ in $\graph'$.
	Adding a complete 3-out-regular graph of zero-cost nodes gives us a strongly connected safe component $\safenodes$ of nodes as rewiring targets, and it ensures that $\graph$ is 3-out-regular. 
	The entire transformation is visualized in \cref{fig:hardness}.
	
		\begin{figure}[t]
		\begin{subfigure}{\linewidth}
			\centering

\begin{tikzpicture}[thick]
	\node[minimum height = 2em,minimum width = 2em,draw,circle] at (0, 0)   (i) {$i$};
	\node[minimum height = 2em,minimum width = 2em,draw,circle] at (1.5, 0)   (j) {$j$};
	\node[minimum height = 2em,minimum width = 2em,draw,circle] at (2.25, -1)   (k) {$k$};
	\node[minimum height = 2em,minimum width = 2em,draw,circle] at (1.5, -2)   (x) {$x$};
	\node[minimum height = 2em,minimum width = 2em,draw,circle] at (0, -2)   (y) {$y$};
	\node[minimum height = 2em,minimum width = 2em,draw,circle] at (-0.75, -1)   (z) {$z$};
	
	\draw (i) -- (j) -- (k) -- (z) -- (x) -- (y) ;
	\draw (j) -- (y) -- (z); 
	\draw (k) -- (i) -- (x);
\end{tikzpicture}
			\vspace*{4.5pt}
			\subcaption{Toy MVC-3 instance $\graph' = (\nodes',\edges')$}\label{fig:toygraph}\vspace*{1em}
			\vspace*{-9pt}\end{subfigure}
		\begin{subfigure}{\linewidth}
			\centering

\begin{tikzpicture}[->,>=stealth',thick]
	\node[minimum height = 2em,minimum width = 2em,draw,circle,fill=black!10] at (1.25, 0)  (g1) {$g_1$};
	\node[minimum height = 2em,minimum width = 2em,circle,fill=black!10,draw=black!20] at (2.5, 0)  (g2) {$g_4$};
	\node[minimum height = 2em,minimum width = 2em,circle,fill=black!10,draw=black!20] at (3.75, 0)  (g3) {$g_3$};
	\node[minimum height = 2em,minimum width = 2em,draw,circle,fill=black!10] at (5, 0)  (g4) {$g_2$};
	
	\node[minimum height = 2em,minimum width = 2em] at (7.5, 0)  (S) {$\safenodes$};
	\node[minimum height = 2em,minimum width = 2em] at (7.5, -2)  (S) {$\nodes'$};
	\node[minimum height = 2em,minimum width = 2em] at (7.5, -4)  (S) {$\overline{\nodes'}$};
	
	\node[minimum height = 2em,minimum width = 2em,draw,circle] at (0, -2)   (i) {$i$};
	\node[minimum height = 2.5em,minimum width = 2em,draw,circle,fill=silver] at (0, -4)   (bi) {$b_i$};
	
	\node[minimum height = 2em,minimum width = 2em,draw,circle] at (1.25, -2)   (j) {$j$};
	\node[minimum height = 2.5em,minimum width = 2em,draw,circle,fill=silver] at (1.25, -4)   (bj) {$b_j$};
	
	\node[minimum height = 2em,minimum width = 2em,draw,circle] at (2.5, -2)   (k) {$k$};
	\node[minimum height = 2.5em,minimum width = 2em,draw,circle,fill=silver] at (2.5, -4)   (bk) {$b_k$};
	
	\node[minimum height = 2em,minimum width = 2em,draw,circle] at (3.75, -2)   (x) {$x$};
	\node[minimum height = 2.5em,minimum width = 2em,draw,circle,fill=silver] at (3.75, -4)   (bx) {$b_x$};
	
	\node[minimum height = 2em,minimum width = 2em,draw,circle] at (5, -2)   (y) {$y$};
	\node[minimum height = 2.5em,minimum width = 2em,draw,circle,fill=silver] at (5, -4)   (by) {$b_y$};
	
	\node[minimum height = 2em,minimum width = 2em,draw,circle] at (6.25, -2)   (z) {$z$};
	\node[minimum height = 2.5em,minimum width = 2em,draw,circle,fill=silver] at (6.25, -4)   (bz) {$b_z$};

	\draw[red] (i) -- (bi); 
	\draw[red] (j) -- (bj); 
	\draw[red] (k) -- (bk);
	\draw[red] (x) -- (bx);
	\draw[red] (y) -- (by);
	\draw[red] (z) -- (bz);
	
	\draw (bi) -- (j);
	\draw (bi) -- (k);
	\draw (bi) -- (x);
	
	\draw (bj) -- (i);
	\draw (bj) -- (k);
	\draw (bj) -- (y);
	
	\draw (bk) -- (i);
	\draw (bk) -- (j);
	\draw (bk) -- (z);
	
	\draw (bx) -- (i);
	\draw (bx) -- (y);
	\draw (bx) -- (z);
	
	\draw (by) -- (j);
	\draw (by) -- (x);
	\draw (by) -- (z);
	
	\draw (bz) -- (k);
	\draw (bz) -- (x);
	\draw (bz) -- (y);
	
	\draw[black!20] (i) -- (g1);
	\draw[black!20] (i) -- (g4);
	\draw[black!20] (j) -- (g1);
	\draw[black!20] (j) -- (g4);
	\draw[black!20] (k) -- (g1);
	\draw[black!20] (k) -- (g4);
	\draw[black!20] (x) -- (g1);
	\draw[black!20] (x) -- (g4);
	\draw[black!20] (y) -- (g1);
	\draw[black!20] (y) -- (g4);
	\draw[black!20] (z) -- (g1);
	\draw[black!20] (z) -- (g4);
	
	\draw (g1)[black!20] to [out=350,in=190] (g2);
	\draw (g1)[black!20] to [out=330,in=210] (g3);
	\draw (g1) to [out=310,in=230] (g4);
	\draw (g2)[black!20] to [out=170,in=10] (g1);
	\draw (g2)[black!20] to [out=350,in=190] (g3);
	\draw (g2)[black!20] to [out=330,in=210] (g4);
	\draw (g3)[black!20] to [out=150,in=30] (g1);
	\draw (g3)[black!20] to [out=170,in=10] (g2);
	\draw (g3)[black!20] to [out=350,in=190] (g4);
	\draw (g4)[black!20] to [out=170,in=10] (g3);
	\draw (g4)[black!20] to [out=150,in=30] (g2);
	\draw (g4) to [out=130,in=50] (g1);

\end{tikzpicture}
			\vspace*{4.5pt}
			\subcaption{\ourproblem-3 instance constructed from $\graph'$}\label{fig:reduction}
		\end{subfigure}
		\caption{%
			Reduction setup for \cref{thm:hardness}. 
			\cref{fig:toygraph} depicts a toy MVC-3 instance, 
			which we transform into a \ourproblem-3 instance as shown in \cref{fig:reduction}.  
			In \cref{fig:reduction}, white nodes represent $\nodes'$, 
			gray nodes represent $\overline{\nodes'}$, 
			silver nodes represent $\safenodes$, 
			and edges $ab$ with $c_a = 0$ and $c_b = 1$ are drawn in red.
			Silver edges and nodes with silver boundaries are needed to ensure that $\graph$ is 3-out-regular, 
			and all edges are traversed with probability $\frac{(1-\pabsorption)}{3}$.
		}\label{fig:hardness}
	\end{figure}
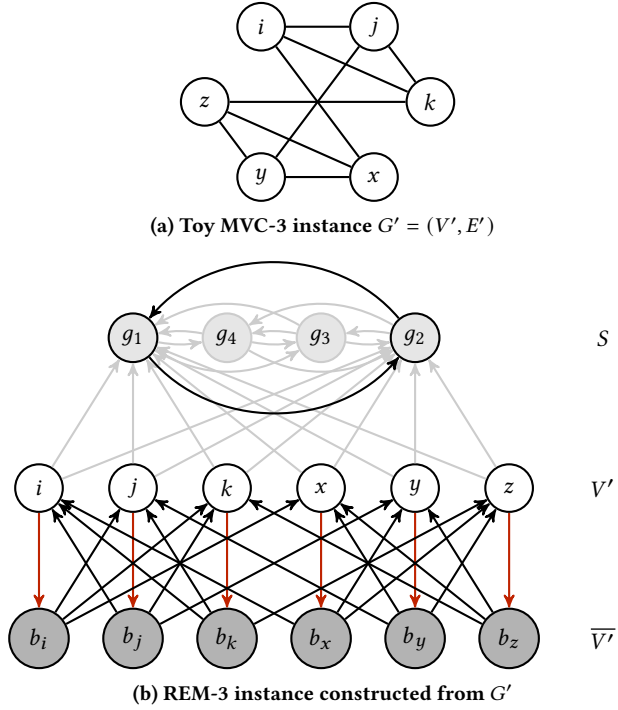	
	
	In the graph $\graph$ thus constructed, 
	the only nodes ever exposed to harm are the $\nnodes'$ nodes in $\nodes'$ and the $\nnodes'$ nodes in $\overline{\nodes'}$. 
	As illustrated in \cref{fig:view-trees}, 
	random walks starting from a node in $\overline{\nodes'}$ only see nodes with cost $1$ after an \emph{even} number of steps, 
	random walks starting from a node in $\nodes'$ only see nodes with cost $1$ after an \emph{odd} number of steps, 
	and as $\graph$ is 3-out-regular, 
	\emph{all} random walks have a branching factor of $3$, 
	such that they see exactly $3^\distance$ (not necessarily distinct) nodes at $\distance$ steps from their origin. 
	Each node in $\nodes'$ has three out-neighbors, and before the first rewiring, 
	exactly one of them is a node with cost $1$. 
	Thus if the random walks do not get absorbed,
	the regularity in our construction
	implies that the probability of encountering a node with cost $1$ after $2$ steps from a node in $\overline{\nodes'}$ is $\frac{3^1}{3^2} = \frac{1}{3}$, 
	just like the probability of encountering a node with cost $1$ after $3$ steps from a node in $\nodes'$ is $\frac{3^1}{3^3} = \frac{1}{9}$.
	Therefore, the starting value of our objective function can be written succinctly as 
	\begin{align}\label{eq:objective-initial}
		\objective{\graph} 
		= 
		\underbrace{\nnodes'\sum_{\distance=0}^{\infty}
			3^{-\distance}(1-\pabsorption)^{2\distance}}_{\text{Contributions from $\overline{\nodes'}$}}
		+
		\underbrace{\nnodes'\sum_{\distance=0}^{\infty}
			3^{-\distance-1}(1-\pabsorption)^{2\distance+1}}_{\text{Contributions from $\nodes'$}}
		\;.
	\end{align}

\begin{figure}[t]
	\centering
	\begin{subfigure}{\linewidth}
		\centering
	\tiny
	\begin{tikzpicture}[level distance=1cm,
	every node/.style = {shape=circle, 
		draw, align=center,thick,minimum width=1em, minimum height=1em},
	edge from parent/.style={draw,->,>=stealth',thick},
	level 1/.style={sibling distance=0.09\linewidth},
	level 2/.style={sibling distance=0.27\linewidth},
	level 3/.style={sibling distance=0.09\linewidth},
	level 4/.style={sibling distance=0.09\linewidth}
	]
	\node {$i$}
	child {  node[fill=silver] {$b_i$}
		child {  node {$j$}
			child { node[fill=silver] {$b_j$}
				child {node {$i$}}
				child {node {$k$}}
				child {node {$y$}}
			}
			child {node[shape=rectangle,fill=black!10] {$S$}
			}
			child {node[shape=rectangle,fill=black!10] {$S$}
			}
		}
		child {node {$k$}
			child {node[fill=silver] {$b_k$}
				child {node {$i$}}
				child {node {$j$}}
				child {node {$z$}}
			}
			child {node[shape=rectangle,fill=black!10] {$S$}
			}
			child {node[shape=rectangle,fill=black!10] {$S$}
			}
		}
		child {node {$y$}
			child {node[fill=silver] {$b_y$}
				child {node {$j$}}
				child {node {$x$}}
				child {node {$z$}}
			}
			child {node[shape=rectangle,fill=black!10] {$S$}
			}
			child {node[shape=rectangle,fill=black!10] {$S$}
			}
		}
	}
	child {node[shape=rectangle,fill=black!10] {$S$}
	}
	child {node[shape=rectangle,fill=black!10] {$S$}
	}
;
	\node[draw=none] at (3.25,0) {$\phantom{^1}0/3^0$};
	\node[draw=none] at (3.25,-1) {$3^0/3^1$};
	\node[draw=none] at (3.25,-2) {$0/3^2$};
	\node[draw=none] at (3.25,-3) {$3^1/3^3$};
	\node[draw=none] at (3.25,-4) {$\phantom{^1}0/3^4$};
\end{tikzpicture}
		\vspace*{1em}
		\subcaption{Walks starting at node $i \in \nodes'$}
	\end{subfigure}
	\begin{subfigure}{\linewidth}
		\centering
	\tiny
\begin{tikzpicture}[level distance=1cm,
	every node/.style = {shape=circle, 
		draw, align=center,thick,minimum width=1em, minimum height=1em},
	edge from parent/.style={draw,->,>=stealth',thick},
	level 1/.style={sibling distance=0.27\linewidth},
	level 2/.style={sibling distance=0.09\linewidth},
	level 3/.style={sibling distance=0.09\linewidth},
	level 4/.style={sibling distance=0.09\linewidth}
	]
	\node[fill=silver] {$b_i$}
		child {  node {$j$}
			child { node[fill=silver] {$b_j$}
				child {node {$i$}
					child {node[fill=silver] {$b_i$}}
				}
				child {node {$k$}
					child {node[fill=silver] {$b_k$}}
				}
				child {node {$y$}
					child {node[fill=silver] {$b_y$}}
				}
			}
			child {node[shape=rectangle,fill=black!10] {$S$}
			}
			child {node[shape=rectangle,fill=black!10] {$S$}
			}
		}
		child {node {$k$}
			child {node[fill=silver] {$b_k$}
				child {node {$i$}
					child {node[fill=silver] {$b_i$}}
				}
				child {node {$j$}
					child {node[fill=silver] {$b_j$}}
				}
				child {node {$z$}
					child {node[fill=silver] {$b_z$}}
				}
			}
			child {node[shape=rectangle,fill=black!10] {$S$}
			}
			child {node[shape=rectangle,fill=black!10] {$S$}
			}
		}
		child {node {$y$}
			child {node[fill=silver] {$b_y$}
				child {node {$j$}
					child {node[fill=silver] {$b_j$}}
				}
				child {node {$x$}
					child {node[fill=silver] {$b_x$}}
				}
				child {node {$z$}
					child {node[fill=silver] {$b_z$}}
				}
			}
			child {node[shape=rectangle,fill=black!10] {$S$}
			}
			child {node[shape=rectangle,fill=black!10] {$S$}
			}
		}
	;
	\node[draw=none] at (3.9,0) {$3^0/3^0$};
	\node[draw=none] at (3.9,-1) {$\phantom{^1}0/3^1$};
	\node[draw=none] at (3.9,-2) {$3^1/3^2$};
	\node[draw=none] at (3.9,-3) {$\phantom{^1}0/3^3$};
	\node[draw=none] at (3.9,-4) {$3^2/3^4$};
\end{tikzpicture}
		\vspace*{1em}
		\subcaption{Walks starting at node $b_i \in \overline{\nodes'}$ (safe nodes at step $4$ not shown)}
	\end{subfigure}
	\caption{%
		Random walks in 3-out-regular directed graphs $\graph = (\nodes,\edges)$ constructed from undirected MVC-3 instances as depicted in \cref{fig:hardness}. 
		All edges are traversed with probability $\frac{1-\pabsorption}{3}$,
		nodes in $\overline{\nodes'}$ are drawn in gray,
		and branches leading into the safe component $\safenodes$ are collated into silver square boxes labeled $\safenodes$.
		The annotations at level $\distance$ to the right of each random-walk tree indicate the fraction of nodes with cost $1$ of all nodes encountered after taking exactly $\distance$ steps.
	}
	\label{fig:view-trees}
\end{figure}
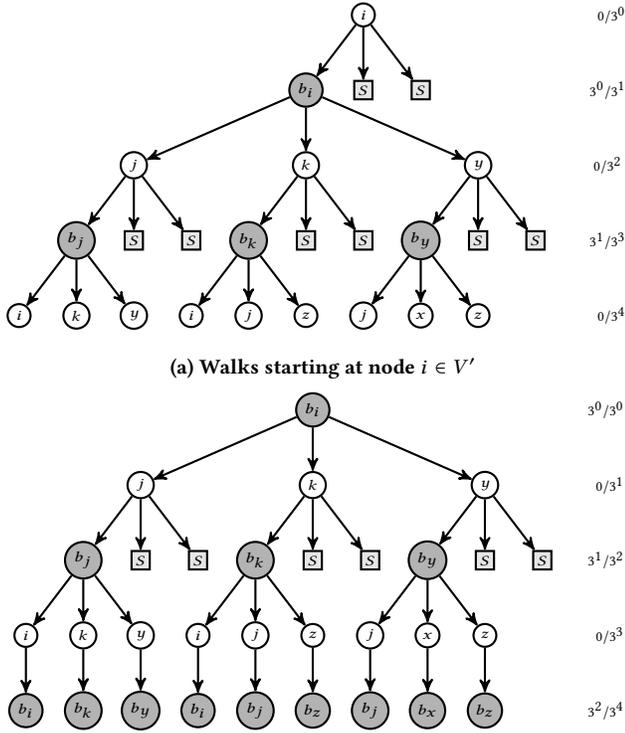	
	
	Since $\safenodes$ contains four safe nodes and $\graph$ is 3-out-regular, we can always rewire edges with unsafe targets to safe targets without creating multi-edges. 
	Therefore,
	as long as $\budget\leq \nnodes'$, an optimal rewiring $\selection$ will contain triples of shape $\rewiring{i}{b_i}{g_x}$, 
	where $g_x$ is any node that is safe \emph{after} the $\budget$ rewirings have been performed (this includes the nodes in  $\safenodes$ but can also include other zero-cost nodes~$i$ for which $\rewiring{i}{b_i}{g_{x'}}$ is part of the rewiring for some other, safe node~$g_{x'}$). 
	Now, each individual rewiring reduces the objective by 
	\begin{align}
		\underbrace{\frac{1}{3}(1-\pabsorption)}_{(1)} + \underbrace{\frac{3}{9}(1-\pabsorption)^2}_{(2)} + \underbrace{\frac{2\mcovered'}{27}(1-\pabsorption)^3}_{(3)}
		+~\smallterm'
		\;,\label{eq:objective-reduction}
	\end{align}
	where $\smallterm'$ is a term summarizing all contributions from walks longer than $3$ steps,
	and $\mcovered'$ is the number of edges that are \emph{newly} covered in $\graph'$ by selecting the source node $i$ of our rewiring in $\graph$ into the vertex cover $\cover$ of $\graph'$, 
	i.e., 
	\begin{align}
		\mcovered' = 
		\cardinality{\{e\in\edges' \mid \cover \cap e = \emptyset\}}~~
		\text{for}~~\cover = \{i\in\nodes'\mid (i,b_i)\in \edges_{\selection}\}\;, 
	\end{align}
	where $\edges_\selection$ is the set of previously rewired edges in $\graph$. 
	
	More elaborately, 
	in \cref{eq:objective-reduction},
	the component marked (1) is the exposure of $i$ to $b_i$ at distance $1$ via the walk $(i,b_i)$, 
	and the component marked (2) is the exposure of $b_j$ to $b_i$ at distance $2$ via the walk $(b_j,i,b_i)$, for the three nodes $j$ such that $\{i,j\}\in\edges'$, 
	where $\cardinality{\{\{x,y\}\in\edges'\mid y = j\}} = 3$ because $\graph'$ is 3-regular.
	The component marked (3) is the sum of 
	(i) the exposure of $i$ to nodes $b_j$ with $\{i,j\}\in\edges'$ at distance $3$ via the walk $(i,b_i,j,b_j)$,
	and (ii) the exposure of nodes $j$ with $\{i,j\}\in\edges'$ to node $i$ at distance $3$ via the walk $(j,b_j,i,b_i)$,
	each of which is 
	\begin{align*}
		\frac{1}{27}(1-\pabsorption)^3\cdot (1-\cardinality{\edges_\selection\cap \{\edge{j}{b_j}\}}) = \begin{cases}
			\frac{1}{27}(1-\pabsorption)^3&\text{if}~\edge{j}{b_j}\notin \edges_\selection\\
			0&\text{otherwise}\;.
		\end{cases}
	\end{align*}
	Hence, the objective function reduces by $\frac{2}{27}(1-\pabsorption)^3$ 
	for each edge $\{i,j\}\in\edges'$ that is covered for the \emph{first} time when we select $i$ into the vertex cover of $\graph'$, 
	and  because $\graph'$ is 3-regular, each rewiring can cover at most 3 new edges, such that $\mcovered'\leq 3$. 
	
	Thus,
	an optimal $\budget$-rewiring of $\graph$ reduces the objective function in \cref{eq:objective-initial} by 
	\begin{align}
		\maxobjective{\graph,\graph_\budget} = \frac{\budget}{3}(1-\pabsorption) 
		+ \frac{3\budget}{9}(1-\pabsorption)^2 
		+ \frac{2\mcovered}{27}(1-\pabsorption)^3
		+ \smallterm\;,
	\end{align}
	where $\mcovered\geq \frac{3}{2}\budget$ is the number of edges in $\graph'$ that are covered by the source nodes of our rewirings,
	and $\smallterm$ is the sum of the small terms $\smallterm'$ associated with each rewiring.
	Therefore, $\graph'$ has a minimum vertex cover of size at most $\budget$ 
	if and only if an optimal $\budget$-rewiring of $\graph$ reduces our objective by
	\begin{align}\label{eq:mvcreduction}
		\maxobjective{\graph,\graph_{\budget}} = \frac{\budget}{3}(1-\pabsorption) 
		+ \frac{3\budget}{9}(1-\pabsorption)^2 
		+ \frac{2\nedges'}{27}(1-\pabsorption)^3
		+ \smallterm\;,
	\end{align}
	i.e., $\graph'$ has a minimum vertex cover of size at most $\budget$  if and only if 
	\begin{align}\nonumber
		\objective{\graph_\budget}
		=~& \objective{\graph} - \maxobjective{\graph,\graph_{\budget}}
		 \\\nonumber
		=~& \nnodes' + \frac{\nnodes'-\budget}{3}(1-\pabsorption)
		+ \frac{3(\nnodes'-\budget)}{9}(1-\pabsorption)^2
		+\frac{3\nnodes'- 2\nedges'}{27}(1-\pabsorption)^3\\\nonumber
		&+ (\neglectedterm-\smallterm)
		\\\nonumber
		=~&
		\nnodes' + \frac{\nnodes'-\budget}{3}(1-\pabsorption)
		+ \frac{3(\nnodes'-\budget)}{9}(1-\pabsorption)^2
		+\frac{3\nnodes'-2\frac{3\nnodes'}{2}}{27}(1-\pabsorption)^3\\\nonumber
		&+ (\neglectedterm-\smallterm)\\\label{eq:mvcreductionf}
		=~&
		\nnodes' + \frac{\nnodes'-\budget}{3}(1-\pabsorption)
		+ \frac{3(\nnodes'-\budget)}{9}(1-\pabsorption)^2
		+ (\neglectedterm-\smallterm)
		\;,
	\end{align}
	where $\neglectedterm\geq\smallterm$ is the \emph{entire} exposure of random walks in~$G$ due to nodes encountered after four or more steps, i.e., 
	\begin{align}\label{eq:neglectedterm}
		\neglectedterm &= \underbrace{\nnodes'\sum_{\distance=0}^{\infty}
			3^{-\distance}(1-\pabsorption)^{2\distance}}_{\text{\emph{All} contributions from $\overline{\nodes'}$}}
		+
		\underbrace{\nnodes'\sum_{\distance=0}^{\infty}
			3^{-\distance-1}(1-\pabsorption)^{2\distance+1}}_{\text{\emph{All} contributions from $\nodes'$}}
		\\\nonumber
		&\quad
		- \underbrace{\nnodes'\sum_{\distance=0}^{1}
			3^{-\distance}(1-\pabsorption)^{2\distance}}_{\text{At most $3$ steps from $\overline{\nodes'}$}}
		- \underbrace{\nnodes'\sum_{\distance=0}^{1}
			3^{-\distance-1}(1-\pabsorption)^{2\distance+1}}_{\text{At most $3$ steps from $\nodes'$}}\;.
	\end{align}

	As we do not know $\smallterm$ exactly, 
	we cannot check \cref{eq:mvcreductionf} directly 
	to decide whether $\graph'$ has a vertex cover of size at most $\budget$.
	Instead, we would like to check if
	\begin{align}\label{eq:ourcheck}
		\maxobjective{\graph,\graph_{\budget}} \geq \frac{\budget}{3}(1-\pabsorption) 
		+ \frac{3\budget}{9}(1-\pabsorption)^2 
		+ \frac{2\nedges'}{27}(1-\pabsorption)^3\;,
	\end{align}
	that is, for the purposes of our decision, 
	we would like to ignore $\smallterm$.
	Observe that as $\smallterm\leq \neglectedterm$,
	we can safely do this if 
	\begin{align}\label{eq:guarantee}
		\neglectedterm < \frac{2}{27}(1-\pabsorption)^3\;,
	\end{align}
	as in this case, the \emph{entire} exposure of random walks due to nodes encountered after four or more steps in $\graph$
	is smaller than the change of the objective function we obtain by covering a \emph{single} new edge in the original MVC-3 instance $\graph'$.
	In \cref{lem:alphasetting}, 
	we prove that if we choose $\pabsorption\geq\frac{1}{2}$,
	then \cref{eq:guarantee} is guaranteed.
	Hence, 
	$\graph'$ has a minimum vertex cover of size at most $\budget$  if and only if \cref{eq:ourcheck} holds, 
	and we obtain the vertex cover of $\graph'$ by setting
	\begin{align*}
		\cover = \{i\in\nodes'\mid (i,b_i)\in \edges_\selection\}\;.
	\end{align*}
\end{proof}

\begin{restatable}{lem}{alphasetting}\label{lem:alphasetting}
	If in the setting of \cref{thm:hardness}, 
	we set the random-walk absorption probability to $\pabsorption \geq \frac{1}{2}$,
	then $\neglectedterm < \frac{2}{27}(1-\pabsorption)^3$.
\end{restatable}
\begin{proof}
	Recall the definition of $\neglectedterm$ from \cref{eq:neglectedterm}, and observe that the infinite series involved have closed-form solutions
	\begin{align*}
		\sum_{\distance=0}^{\infty}
		3^{-\distance}(1-\pabsorption)^{2\distance} 
		= \sum_{\distance=0}^{\infty}\frac{(1-\pabsorption)^{2\distance}}{3^{\distance}}
		=\sum_{\distance=0}^{\infty}\bigg(\frac{(1-\pabsorption)^{2}}{3}\bigg)^\distance\\
		= \frac{1}{1-\frac{(1-\pabsorption)^{2}}{3}}\;,~\text{and}\\
		\sum_{\distance=0}^{\infty}
		3^{-\distance-1}(1-\pabsorption)^{2\distance+1} 
		= \sum_{\distance=0}^{\infty}\frac{(1-\pabsorption)^{2\distance+1} }{3^{\distance+1}}
		= \sum_{\distance=0}^{\infty}\frac{1-\pabsorption}{3}\bigg(\frac{(1-\pabsorption)^{2}}{3}\bigg)^\distance\\
		= \frac{\frac{1-\pabsorption}{3}}{1-\frac{(1-\pabsorption)^{2}}{3}}
		\;,
	\end{align*}
	and that the partial sums evaluate to
	\begin{align*}
		\sum_{\distance=0}^{1}
		3^{-\distance}(1-\pabsorption)^{2\distance} 
		= 1 + \frac{1}{3}(1-\pabsorption)^2
		\;,~\text{and}\\
		\sum_{\distance=0}^{1}
		3^{-\distance-1}(1-\pabsorption)^{2\distance+1}
		= \frac{1}{3}(1-\pabsorption) + \frac{1}{9}(1-\pabsorption)^3
		\;.\\
	\end{align*}
	Using these equalities to rewrite \cref{eq:neglectedterm} for $\neglectedterm$, 
	and intermittently setting $x = 1-\pabsorption$, 
	where for $\pabsorption\geq\frac{1}{2}$, we have $x \leq \frac{1}{2}$,
	we obtain
	\begin{align*}
		\neglectedterm 
		&=
		\frac{1 + \frac{(1-\pabsorption)}{3}}{1-\frac{(1-\pabsorption)^2}{3}}
		- 1 - \frac{(1-\pabsorption)^2}{3} - \frac{(1-\pabsorption)}{3} - \frac{(1-\pabsorption)^3}{9}\\
		&=
		\frac{1 + \frac{x}{3}}{1-\frac{x^2}{3}}
		- 1 
		- \frac{x^2}{3} 
		- \frac{x}{3} 
		- \frac{x^3}{9}
		\\
		&=
		\frac{1 + \frac{x}{3}}{1-\frac{x^2}{3}}
		- \frac{1-\frac{x^2}{3}}{1-\frac{x^2}{3}}
		- \frac{\frac{x^2}{3}-\frac{x^2}{3}\frac{x^2}{3}}{1-\frac{x^2}{3}}
		- \frac{\frac{x}{3}-\frac{x}{3}\frac{x^2}{3}}{1-\frac{x^2}{3}} 
		- \frac{\frac{x^3}{9}-\frac{x^3}{9}\frac{x^2}{3}}{1-\frac{x^2}{3}}
		\\
		&=
		\frac{
			1 
			+ 
			\frac{x}{3}
			-
			1
			+
			\frac{x^2}{3}
			-
			\frac{x^2}{3}
			+
			\frac{x^2}{3}\frac{x^2}{3}
			-
			\frac{x}{3}
			+
			\frac{x}{3}\frac{x^2}{3}
			-
			\frac{x^3}{9}
			+
			\frac{x^3}{9}\frac{x^2}{3}
			}{
			1-\frac{x^2}{3}
		}\\
		&=
		\frac{
			\frac{x^4}{9}
			+
			\frac{x^5}{27}
		}{
			1-\frac{x^2}{3}
		}
		= 
		\frac{\frac{x}{9}
			+
			\frac{x^2}{27}}{1-\frac{x^2}{3}}x^3
		\\
		&\leq
		\frac{\frac{1}{18}
			+
			\frac{1}{108}}{1-\frac{1}{12}}x^3
		= 
		\frac{108 + 18}{18\cdot 108}\cdot\frac{12}{11}x^3 \\
		&= 
		\frac{126}{3\cdot 54 \cdot 11}x^3
		=
		\frac{126}{1\,749}x^3
		=
		\frac{3\,402}{1\,749\cdot 27}x^3\\
		&<
		\frac{3\,498}{1\,749\cdot 27}x^3
		=
		\frac{2}{27}x^3 
		= 
		\frac{2}{27}(1-\pabsorption)^3\;,
	\end{align*}
as required. 
\end{proof}

Note that the choice of $\pabsorption\geq\frac{1}{2}$ in \cref{lem:alphasetting} is almost tight, 
as for $\alpha = 1 - \frac{1}{10} (\sqrt{201} - 9) \approx 0.48$, 
we have
\begin{align}
	\frac{\frac{1-\pabsorption}{9}+\frac{(1-\pabsorption)^2}{27}}{1-\frac{(1-\pabsorption)^2}{3}} = \frac{2}{27}\;.
\end{align}
We present the slightly looser bound as it suffices to prove \cref{thm:hardness} and simplifies the presentation.

\subsection{Hardness of Approximation for \ourproblem}
\label{apx:hardness:apx}

\apxhardness*
\begin{proof}
	Under the UGC, 
	MVC is hard to approximate to within a factor of $(2-\varepsilon)$
	\cite{khot2008vertex},
	and it is generally hardest to approximate on regular graphs \cite{feige2003vertex}. 
	Therefore, consider again the reduction construction from the proof of \cref{thm:hardness} with an original MVC-3 graph $\graph'=(\nodes',\edges')$ as well as a transformed \ourproblem graph $\graph = (\nodes,\edges)$, 
	and assume that $\pabsorption = \frac{1}{2}$, 
	satisfying \cref{lem:alphasetting}.
	
	A solution to \ourproblem on a graph derived from an MVC-3 instance that has
	a minimum vertex cover of size $\budget$ 
	which approximates the optimum to within an additive error of 
	\begin{align}
		\frac{2\budget}{27}(1-\pabsorption)^3 
		-
		\frac{2\frac{\budget-\varepsilon\budget}{2}}{27}(1-\pabsorption)^3 
		= \frac{(1+\varepsilon)\budget}{27}(1-\pabsorption)^3
		=
		\frac{(1+\varepsilon)\budget}{27\cdot 8}
	\end{align}
	would rewire edges such that $\frac{\budget-\varepsilon\budget}{2}$ edges in the MVC-3 instance remain uncovered. 
	In this case, taking both endpoints of all uncovered edges yields a vertex cover of size $\budget + 2\frac{\budget-\varepsilon\budget}{2} = (2-\varepsilon)\budget$.
	Thus, if there existed an algorithm $\analgorithm$ discovering the stated approximate solution to \ourproblem in polynomial time, 
	we could obtain a $(2-\varepsilon)$-approximation to MVC-3 in polynomial time by transforming the MVC-3 instance into a \ourproblem-3 instance, 
	running $\analgorithm$ for all integers $\budget \in \{\frac{\nnodes'}{4},\dots,\frac{2\nnodes'}{3}\}$,
	where $\frac{\nnodes'}{4}$ and $\frac{2\nnodes'}{3}$ are the minimum resp. maximum cardinality of an MVC on a 3-regular undirected graph with $\nnodes'$ nodes, 
	reconstructing the vertex cover solutions, 
	and finally picking the solution with the smallest cardinality. 
	This would contradict the UGC. 
	Observing that the cardinality of an MVC in 3-regular undirected graphs with $\nnodes'$ nodes is in $\budget \in \Theta(\nnodes')$,
	that $\nnodes' \in \Theta(\nnodes)$, 
	and that $\frac{(1+\smallterm)}{27\cdot 8}\in\Theta(1)$,
	the claim follows.
\end{proof}

\subsection[Submodularity of $f$-Delta]{Submodularity of \maxobjectivef}
\label{apx:hardness:submod}

\submodularitywithduplicates*
\begin{proof}
	By assumption, there exists a safe node in $\graph$. 
	Therefore, fix a safe node $\safenode$, 
	and observe that $\safenode$ is an optimal rewiring target because $\unitvector_\safenode^T\fundamental\labelingvector = 0$.
	Hence, there exists an optimal strategy for maximizing \maxobjectivef that selects only rewirings $\rewiring{i}{j}{\safenode}$ with $\edge{i}{j}\in\edges_{\unsafenodes\unsafenodes}$. 
	Now denote the set of rewirings as $\selection$, 
	and
	the set of rewired edges as $\edges_{\selection} = \{(i,j)\mid \rewiring{i}{j}{k}\in\selection\}$. 
	Knowing that there exists an optimal rewiring for which $\edges_{\selection}\subseteq \edges_{\unsafenodes\unsafenodes}$, 
	we can define a set function $\maxsetobjectivef$ over the set $\edges_{\unsafenodes\unsafenodes}$ that is equivalent to \maxobjectivef as
	 \begin{align}
		\maxsetobjective{\edges_{\selection}} = \objective{\graph} - \objective{\graph_{\edges_{\selection}}}\;.
		\label{eq:maxsetobjective}
	\end{align}

	The function \maxsetobjectivef is \emph{monotone} because we only perform rewirings from $\edges_{\unsafenodes\unsafenodes}$ to $\safenode$,
	and no such rewiring can decrease \maxsetobjectivef.
	To see that \maxsetobjectivef is also \emph{submodular},
	fix $\edges_\selection\subseteq \edges_{\unsafenodes\unsafenodes}$, 
	and consider $x_1\neq x_2\in \edges_{\unsafenodes\unsafenodes}\setminus\edges_\selection$.
	Observe that $x_1$ and $x_2$ consist of unsafe nodes,
	which cannot be reachable from $\safenode$---%
	otherwise, $\unitvector_\safenode^T\fundamental\labelingvector > 0$, 
	and $\safenode$ would not be safe. 
	Hence, there is no exposure to harm that is \emph{only} removed when both $x_1$ \emph{and} $x_2$ are rewired, 
	and we have
	\begin{align}\nonumber
		\objective{\graph_{\edges_{\selection}}}
		-
		\objective{\graph_{\edges_{\selection}\cup\{x1,x_2\}}}
		\phantom{~.}\\ 
		\leq
		\big(\objective{\graph_{\edges_{\selection}}}
		-
		\objective{\graph_{\edges_{\selection}\cup\{x_1\}}}\big)
		+ 
		\big(\objective{\graph_{\edges_{\selection}}}
		-
		\objective{\graph_{\edges_{\selection}\cup\{x_2\}}}\big)
		\;.
		\label{eq:submodineq}
	\end{align}
	Using the definition from \cref{eq:maxsetobjective}, 
	we obtain
	\begin{align*}
		\objective{\graph_{\edges_{\selection}}}
		-
		\objective{\graph_{\edges_{\selection}\cup\{x1,x_2\}}}
		&= \objective{\graph_{\edges_{\selection}}} - \objective{\graph} + \maxsetobjective{\edges_\selection\cup\{x_1,x_2\}}\;,\\
		\objective{\graph_{\edges_{\selection}}}
		-
		\objective{\graph_{\edges_{\selection}\cup\{x_1\}}}
		&= 
		\objective{\graph_{\edges_{\selection}}}
		-
		\objective{\graph}
		+ \maxsetobjective{\edges_\selection\cup\{x_1\}}\;,~\text{and}\\
		\objective{\graph_{\edges_{\selection}}}
		-
		\objective{\graph_{\edges_{\selection}\cup\{x_2\}}}
		&= 
		\objective{\graph_{\edges_{\selection}}}
		-
		\objective{\graph}
		+ \maxsetobjective{\edges_\selection\cup\{x_2\}}\;,
	\end{align*}
	for the three parts of \cref{eq:submodineq}.
	Putting things together, we obtain
	\begin{align*}
	\objective{\graph_{\edges_{\selection}}} 
	- \objective{\graph} 
	+ \maxsetobjective{\edges_\selection\cup\{x_1,x_2\}}\phantom{~,}\\
	\leq
	\objective{\graph_{\edges_{\selection}}}
	-
	\objective{\graph}
	+ \maxsetobjective{\edges_\selection\cup\{x_2\}}\phantom{~,}\\
	+
	~\objective{\graph_{\edges_{\selection}}}
	-
	\objective{\graph}
	+ \maxsetobjective{\edges_\selection\cup\{x_1\}}\phantom{~,}\\
	\Leftrightarrow
	\objective{\graph} 
	- 
	\objective{\graph_{\edges_{\selection}}} 
	+ \maxsetobjective{\edges_\selection\cup\{x_1,x_2\}}\phantom{~,}\\
	\leq
	\maxsetobjective{\edges_\selection\cup\{x_1\}}
	+ 
	\maxsetobjective{\edges_\selection\cup\{x_2\}}\phantom{~,}\\
	\Leftrightarrow
	\maxsetobjective{\edges_\selection}
	+ \maxsetobjective{\edges_\selection\cup\{x_1,x_2\}}
	\leq
	\maxsetobjective{\edges_\selection\cup\{x_1\}}
	+ 
	\maxsetobjective{\edges_\selection\cup\{x_2\}}\;,
	\end{align*}
	which is the definition of submodularity.
\end{proof}

Observe that \cref{cor:apxguarantee}, 
which follows from \cref{lem:submoddupes,thm:submodularity}, 
does not contradict \cref{thm:apxhardness}:
As for the graphs used in \cref{thm:hardness} and \cref{thm:apxhardness}, 
which satisfy the precondition of \cref{thm:submodularity},
the value of $\maxobjectivef$ stated in \cref{eq:mvcreduction} is
\begin{align*}
	\maxobjective{\graph,\graph_{\budget}} = \frac{\budget}{3}(1-\pabsorption) 
	+ \frac{3\budget}{9}(1-\pabsorption)^2 
	+ \frac{2\nedges'}{27}(1-\pabsorption)^3
	+ \smallterm\;,
\end{align*}
the $(1-\nicefrac{1}{e})$-approximation of $\maxobjectivef$ guaranteed by \cref{cor:apxguarantee} 
still loses an additive term of $\Theta(\nnodes)$ and $\Theta(\budget)$, 
as required by \cref{thm:apxhardness}.

Furthermore, note that \cref{cor:apxguarantee} does not provide any approximation guarantee for the minimization of $\objectivef$: 
Although $\objectivef$ is necessarily supermodular when $\maxobjectivef$ is submodular, 
approximation guarantees from submodular maximization 
do not generally carry over to supermodular minimization \cite{ilev2001approximation,zhang2022fast}.

\subsection[Components of Delta]{Components of $\Delta$}
\label{apx:hardness:deltaanalysis}
\deltacomponents*
\begin{proof}
	For $\rho$, we have
	\begin{align}\label{eq:rho}
		\rho = 1+ \jkvector^T\fundamental\ivector
		= 1 + \jkvector^T \probability{ij} \fundamental[:,i]
		= 1 + \probability{ij}\fundamental[j,i] - \probability{ij}\fundamental[k,i]\thinspace.
	\end{align}
	For a node $x$, $\probability{ij}\fundamental[x,i]$ is the expected number of times we traverse the edge $\edge{i}{j}$ in a random walk starting at $x$. 
	Now, the probability that we reach $j$ from $k\notin\{i,j\}$ is at most $(1-\pabsorption)$, 
	and the probability that we traverse $\edge{i}{j}$ from $k$ without first visiting $j$ is at most $(1-\pabsorption)\probability{ij}$.
	Since $\pabsorption> 0$ and $\probability{ij}\leq 1-\pabsorption$, therefore, we have
	\begin{align}
		\probability{ij}\fundamental[k,i] \leq (1-\pabsorption)\probability{ij} + (1-\pabsorption)\probability{ij}\fundamental[j,i]
		< 1 + \probability{ij}\fundamental[j,i]\;,
	\end{align}
	and hence, $\rho > 0$.
	
	For $\sigma$, we have
	\begin{align}\label{eq:sigma}
		\sigma = \onevector^T \fundamental \ivector = \probability{ij}  \sum_x \fundamental[x,i] \;,
	\end{align}
	which is positive as all row sums of $\fundamental$ are positive.
	
	For $\tau$, we have 
	\begin{align}\label{eq:tau}
		\tau =	\jkvector^T\fundamental\labelingvector 
		= \unitvector_j^T\fundamental\labelingvector - \unitvector_k^T\fundamental\labelingvector
		\;,
	\end{align}
	which is 
	\emph{positive} (resp. \emph{negative}) if $j$ is \emph{more} (resp. \emph{less}) 
	exposed to harm than~$k$,
	and \emph{zero} if both nodes are \emph{equally} exposed to harm.
\end{proof}

\section{Other Graph Edits}
\label{apx:edits}

In this section, we define and analyze two other graph edits, 
which are less natural for recommendation graphs but potentially relevant in other applications: 
edge deletions and edge insertions.

\subsection{Edge Deletions}
An edge deletion removes an edge $(i,j)$ from $\graph$, 
redistributing the $\probability{ij}$ to the remaining edges outgoing from $i$.
Assuming that we redistribute the freed probability mass evenly among the remaining out-neighbors of $i$, 
the necessary changes are summarized in \cref{tab:deletion}. 
We require $\outdegree{i} > 1$, since otherwise, $i$ would have no remaining neighbors among which to distribute the unused probability mass (and to exclude division by zero), 
which would effectively require us to create a new absorbing state.

What can we say about the components of $\Delta=\nicefrac{\sigma\tau}{\rho}$?
For $\rho$, 
\begin{align*}
	\rho &= 1+ \jkvector^T\fundamental\ivector
	= 1 + \probability{ij}\fundamental[j,i]
	-  \frac{\probability{ij}}{\outdegree{i}-1}\underset{k\in\outneighbors{i}\setminus\{j\}}{\sum}  \fundamental[k,i]\;,
	\\
\end{align*} 
which generalizes what we observed for edge rewirings.
With the same reasoning as for edge rewiring, for each $k\in\outneighbors{i}$, we have
\begin{align*}
	\probability{ij}\fundamental[k,i] \leq (1-\pabsorption) + \probability{ij}\fundamental[j,i] 
	&< 1 + \probability{ij}\fundamental[j,i]\\
	\Leftrightarrow
	\frac{\probability{ij}}{\outdegree{i}-1}\fundamental[k,i] 
	&< \frac{1}{\outdegree{i}-1} + \frac{\probability{ij}}{\outdegree{i}-1}\fundamental[j,i]\;,
\end{align*}
such that $\rho$ again must be positive.
For $\sigma = \onevector^T \fundamental \ivector$, 
as $\ivector$ is exactly the same as for edge rewirings, 
the analysis for $\sigma$ under edge rewiring holds analogously.
For $\tau$, we get 
\begin{align*}
	\tau =	\jkvector^T\fundamental\labelingvector
	= \unitvector_j\fundamental\labelingvector 
	- \frac{1}{\outdegree{i}-1}\underset{k\in\outneighbors{i}\setminus\{j\}}{\sum} \unitvector_k\fundamental\labelingvector
	\;,
\end{align*}
which can have any sign, 
and which we would like to be positive because $\rho$ and $\sigma$ are positive, too.
Intuitively, this generalizes what we observed for edge rewirings:
To maximize $\tau$, we need to maximize the difference between the cost-scaled row sum of $j$ 
and the \emph{average} of the cost-scaled row sums of all other out-neighbors of $i$.

\subsection{Edge Insertions}
An edge insertion adds an edge $(i,j)$ into $\graph$ with a freely chosen $\probability{ij} \leq 1-\pabsorption$, 
reducing the probability masses associated with the other edges outgoing from $i$ proportionally.
Assuming that we subtract the required probability mass evenly from the original out-neighbors of $i$, 
the necessary changes are summarized in \cref{tab:insertion}.

What can we say about the components of $\Delta=\nicefrac{\sigma\tau}{\rho}$?
For $\rho$, 
\begin{align*}
	\rho &= 1+ \jkvector^T\fundamental\ivector
	= 1 + \frac{1}{\outdegree{i}}\sum_{k\in\outneighbors{i}}\probability{ij}\fundamental[k,i] - \probability{ij}\fundamental[j,i]\;,
\end{align*} 
which generalizes what we observed for edge rewirings.
Unfortunately, as in this case, the edge $(i,j)$ does not factor into the computation of $\fundamental$ (as is the case for edge rewirings and edge deletions), 
we cannot guarantee that $\rho$ is always positive.
For $\sigma = \onevector^T \fundamental \ivector$, 
as $\ivector$ is exactly the same as for edge rewirings, 
the analysis for $\sigma$ under edge rewiring holds analogously.
For $\tau$, we get 
\begin{align*}
	\tau =	\jkvector^T\fundamental\labelingvector
	= \frac{1}{\outdegree{i}}\sum_{k\in\outneighbors{i}}\unitvector_k\fundamental\labelingvector - \unitvector_j\fundamental\labelingvector\;,
\end{align*}
which can have any sign, 
and which we would like to be positive if $\rho$ is positive, 
and negative if $\rho$ is negative.
Intuitively, this generalizes what we observed for edge rewirings:
To maximize the $\tau$, 
we need to maximize the difference between the \emph{average} of the cost-scaled row sums of all out-neighbors of $i$ and the cost-scaled row sum of~$j$.

\begin{table}[!t]
	\centering
	\caption{Summary of an edge deletion $-(i,j)$ in a~graph $\graph = (\nodes,\edges)$
		with random-walk transition matrix $\transitionmatrix$
		and fundamental matrix $\fundamental = (\identity - \transitionmatrix)^{-1}$.}
	\label{tab:deletion}
	\begin{tabular}{l}
		\toprule
		$\graph'=(\nodes,\edges')$, for $\edges' = \edges \setminus \{(i,j)\}$, $(i,j)\in \edges$\\
		\midrule
		$\transitionmatrix'[x,y] = \begin{cases}
			0&\text{if}~x = i~\text{and}~y = j\\
			\transitionmatrix[i,y] + \frac{\transitionmatrix[i,j]}{\outdegree{i}-1} 
			&\text{if}~x = i~\text{and}~y \in\outneighbors{i}\setminus\{j\}\\
			0&\text{otherwise}\;.
		\end{cases}$\\
		\midrule
		$\fundamental'
		= \fundamental - \frac{\fundamental\ivector\jkvector^T\fundamental}{1 + \jkvector^T\fundamental\ivector}$, with $\ivector = \probability{ij}\unitvector_i$, $\jkvector = \unitvector_j-\frac{1}{\outdegree{i}-1}\underset{k\in\outneighbors{i}\setminus\{j\}}{\sum} \unitvector_k$\\
		\bottomrule
	\end{tabular}
\end{table}
\begin{table}[!t]
	\centering
	\caption{Summary of an edge insertion $+(i,j)$ in a~graph $\graph = (\nodes,\edges)$
		with random-walk transition matrix $\transitionmatrix$
		and fundamental matrix $\fundamental = (\identity - \transitionmatrix)^{-1}$.}
	\label{tab:insertion}
	\begin{tabular}{l}
		\toprule
		$\graph'=(\nodes,\edges')$, for $\edges' = \edges \cup \{(i,j)\}$, $(i,j)\notin E$\\
		\midrule
		$\transitionmatrix'[x,y] = \begin{cases}
			\probability{ij}&\text{if}~x = i~\text{and}~y = j\\
			\transitionmatrix[i,y] - \frac{\probability{ij}}{\outdegree{i}} 
			&\text{if}~x = i~\text{and}~y \in\outneighbors{i}\\
			0&\text{otherwise}\;,
		\end{cases}$\\
		\hspace*{4.5em}for $\probability{ij} \leq 1-\pabsorption$ chosen freely.\\
		\midrule
		$\fundamental'
		= \fundamental - \frac{\fundamental\ivector\jkvector^T\fundamental}{1 + \jkvector^T\fundamental\ivector}$, with $\ivector = \probability{ij}\unitvector_i$, $\jkvector = - \unitvector_j + \frac{1}{\outdegree{i}}\underset{k\in\outneighbors{i}\setminus\{j\}}{\sum} \unitvector_k$\\
		\bottomrule
	\end{tabular}
\end{table}

\section{Omitted Pseudocode}
\label{apx:pseudocode}

In the main paper, 
we omitted the pseudocode for \ourmethod, 
our algorithm for 
heuristic greedy $\budget$-rewiring exposure minimization (\ourproblem) and 
heuristic greedy $\qualitythreshold$-relevant $\budget$-rewiring exposure minimization (\ourproblemtwo). 
We now provide this pseudocode as \cref{alg:greedyrelevant:rem} for \ourproblem and \cref{alg:greedyrelevant} for \ourproblemtwo.
Furthermore, we also state the pseudocode for the algorithms leading up to \ourmethod, 
na\"ive greedy $\budget$-rewiring exposure minimization
and exact greedy $\budget$-rewiring exposure minimization, 
as \cref{alg:naivegreedy,alg:greedy}.

\begin{algorithm}[h]

\begin{algorithmic}[1]
	\Require Graph $\graph = (\nodes, \edges)$, 
	transition matrix $\transitionmatrix$, 
	costs $\labelingvector$, budget $\budget$
	\Ensure Set of $\budget$ rewirings $\selection$ of shape $(i,j,k)$ 
	\State $\selection \gets \emptyset$
	\For{$i \in \naturals_{\leq \budget}$}
	\State Compute $\onevector^T\fundamental$ and $\fundamental\labelingvector$\Comment{$\bigoh{\niter\nedges}$}
	\State 
	Compute $\onevector^T\fundamental\ivector$ for all $\edge{i}{j}\in \edges$\Comment{$\bigoh{\nedges}$}
	\State $\kcandidates \gets \{k\mid \fundamental\labelingvector[k]~\in \text{$\{(\maxoutdegree+2)$ smallest $\fundamental\labelingvector$-values}\}\}$\Comment{$\bigoh{\nnodes}$}
	\State Compute $\jkvector^T\fundamental\labelingvector$ for $j\in \nodes$ and $k\in \kcandidates$\Comment{$\bigoh{\maxoutdegree\nnodes}$}
	\State \Call{\greedyname}{}(\thinspace)
	\EndFor
	\State \Return $\selection$
	\Statex
	\Function{\greedyname}{}(\thinspace) 
	\State$\heuristic, i', j', k' \gets 0, \bot, \bot, \bot$
	\For{$\edge{i}{j}\in \edges$}\Comment{$\bigoh{\nedges}$}
	\For{$k_{ij}\in \kcandidates\setminus (\outneighbors{i}\cup \{i\})$}\Comment{$\bigoh{\maxoutdegree}$}
	\State $\ivector \gets \transitionmatrix[i,j]\unitvector_i$
	\State 
	$\jkvector \gets \unitvector_j - \unitvector_{k_{ij}}$
	\State $\heuristic_{ijk} \gets (\onevector^T\fundamental \ivector)(\jkvector^T\fundamental\labelingvector)$\Comment{$\bigoh{1}$}
	\If{$\heuristic_{ijk} > \heuristic$}
	\State $\heuristic, i', j', k' \gets \heuristic_{ijk}, i, j, k_{ij}$
	\EndIf
	\EndFor
	\EndFor
	\State $\edges \gets (\edges \setminus \{(i',j')\}) \cup \{(i',k')\}$
	
	\State 
	$\transitionmatrix[i',k'] \gets \transitionmatrix[i',j']$
	\State
	$\transitionmatrix[i',j'] \gets 0$
	\State$\selection \gets \selection \cup \{(i',j',k')\}$
	\algrenewcommand{\alglinenumber}[1]{\color{black!25}\footnotesize#1:}
	\EndFunction
	
\end{algorithmic}
	\caption{%
		Heuristic greedy \ourproblem with \ourmethod. 
	}\label{alg:greedyrelevant:rem}
\end{algorithm}

\begin{algorithm}[h]

\begin{algorithmic}[1]
	\Require \!Graph $\graph = (\nodes, \edges)$, 
	transition matrix $\transitionmatrix$, 
	costs $\labelingvector$, budget $\budget$,\newline
	\hspace*{1.2em} 
	relevance matrix $\relevancematrix$, relevance function $\qualityf$, quality threshold~$\qualitythreshold$
	\Ensure Set of $\budget$ rewirings $\selection$ of shape $(i,j,k)$ 
	
	\State $\selection \gets \emptyset$
	\State $\qualityset \gets \{\rewiring{i}{j}{k}~\mid \edge{i}{j}\in\edges, \edge{i}{k}\notin\edges, \rewiring{i}{j}{k}~\text{is $\qualitythreshold$-permissible}\}$
	\Statex\Comment{$\bigoh{\qualitycomplexitystart}$}
	\For{$i \in \naturals_{\leq \budget}$}
	\State Compute $\onevector^T\fundamental$ and $\fundamental\labelingvector$\Comment{$\bigoh{\niter\nedges}$}
	\State 
	Compute $\onevector^T\fundamental\ivector$ for all $\edge{i}{j}\in \edges$\Comment{$\bigoh{\nedges}$}
	\For{$\edge{i}{j}\in\edges$}\Comment{$\bigoh{\nedges}$}
	\State $k_{ij} \gets \arg\min\{\unitvector_k^T\fundamental\labelingvector\mid\rewiring{i}{j}{k}\in \qualityset\}$\Comment{$\bigoh{\nqpermissible}$}
	\State Compute $\jkvector^T\fundamental\labelingvector$ for $j$ and $k_{ij}$\Comment{$\bigoh{1}$}
	\EndFor
	\State \Call{\greedyname}{}(\thinspace)
	\EndFor
	\State \Return $\selection$
	\Statex
	\Function{\greedyname}{}(\thinspace) 
	\State$\heuristic, i', j', k' \gets 0, \bot, \bot, \bot$
	\For{$\edge{i}{j}\in \edges$}\Comment{$\bigoh{\nedges}$}
	\State $\ivector \gets \transitionmatrix[i,j]\unitvector_i$ 
	\State 
	$\jkvector \gets \unitvector_j - \unitvector_{k_{ij}}$
	\State $\heuristic_{ijk} \gets (\onevector^T\fundamental \ivector)(\jkvector^T\fundamental\labelingvector)$\Comment{$\bigoh{1}$}
	\If{$\heuristic_{ijk} > \heuristic$}
	\State $\heuristic, i', j', k' \gets \heuristic_{ijk}, i, j, k_{ij}$
	\EndIf
	\EndFor
	\State $\edges \gets (\edges \setminus \{(i',j')\}) \cup \{(i',k')\}$
	
	\State 
	$\transitionmatrix[i',k'] \gets \transitionmatrix[i',j']$ 
	\State
	$\transitionmatrix[i',j'] \gets 0$
	\State$\selection \gets \selection \cup \{(i',j',k')\}$
	\State Update $\qualityset$\Comment{$\bigoh{\nqpermissible\qualitycomplexity}$}
	\EndFunction
	
\end{algorithmic}
	\caption{%
		Heuristic greedy \ourproblemtwo with \ourmethod. 
	}\label{alg:greedyrelevant}
\end{algorithm}

\clearpage

\begin{algorithm}[h]

\begin{algorithmic}[1]
	\Require Graph $\graph = (\nodes, \edges)$, 
	transition matrix $\transitionmatrix$, 
	costs $\labelingvector$,
	budget $\budget$
	\Ensure Set of $\budget$ rewirings $\selection$ of shape $(i,j,k)$ 
	
	\State $\fundamental \gets (\identity - \transitionmatrix)^{-1}$\Comment{\cref{eq:fundamental}}\label{line:fundamental}
	\State $\selection \gets \emptyset$\label{line:emptyset}
	\For{$i \in \naturals_{\leq \budget}$}\label{line:budget}
	\State \Call{\greedyname}{}(\thinspace)\label{line:greedycall}
	\EndFor
	\State \Return $\selection$\label{line:rewiringreturn}
	
	\item[]
	
	\Function{\greedyname}{}(\thinspace) 
	\State$\Delta, i', j', k' \gets 0, \bot, \bot, \bot$\label{line:maxdelta}
	\For{$(i,j)\in \edges$}\label{line:alledges}
	\For{$k\in \nodes\setminus(\outneighbors{i}\cup\{i\})$}\label{line:alltargets}
	\State $\ivector \gets \transitionmatrix[i,j]\unitvector_i$ 
	\State 
	$\jkvector \gets \unitvector_j - \unitvector_k$
	\State $\Delta_{ijk} \gets \frac{(\onevector^T\fundamental \ivector)(\jkvector^T\fundamental\labelingvector)}{1 + \jkvector^T\fundamental\ivector}$\Comment{\cref{eq:delta}}\label{line:delta}
	\If{$\Delta_{ijk} > \Delta$}
	\State$\Delta, i', j', k' \gets \Delta_{ijk}, i, j, k$\label{line:newmaxdelta}
	\EndIf
	\EndFor
	\EndFor
	\State $\edges \gets (\edges \setminus \{(i',j')\}) \cup \{(i',k')\}$\Comment{\cref{tab:rewiring}}\label{line:edgeupdate}
	
	\State 
	$\transitionmatrix[i',k'] \gets \transitionmatrix[i',j']$
	\State
	$\transitionmatrix[i',j'] \gets 0$
	\State $\fundamental \gets \fundamental - \frac{\fundamental\ivector\jkvector^T\fundamental}{1 + \jkvector^T\fundamental\ivector}$\label{line:fundamentalupdate}
	\State$\selection \gets \selection \cup \{(i',j',k')\}$\label{line:selectionupdate}
	\EndFunction
	
\end{algorithmic}
	\caption{Na\"ive greedy \ourproblem.}\label{alg:naivegreedy}
\end{algorithm}

\begin{algorithm}[h]

\begin{algorithmic}[1]
	\Require Graph $\graph = (\nodes, \edges)$, 
	transition matrix $\transitionmatrix$, 
	costs $\labelingvector$,
	budget $\budget$
	\Ensure Set of $\budget$ rewirings $\selection$ of shape $(i,j,k)$ 
	
	\State $\selection \gets \emptyset$
	\For{$i \in \naturals_{\leq \budget}$}
	\State Precompute $\onevector^T\fundamental$ and $\fundamental\labelingvector$\Comment{$\bigoh{\niter\nedges}$}
	\State Precompute $\onevector^T\fundamental\ivector$ for $(i,j)\in \edges$\Comment{$\bigoh{\nedges}$}
	\State Precompute $\fundamental\ivector$ for $(i,j)\in\edges$\Comment{$\bigoh{\nnodes^2}$}
	\State Precompute $\jkvector^T\fundamental\labelingvector$ for $j\neq k\in \nodes$\Comment{$\bigoh{\niter\nnodes^2}$}
	\State \Call{\greedyname}{}(\thinspace)
	\EndFor
	\State \Return $\selection$
	
	\item[]
	
	\Function{\greedyname}{}(\thinspace) 
	\State$\Delta, i', j', k' \gets 0, \bot, \bot, \bot$
	\For{$(i,j)\in \edges$}\Comment{$\bigoh{\nedges}$}
	\For{$k\in \nodes\setminus(\outneighbors{i}\cup\{i\})$}\Comment{$\bigoh{\nnodes}$}
	\State $\ivector \gets \transitionmatrix[i,j]\unitvector_i$
	\State 
	$\jkvector \gets \unitvector_j - \unitvector_k$
	\State $\Delta_{ijk} \gets \frac{(\onevector^T\fundamental \ivector)(\jkvector^T\fundamental\labelingvector)}{1 + \jkvector^T\fundamental\ivector}$\Comment{$\bigoh{1}$}
	\If{$\Delta_{ijk} > \Delta$}
	\State$\Delta, i', j', k' \gets \Delta_{ijk}, i, j, k$
	\EndIf
	\EndFor
	\EndFor
	\State $\edges \gets (\edges \setminus \{(i',j')\}) \cup \{(i',k')\}$
	
	\State 
	$\transitionmatrix[i',k'] \gets \transitionmatrix[i',j']$
	\State
	$\transitionmatrix[i',j'] \gets 0$
	\State$\selection \gets \selection \cup \{(i',j',k')\}$
	\EndFunction
	
\end{algorithmic}
	\caption{Exact greedy \ourproblem.}\label{alg:greedy}
\end{algorithm}

\section{Reproducibility Information}
\label{apx:reproducibility}

We make all code, datasets, and results publicly available.\!\footnote{\oururl}

\subsection[Choice of Relevance Function Theta]{Choice of Relevance Function $\qualityf$}
\label{apx:reproducibility:qualityf}

In our experiments, we instantiate $\qualityf$ with the normalized discounted cumulative gain ($\ndcg$), 
a popular measure of ranking quality. 

Given $\relevancematrix$, 
and denoting as $\rank_i(j)$ the relevance rank of $j$ for $i$,
we define the \emph{discounted cumulative gain} 
(\dcg) \cite{jarvelin2002cumulated} of node $i$ as
\begin{align}
	\dcg(i) &= \sum_{j\in\outneighbors{i}} \frac{\relevancematrix[i,j]}{\log_2(1 + \rank_i(j))}\;.
\end{align}
Denoting as $\topranked_{\outdeg}\!(i) = \{j\mid \rank_i(j)\leq \outdegree{i}\}$ 
the $\outdegree{i}$ most relevant nodes for node $i$, 
we obtain the \emph{normalized} \dcg as
\begin{align}
	\ndcg(i) &= \frac{\dcg(i)}{\idcg(i)}\;,~\text{where}\\
	\idcg(i) &= \!\!\!\!\sum_{j\in\topranked_{\outdeg}\!(i)} \frac{\relevancematrix[i,j]}{\log_2(1 + \rank_i(j))}\phantom{nI\;.}
\end{align}
is the \emph{ideal} discounted cumulative gain.
Asserting that $\ndcg(i) \geq \qualitythreshold$ in the original graph $\graph$ 
(which holds when a system simply recommends the top-ranked items, 
such that before rewiring, 
$\ndcg(i) = \idcg(i)$ for all $i\in\nodes$) 
allows us to require that all rewirings $\rewiring{i}{j}{k}$ maintain $\ndcg(i) \geq \qualitythreshold$ for $i\in\nodes$. 

We can compute the \ndcg in time $\bigoh{\maxoutdegree}$,
assuming that lookup operations for matrix elements and relevance ranks take constant time.
As $\maxoutdegree\in\bigoh{1}$ for $\bigoh{1}$-out-regular graphs, 
this entails that in our experiments, 
we can evaluate $\qualityf\equiv \ndcg$ in constant time.

\subsection{Binarization Thresholds for \fabbrialg}
\label{apx:reproducibility:binarization}

As \fabbrialg can only handle binary costs, 
we transform nonbinary costs $\labeling$ into binary costs $\labeling'$ by thresholding to guarantee $\labeling_i \geq \roundingthreshold \Leftrightarrow \labeling'_i = 1$ for some rounding threshold $\roundingthreshold\in(0,1]$. 
In our experiments, we use the binarization thresholds $1.0$, $0.6$, and $0.4$, 
which are chosen to ensure that they yield \emph{different} binarized costs given our original real-valued costs as inputs.
Note, however, that the resulting problem instances differ from the original instances, 
and as such,
it is hardly possible to fairly compare \ourmethod with \fabbrialg in the real-valued setting.
Hence, we focus our performance comparisons on the binary setting.

\subsection{Other Parameters}
\label{apx:reproducibility:poweriteration}

\paragraph{Hedging against small fluctuations}
When developing \ourmethod, 
we state that we can hedge against small fluctuations in the relationship between $\Delta$ and $\heuristic$ by computing $\Delta$ exactly for the rewiring candidates associated with the $\bigoh{1}$ largest values of $\heuristic$ before selecting the final rewiring. 
In our experiments, we compute $\Delta$ exactly for the top $100$ rewiring candidates.

\paragraph{Error bounds for power iteration}
Recall that $\norm{\somematrix^\niter} \leq \norm{\somematrix}^\niter$ for any square matrix $\somematrix$, associated matrix norm $\norm{\cdotB}$, and non-negative integer $\niter$.
Recall further that each row of $\transitionmatrix$ sums to $(1-\pabsorption)$, 
such that $\norm{\transitionmatrix}_{\infty} = (1-\pabsorption)$ and
\begin{align}
	\norm{\transitionmatrix^\niter}_{\infty} = \norm{\transitionmatrix}_{\infty}^\niter = (1-\pabsorption)^\niter\;.
\end{align}
Therefore, we can bound the approximation error in the infinity norm of our approximation of $\fundamental$ as 
\begin{align}
	\sum_{i=0}^{\infty}\transitionmatrix^i - \sum_{i=0}^{\niter} \transitionmatrix^i 
	= \frac{1}{1-(1-\pabsorption)} - \frac{1-(1-\pabsorption)^\niter}{1-(1-\pabsorption)}
	= \frac{(1-\pabsorption)^\niter}{\pabsorption}\;.
\end{align}
Thus, to obtain an approximation error on the row sums of at most $\apxerror$, we need to set the number of iterations $\niter$
\begin{align}
	\apxerror \leq \frac{(1-\pabsorption)^\niter}{\pabsorption} 
	\Leftrightarrow
	\apxerror\pabsorption \leq (1-\pabsorption)^\niter
	\Leftrightarrow 
	\niter \geq \log_{1-\pabsorption}\apxerror\pabsorption = \frac{\log \apxerror\pabsorption}{\log 1-\pabsorption}\;.
\end{align}
For example, to obtain an absolute approximation error $\apxerror \leq 0.01$ on the row sums, given an absorption probability of $\pabsorption = 0.05$, 
we need to set
\begin{align*}
	\niter \geq \frac{\log 0.005}{\log 0.95} = 148.18 \approx 150\;,
\end{align*}
independently of $\nnodes$. 
This justifies the assumption that $\niter\in\bigoh{1}$, and prompts us to set the number of iterations in all power iteration calculations to $150$.

\section{Dataset Information}
\label{apx:datasets}

In this section, we provide more information on our datasets and the cost assignments used in our experiments.

\subsection{Synthetic Data}
\label{apx:datasets:synthetic}

\subsubsection{Preprocessing}
\label{apx:datasets:synthetic:preprocessing}
Since the viewports of popular electronic devices typically fit around five recommendations,
as our synthetic data, 
we generate synthetic 5-out-regular graphs.
We experiment with four graph sizes using 
three absorption probabilities, 
two shapes of probability distributions over out-edges, 
three fractions of latently harmful nodes, 
and two cost functions, 
one binary and one real-valued based on a mixture of two beta distributions as depicted in \cref{fig:betamixture}, 
to assign costs to nodes.
We state the details on these parameters in \cref{tab:syntheticparameters}. 
For each of the resulting $144$ configurations, 
we place edges using two different edge-placement models, 
\synuni and \synhom, 
for a total of $288$ graphs.
For each node $i$, \synuni chooses $\outregulardegree$ distinct nodes $j\neq i$ as targets for its edges uniformly at random, 
whereas \synhom chooses $\outregulardegree$ distinct targets by sampling each $j$ with probability $\frac{1-\cardinality{\labeling(i)-\labeling(j)}}{\sum_{j\in\nodes}(1-\cardinality{\labeling(i)-\labeling(j)})}$. 
Hence, in \synhom, edges are drawn preferentially between nodes of similar costs, implementing homophily, 
whereas in \synuni, edges are drawn uniformly at random.

\begin{figure}[t]
	\centering
	\includegraphics[width=0.5\linewidth]{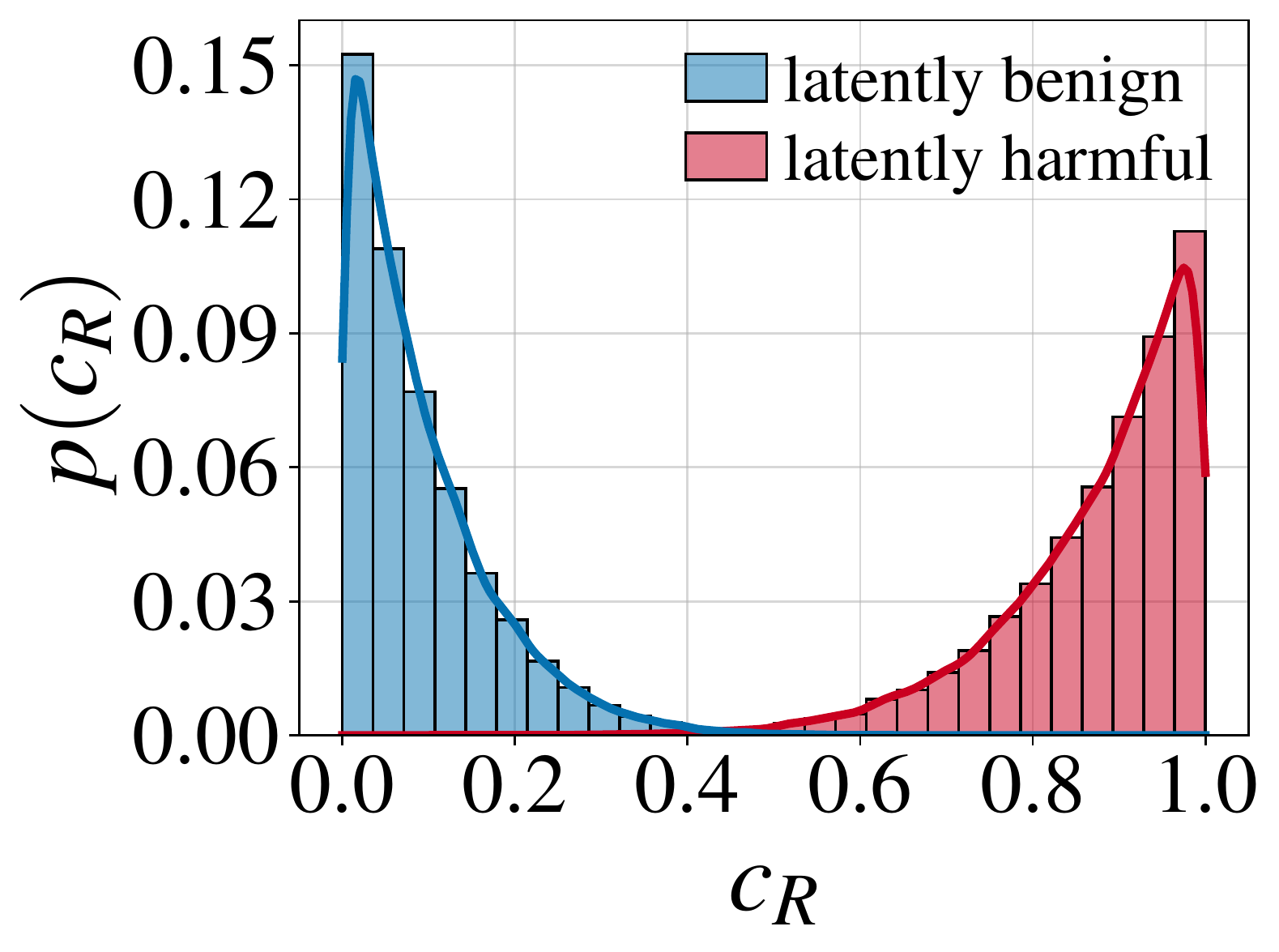}
	\caption{%
		Histograms and kernel density estimates of $50\,000$ draws from the beta distributions used to assign costs to nodes in \synuni and \synhom under $\labeling_R$. 
		The cost of latently benign nodes is drawn from $\mathrm{Beta}(1,10)$, whereas the cost of latently harmful nodes is drawn from $\mathrm{Beta}(7,1)$.
	}
	\label{fig:betamixture}
\end{figure}%
~
\begin{figure}[b]
	\centering
	\begin{subfigure}{0.5\linewidth}
		\centering
		\includegraphics[width=\linewidth]{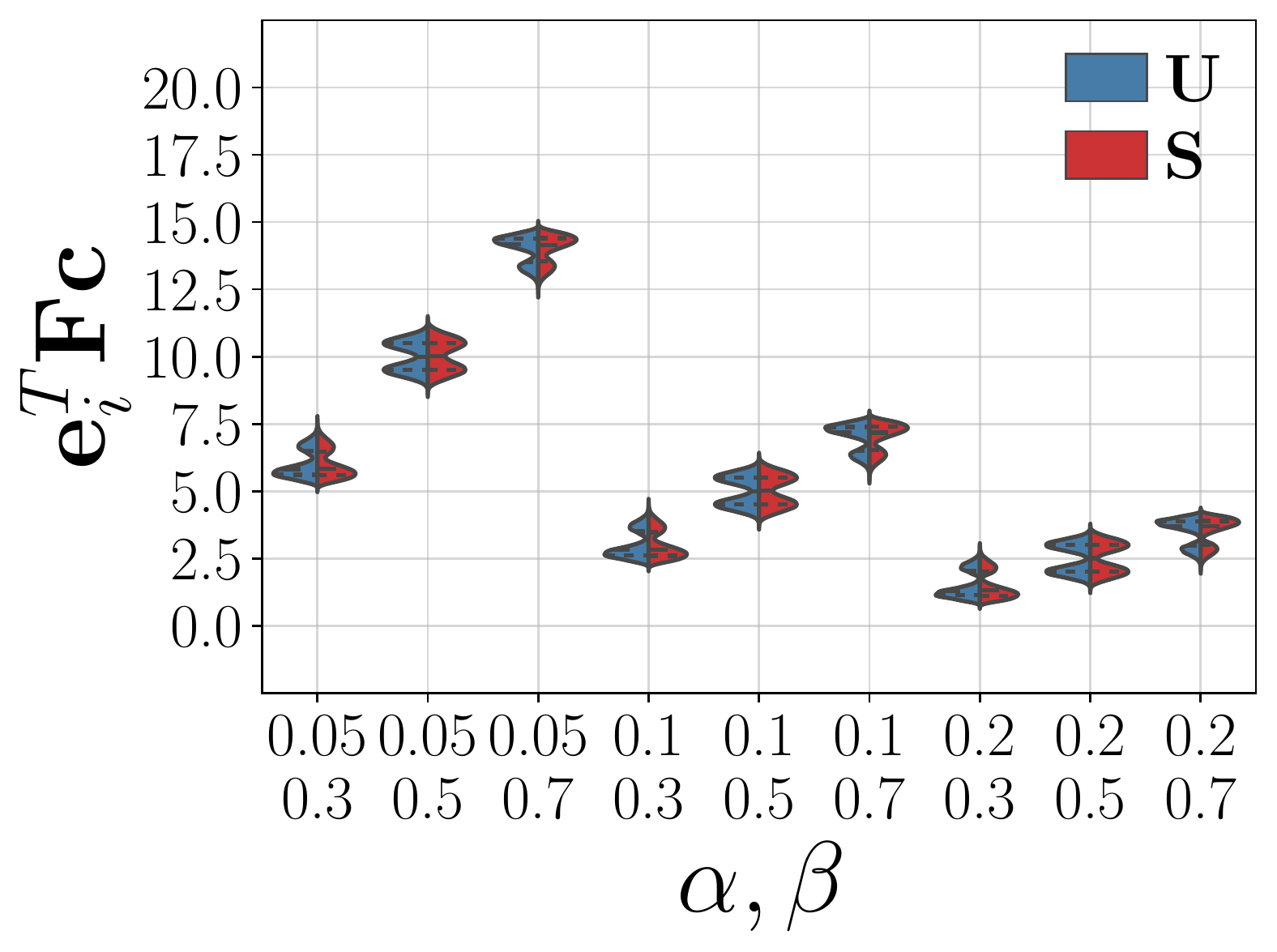}
		\subcaption{\synuni, cost function $\labeling_B$}
	\end{subfigure}~
	\begin{subfigure}{0.5\linewidth}
		\centering
		\includegraphics[width=\linewidth]{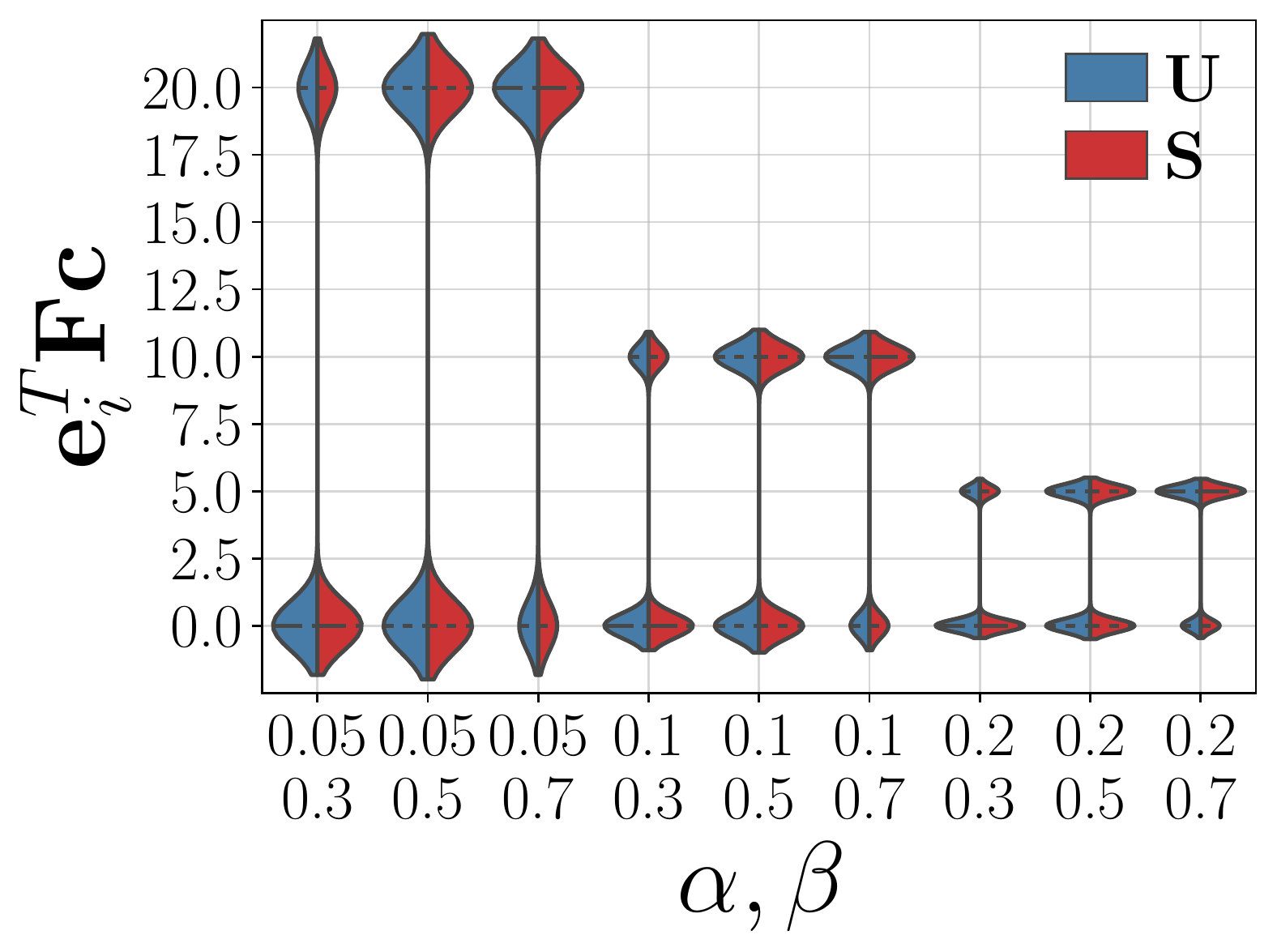}
		\subcaption{\synhom, cost function $\labeling_B$}
	\end{subfigure}
	\begin{subfigure}{0.5\linewidth}
		\centering
		\includegraphics[width=\linewidth]{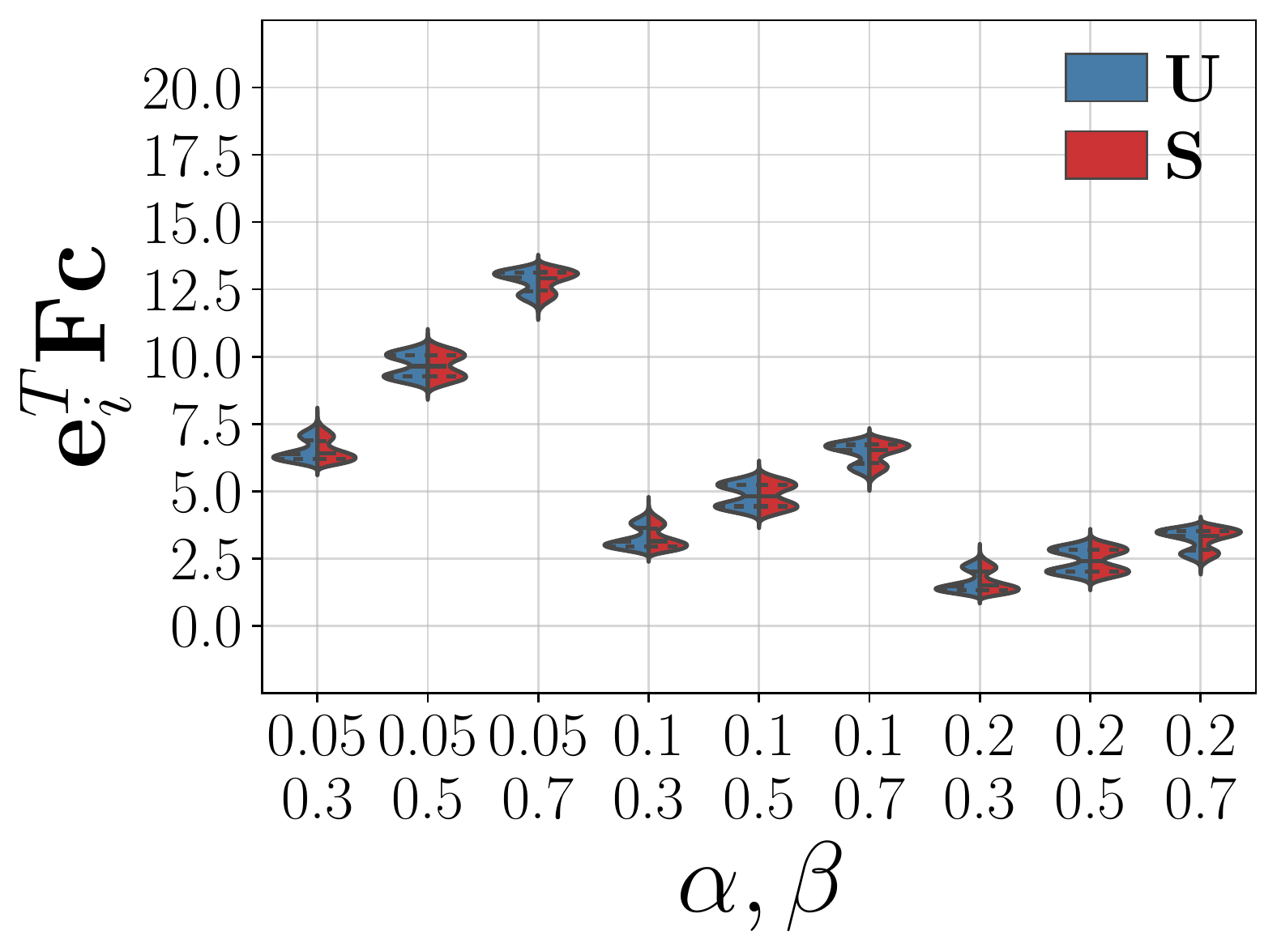}
		\subcaption{\synuni, cost function $\labeling_R$}
	\end{subfigure}~
	\begin{subfigure}{0.5\linewidth}
		\centering
		\includegraphics[width=\linewidth]{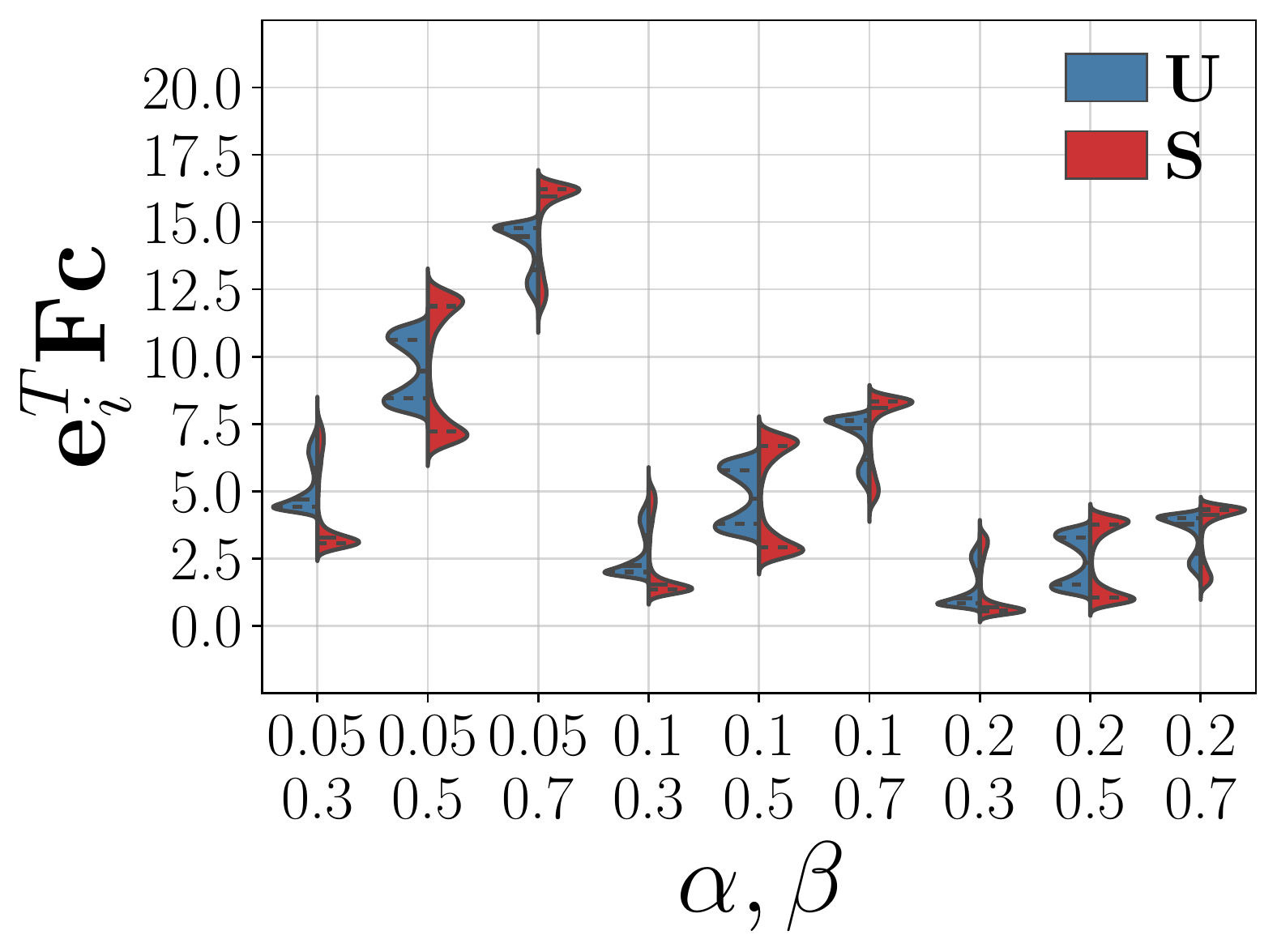}
		\subcaption{\synhom, cost function $\labeling_R$}
	\end{subfigure}
	\caption{%
		Distributions of initial exposures $\unitvector_i^T\fundamental\labelingvector$ 
		for nodes in \synuni and \synhom graphs with $100\,000$ nodes, 
		for all combinations of absorption probabilities $\pabsorption\in\{0.05,0.1,0.2\}$ ($x$-axis label, first row), 
		fractions of latently harmful nodes $\fractionbad\in\{0.3,0.5,0.7\}$  ($x$-axis label, second row), 
		and out-edge probability distributions $\probabilityshape\in\{\mathbf{U},\mathbf{S}\}$ (color).
		While the choice of the cost function barely makes a difference under random edge placements (\synuni, left), 
		it has a tremendous impact under homophilous edge placements (\synhom, right).
	}
	\label{fig:syn-costs}
\end{figure}

\subsubsection{Statistics}
\label{apx:datasets:synthetic:statistics}

In \cref{fig:syn-costs}, we show the distributions of initial exposures for nodes in our \synuni and \synhom graphs. 
For \synuni, we observe that the range of initial exposures is small and the cost function choice barely makes a difference,
which is expected as edges are placed uniformly at random. 
In contrast, for \synhom, 
we observe the maximum range of initial exposures under $\labeling_B$, 
as homophilous sampling under binary costs effectively splits the graph into two components consisting of harmful and benign nodes, respectively.
Under $\labeling_R$, 
we still observe a range of initial exposures that is twice to thrice as large as in \synuni graphs, 
and the probability shape $\probabilityshape\in\{\mathbf{U},\mathbf{S}\}$ strongly influences the distribution of initial exposures.
These are again effects of homophilous sampling.

\begin{table}[t]
	\centering
	\caption{Parameters used to generate \synuni and \synhom graphs.}
	\label{tab:syntheticparameters}
	\begin{tabular}{p{0.28\linewidth}@{\hskip 2pt}p{0.69\linewidth}}
		\toprule
		Parameter&Meaning\\
		\midrule
		$\outregulardegree = 5$&Regular out-degree\\
		$\nnodes \in \{10^i\mid i$\newline \phantom{1000}$\in \{2,3,4,5\}\}$&Number of nodes in $\graph$\\
		$\pabsorption\in\{0.05,0.1,0.2\}$&Random-walk absorption probability\\
		$\probabilityshape\in\{\mathbf{U},\mathbf{S}\}$&Probability shape over a node's out-edges:\newline 
		-- \textbf{U}niform\newline \phantom{-- }$\big(\probability{ij}=\frac{1-\pabsorption}{\outregulardegree}$ for all $\edge{i}{j}\in\edges\big)$;\newline -- \textbf{S}kewed\newline \phantom{-- }$\big(\frac{1-\pabsorption}{\outregulardegree}\cdot p$ for $p \in \langle0.35, 0.25, 0.20, 0.15, 0.05\rangle\big)$\\
		$\fractionbad\in\{0.3,0.5,0.7\}$&Fraction of latently harmful nodes\newline
		-- \synuni: fraction of nodes $i$ with $\labeling_i = 1$\newline
		-- \synhom: fraction of nodes drawn from the beta\newline \phantom{-- }distribution with parameters $\alpha=7$, $\beta=1$
		\\
		$\labeling\in\{\labeling_B,\labeling_R\}$&Cost functions\vspace*{1pt}\newline
		-- $\labeling_B(i) =\begin{cases}
			1&\text{if $i$ is latently harmful}\\
			0&\text{otherwise}\\
		\end{cases}$\newline
		-- $\labeling_R(i) =\begin{cases}
			\mathrm{Beta}(7,1)&\text{if $i$ is latently harmful}\\
			\mathrm{Beta}(1,10)&\text{otherwise}\\
		\end{cases}$
		\\
		\bottomrule
	\end{tabular}
\end{table}

\subsection{Real-World Data}
\label{apx:datasets:real}

\subsubsection{Preprocessing}

\paragraph{YouTube datasets}
For our YouTube datasets, 
like \citeauthor{fabbri2022rewiring} \cite{fabbri2022rewiring}, 
who experiment with a prior (not uniquely identified) version of this dataset,
we generate $\outregulardegree$-regular recommendation graphs for $\outregulardegree\in\{5,10,20\}$ that contain only videos with at least 100\,000, resp.\ 10\,000, views as nodes 
(\ytone, resp.\ \yttwo). 
Similar to \citeauthor{fabbri2022rewiring} \cite{fabbri2022rewiring}
we treat the observed recommendations as implicit feedback interactions, 
eliminate sinks in the observed recommendation graph, 
use alternating least squares to generate relevance scores \cite{hu2008collaborative}, 
and then take the nodes with the top scores as targets of out-edges in our reconstructed recommendation graphs. 
We additionally distinguish three absorption probabilities $\pabsorption\in\{0.05,0.1,0.2\}$ 
and two shapes $\probabilityshape\in\{\mathbf{U},\mathbf{S}\}$ of probability distributions over out-edges in our random-walk model, 
which leaves us with 36 transition matrices from six underlying graph structures.

\paragraph{NELA-GT datasets}

To create our NELA-GT datasets, 
we restrict ourselves to news items of \emph{at least 140 characters} (thus excluding boilerplate messages which we suspect were captured by accident)
that were published in \emph{January 2021} by one of the 341 outlets for which all veracity labels are present, 
and consider the articles containing the authors' \emph{January 6 keywords} (\nelaone), 
the articles containing the authors' \emph{COVID-19 keywords} (\nelatwo), 
and the collection containing \emph{all} articles (\nelathree). 
After embedding all news items using the \emph{all-MiniLM-L6-v2} model from the Sentence-Transformers library \cite{reimers2019sentence}, 
we compute pairwise cosine similarities $\cos(i,j)$ between all articles from the respective collection,
transform these similarities into relevance scores between 0 and 1 via min-max-normalization ($\relevancematrix[i,j] = \frac{\cos(i,j)+1}{2}$), 
and take the news items with the highest scores as targets of out-edges in our initial news feed graphs.

\subsubsection{Cost functions}
\label{apx:data:costfunctions}
For both the \yt and the \nf datasets, we measure performance based on four different cost functions $\labeling$ 
from two binary and two real-valued assignments of costs to channels and their videos ($\labeling_{B1}$, $\labeling_{B2}$, resp.\ $\labeling_{R1}$, $\labeling_{R2}$).

\paragraph{YouTube datasets}
Our cost functions for the \yt datasets map the channel categories provided with the original dataset to costs based on different mapping rules.
\Cref{tab:yt-costs} details the assignment of costs to video channels under our four different cost functions,
and in \cref{tab:yt-stats}, we provide the number of videos and the number of channels per category in \ytone and \yttwo. 
Additionally, 
\cref{tab:yt-starting-costs} lists the expected initial exposures of nodes in each of our YouTube recommendation graphs. 

\paragraph{NELA-GT datasets}
The costs we assign to nodes in our \nf datasets 
are based on the \emph{Media Bias/Fact Check scores} as well as the \emph{questionable source} and \emph{conspiracy/pseudoscience} flags of news outlets 
provided with the original dataset.
As the number of news outlets covered by this dataset is too large to detail their individual cost assignments,
here, we instead state how we transform the labels provided with the dataset into cost assignments under our four different cost functions. 
We define
\begin{align*}
	\labeling_{B1} &= \begin{cases}
		1&\text{if}~\texttt{questionable\_source} = 1\\&\text{or}~\texttt{conspiracy\_pseudoscience} = 1\\
		0&\text{otherwise,}
	\end{cases}\\
	\labeling_{B2} &= \begin{cases}
		1&\text{if}~\texttt{factuality} \leq 2\\
		0&\text{otherwise,}
	\end{cases}\\
\labeling_{R1} &= 1 - \frac{\texttt{factuality}}{5}\;,~\text{and}\\
\labeling_{R2} &= 1 - \frac{\texttt{label}}{2}\;,
\end{align*}
where \texttt{typewritten variables} are the names of the corresponding columns in the original data, 
$\texttt{questionable\_source}\in\{0,1\}$, 
$\texttt{conspiracy\_pseudoscience}\in\{0,1\}$, 
$\texttt{factuality}\in\{0,1,2,3,4,5\}$, and $\texttt{label}\in\{0,1,2\}$ is an aggregate label combining the other scores.
An overview of the resulting cost assignments in each of our \nf datasets is given in \cref{tab:nela-cost-counts}. 
In \cref{tab:nela-initial-exposures}, 
we additionally state the expected total exposure as well as the 	total segregation and the maximum segregation from \citeauthor{fabbri2022rewiring} \cite{fabbri2022rewiring} for all \nf datasets with $\pabsorption = 0.05$ and  $\probabilityshape=\textbf{U}$.

\begin{table}[b]
	\centering
	\caption{Costs of videos from each channel under our four different cost functions.}
	\label{tab:yt-costs}
	\begin{tabular}{lrrrr}
\toprule
             Category &  $c_{B1}$ &  $c_{B2}$ &  $c_{R1}$ &  $c_{R2}$ \\
\midrule
             Alt-lite &       1.0 &       1.0 &       0.8 &       0.8 \\
            Alt-right &       1.0 &       1.0 &       1.0 &       1.0 \\
                Incel &       0.0 &       1.0 &       0.4 &       0.6 \\
			IDW &       1.0 &       1.0 &       0.6 &       0.2 \\
                MGTOW &       0.0 &       1.0 &       0.4 &       0.6 \\
                  MRA &       0.0 &       1.0 &       0.2 &       0.4 \\
                 NONE &       0.0 &       0.0 &       0.0 &       0.0 \\
                  PUA &       0.0 &       1.0 &       0.2 &       0.4 \\
               center &       0.0 &       0.0 &       0.0 &       0.0 \\
                 left &       0.0 &       0.0 &       0.0 &       0.0 \\
          left-center &       0.0 &       0.0 &       0.0 &       0.0 \\
                right &       0.0 &       0.0 &       0.0 &       0.0 \\
         right-center &       0.0 &       0.0 &       0.0 &       0.0 \\
\bottomrule
\end{tabular}

\end{table}

\subsubsection{Statistics}
\label{apx:datasets:statistics}

\paragraph{In-degree distributions}
As our real-world graphs are $\outregulardegree$-out-regular by construction, their out-degree distributions are uniform. 
In contrast, the in-degree distributions of these graphs are highly skewed.
In \cref{fig:real-indegrees-zero}, 
we show the normalized in-degree distributions of our two largest real-world datasets, \nelathree and \yttwo. 
Note that in-degrees, at least visually, 
appear to be \emph{exponentially} distributed in the \nelathree graph and \emph{power-law} distributed in the \yttwo graph. 
Further, as illustrated in \cref{fig:real-indegrees}, 
even when considering only non-zero in-degrees and \emph{all} real-world graphs, 
the \nf graphs appear to be about an order of magnitude less skewed than the \yt graphs.

\clearpage

\begin{table}[t]
	\centering
	\caption{Number of videos $\cardinality{V}$ and number of channels $\cardinality{C}$ for each of our YouTube datasets.}
	\label{tab:yt-stats}
	\begin{tabular}{lrrrr}
\toprule
 & \textsc{YT-100k} & \textsc{YT-10k} & \textsc{YT-100k} & \textsc{YT-10k} \\
Category & $|V|$ & $|V|$ & $|C|$ & $|C|$ \\
\midrule
Alt-lite & 8\,908 & 31\,483 & 90 & 106 \\
Alt-right & 658 & 4\,685 & 41 & 71 \\
Incel & 44 & 322 & 13 & 28 \\
IDW & 6\,720 & 19\,146 & 79 & 85 \\
MGTOW & 431 & 6\,863 & 49 & 71 \\
MRA & 167 & 1\,522 & 17 & 27 \\
NONE & 2\,477 & 6\,590 & 21 & 30 \\
PUA & 4\,414 & 14\,209 & 87 & 119 \\
center & 2\,503 & 10\,117 & 16 & 16 \\
left & 4\,433 & 14\,705 & 16 & 16 \\
left-center & 8\,587 & 33\,617 & 24 & 24 \\
right & 370 & 3\,253 & 6 & 6 \\
right-center & 703 & 4\,060 & 5 & 5 \\
\bottomrule
\end{tabular}

\end{table}

\begin{figure}[t]
	\centering
	\begin{subfigure}{0.5\linewidth}
		\centering
		\includegraphics[width=\linewidth]{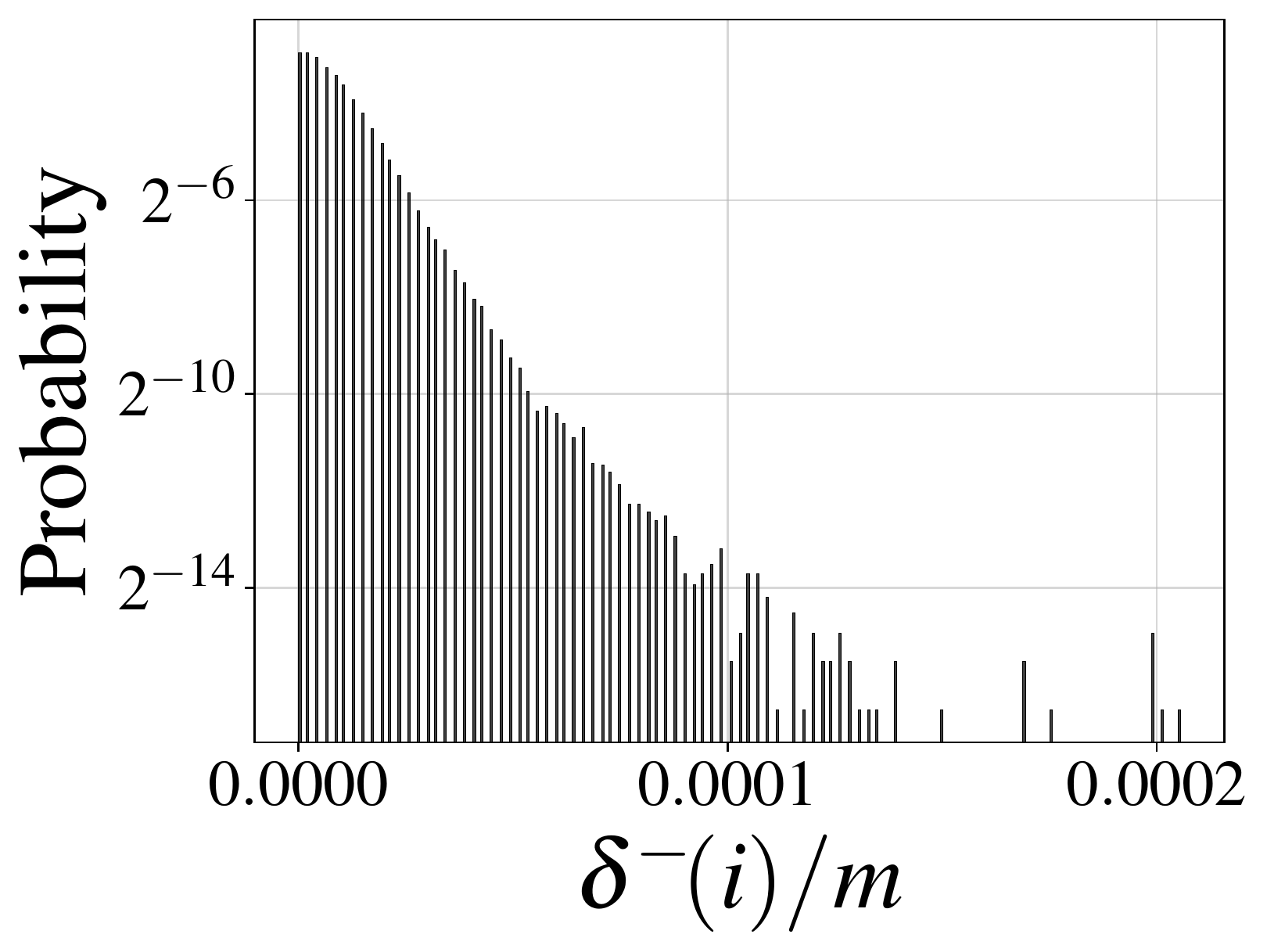}
		\subcaption{\nelathree}
	\end{subfigure}~%
	\begin{subfigure}{0.5\linewidth}
		\centering
		\includegraphics[width=\linewidth]{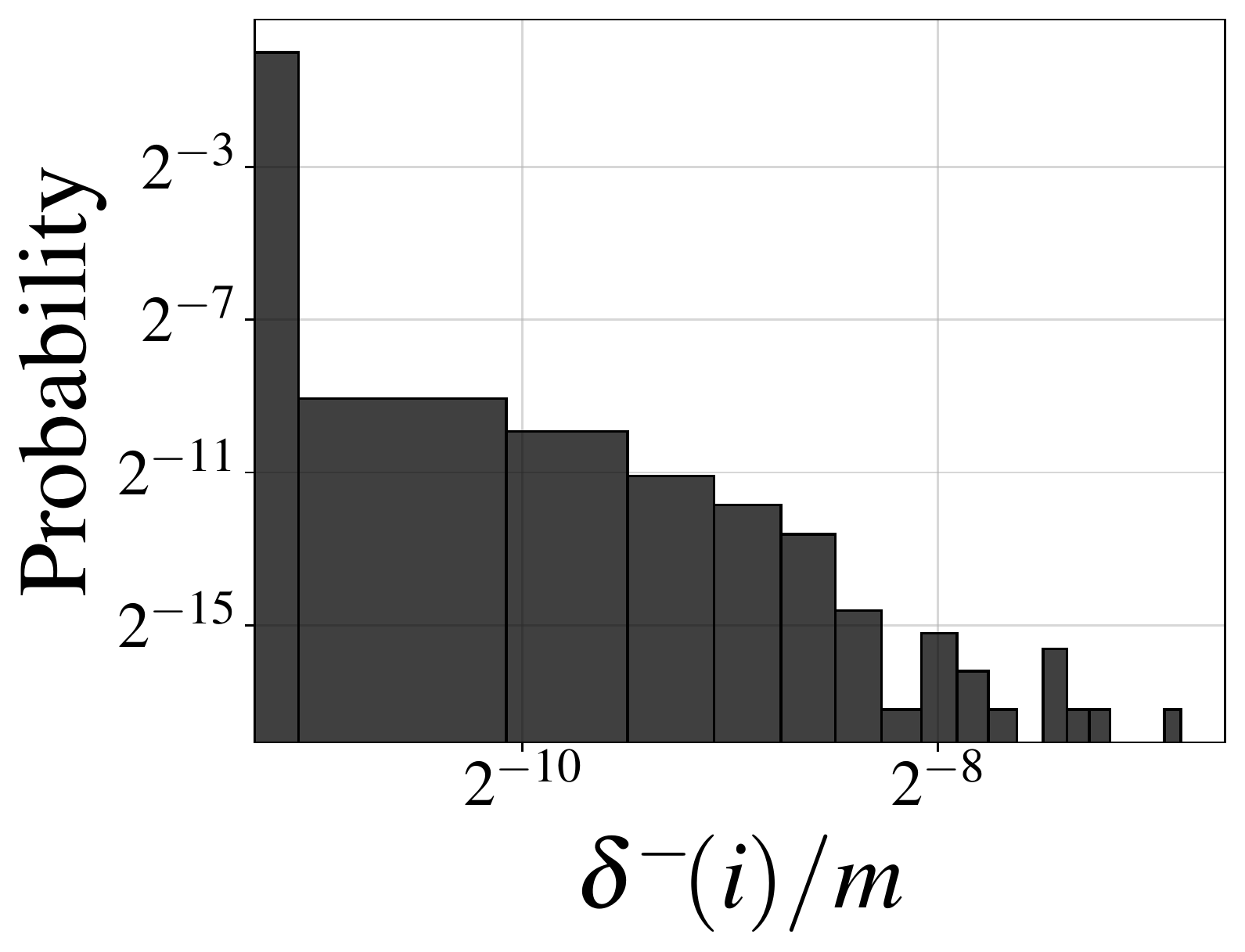}
		\subcaption{\yttwo}
	\end{subfigure}
	\caption{%
		Normalized in-degree distributions of our two largest real-world datasets for $\outregulardegree = 20$. 
		The in-degree distribution of the \yttwo graph is considerably more skewed than the in-degree distribution of the \nelathree graph.
	}\label{fig:real-indegrees-zero}
\end{figure}

\begin{figure}[t]
	\centering
	\begin{subfigure}{0.5\linewidth}
		\centering
		\includegraphics[width=\linewidth]{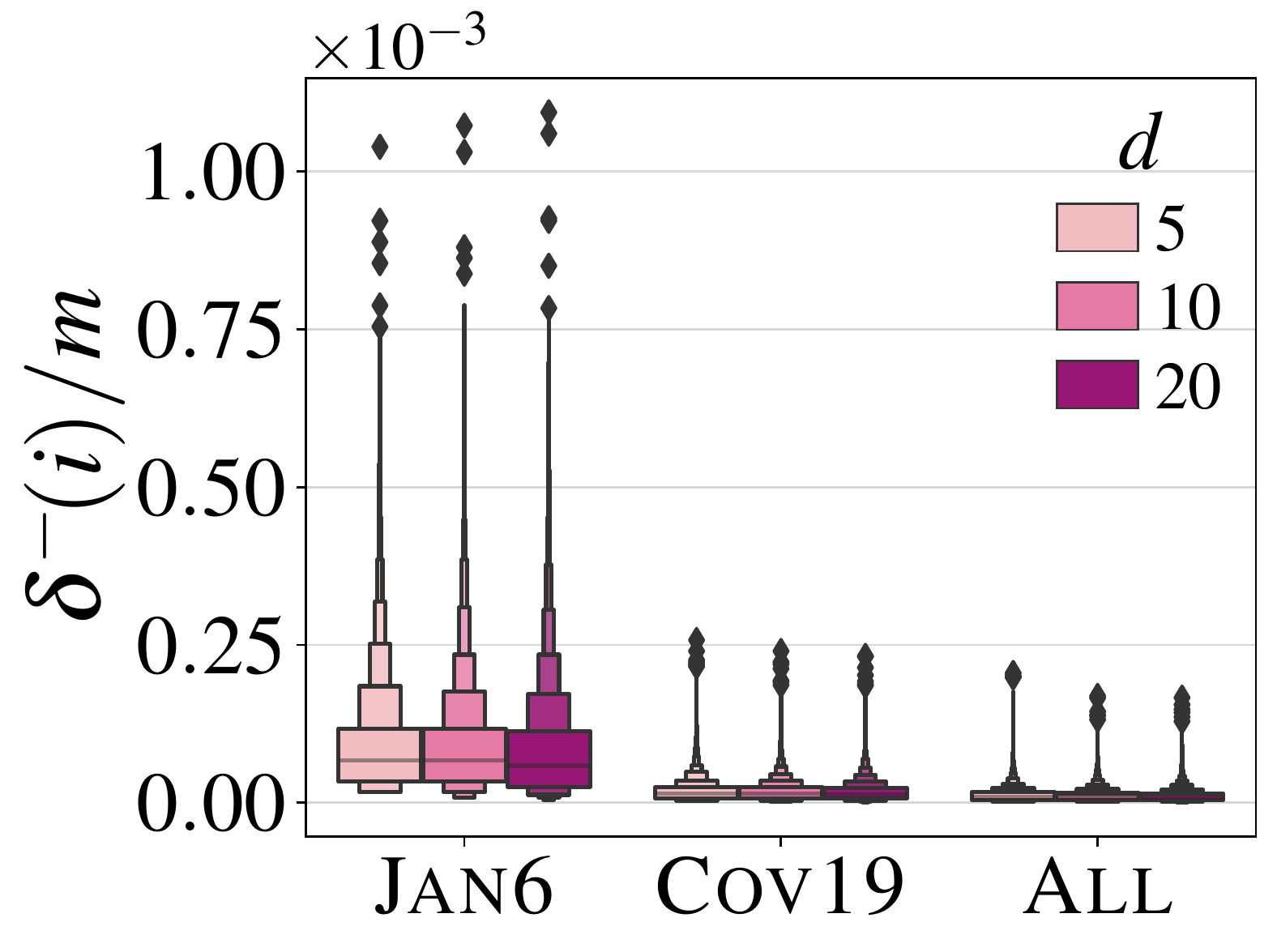}
		\subcaption{\nf}
	\end{subfigure}~%
	\begin{subfigure}{0.5\linewidth}
		\centering
		\includegraphics[width=\linewidth]{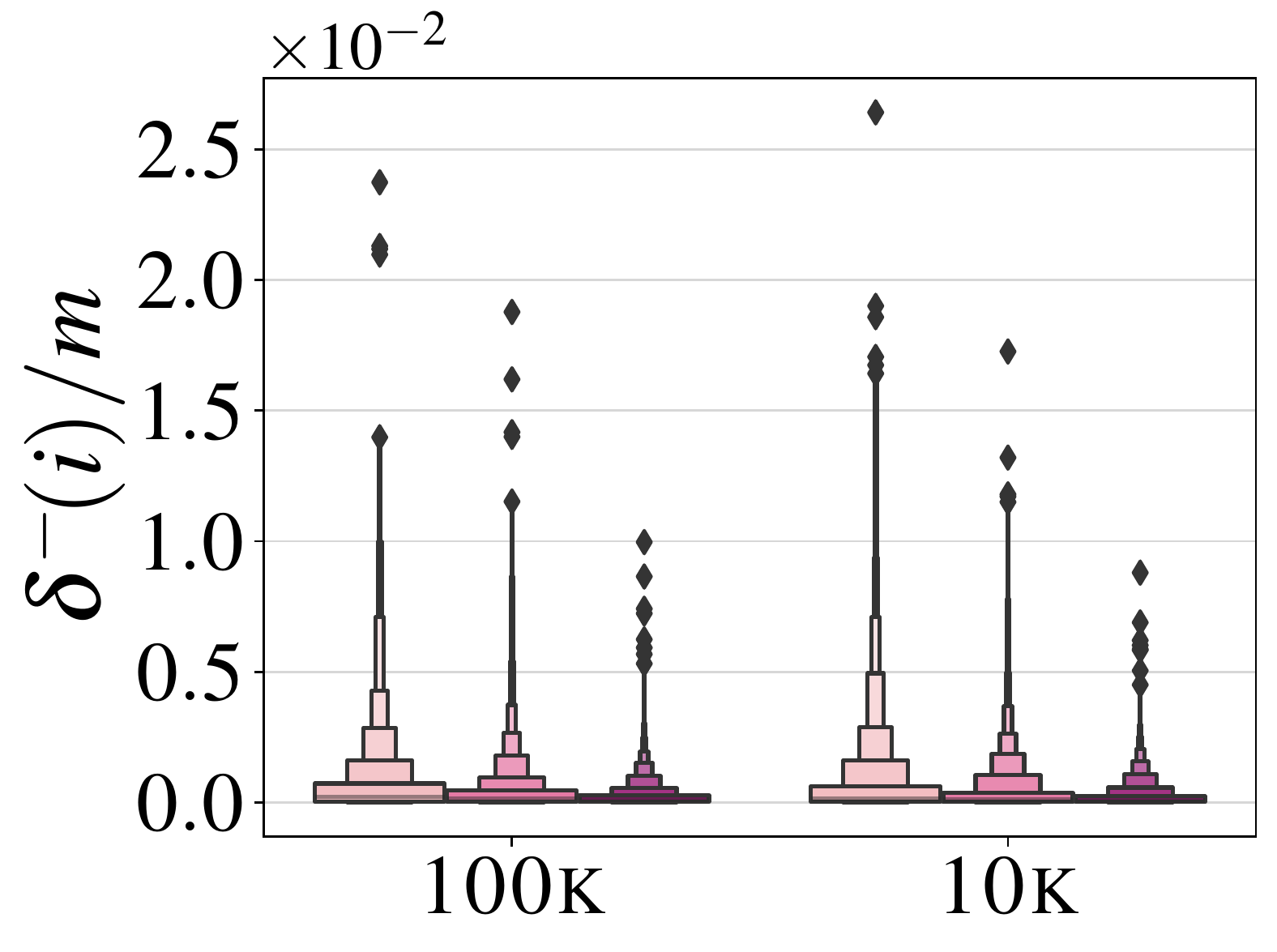}
		\subcaption{\yt}
	\end{subfigure}
	\caption{%
		Distributions of \emph{nonzero} in-degrees $\indegree{i}$, normalized by the number of edges $\nedges$, 
		for each of our real-world datasets. 
		The \nf datasets are an order of magnitude more balanced than the \yt datasets, 
		and within one collection, smaller graphs have more skewed in-degree distributions.
	}\label{fig:real-indegrees}
\end{figure}

\paragraph{Relevance-score distributions}

Complementing the discussion in the main paper, 
in \cref{fig:relevancescores}, we show the relevance score distribution for each of our real-world datasets.
Note that the \ndcg used in our \ourproblemtwo experiments does not expect relevance scores to lie within a particular range, 
and that the relevance scores obtained by preprocessing \yt are not strictly bounded, 
but they are guaranteed to \emph{mostly} lie between $0$ and $1$, 
whereas the relevance scores obtained by preprocessing \nf are directly cosine similarities, rescaled to lie between $0$ and $1$.
As the relevance scores of the \nf datasets are very concentrated, 
the quality threshold $\qualitythreshold$ hardly constrains our rewiring options when solving \ourproblemtwo on \nf graphs.

\begin{table}[b]
	\centering
	\caption{%
		Expected initial exposure $\nicefrac{\objective{\graph}}{\nnodes}$ of nodes in $\graph$ for the YouTube datasets.
	}
	\label{tab:yt-starting-costs}
	\begin{tabular}{lrrrrrrr}
\toprule
 & $d$ & $\alpha$ & $\probabilityshape$ &    $c_{B1}$       &     $c_{B2}$      &    $c_{R1}$        &     $c_{R2}$      \\
\midrule
\ytone & 5  & 0.05 & \textbf{S} &     6.318 &     8.129 &     4.335 &     2.518 \\
&    &      & \textbf{U} &     6.506 &     8.475 &     4.486 &     2.637 \\
&    & 0.10 & \textbf{S} &     3.251 &     4.245 &     2.303 &     1.491 \\
&    &      & \textbf{U} &     3.357 &     4.412 &     2.377 &     1.530 \\
&    & 0.20 & \textbf{S} &     1.694 &     2.234 &     1.245 &     0.900 \\
&    &      & \textbf{U} &     1.737 &     2.297 &     1.272 &     0.908 \\
& 10 & 0.05 & \textbf{S} &     6.387 &     8.316 &     4.440 &     2.688 \\
&    &      & \textbf{U} &     6.605 &     8.701 &     4.623 &     2.842 \\
&    & 0.10 & \textbf{S} &     3.355 &     4.417 &     2.395 &     1.584 \\
&    &      & \textbf{U} &     3.466 &     4.590 &     2.482 &     1.647 \\
&    & 0.20 & \textbf{S} &     1.750 &     2.317 &     1.290 &     0.938 \\
&    &      & \textbf{U} &     1.796 &     2.382 &     1.324 &     0.961 \\
& 20 & 0.05 & \textbf{S} &     6.983 &     9.026 &     4.880 &     3.014 \\
&    &      & \textbf{U} &     7.372 &     9.444 &     5.153 &     3.195 \\
&    & 0.10 & \textbf{S} &     3.606 &     4.716 &     2.582 &     1.722 \\
&    &      & \textbf{U} &     3.749 &     4.874 &     2.687 &     1.801 \\
&    & 0.20 & \textbf{S} &     1.844 &     2.429 &     1.359 &     0.989 \\
&    &      & \textbf{U} &     1.894 &     2.486 &     1.398 &     1.022 \\\midrule
\yttwo  & 5  & 0.05 & \textbf{S} &     4.198 &     5.597 &     2.992 &     1.926 \\
       &    &      & \textbf{U} &     4.173 &     5.785 &     3.040 &     2.066 \\
       &    & 0.10 & \textbf{S} &     2.330 &     3.217 &     1.734 &     1.244 \\
       &    &      & \textbf{U} &     2.401 &     3.369 &     1.801 &     1.315 \\
       &    & 0.20 & \textbf{S} &     1.309 &     1.854 &     1.019 &     0.806 \\
       &    &      & \textbf{U} &     1.353 &     1.922 &     1.053 &     0.834 \\
       & 10 & 0.05 & \textbf{S} &     5.093 &     6.729 &     3.641 &     2.377 \\
       &    &      & \textbf{U} &     5.712 &     7.576 &     4.101 &     2.704 \\
       &    & 0.10 & \textbf{S} &     2.729 &     3.743 &     2.027 &     1.448 \\
       &    &      & \textbf{U} &     2.958 &     4.063 &     2.203 &     1.584 \\
       &    & 0.20 & \textbf{S} &     1.450 &     2.046 &     1.125 &     0.883 \\
       &    &      & \textbf{U} &     1.525 &     2.152 &     1.185 &     0.932 \\
       & 20 & 0.05 & \textbf{S} &     6.185 &     8.186 &     4.405 &     2.820 \\
       &    &      & \textbf{U} &     6.741 &     8.987 &     4.819 &     3.094 \\
       &    & 0.10 & \textbf{S} &     3.120 &     4.285 &     2.310 &     1.625 \\
       &    &      & \textbf{U} &     3.306 &     4.569 &     2.460 &     1.741 \\
       &    & 0.20 & \textbf{S} &     1.577 &     2.228 &     1.222 &     0.949 \\
       &    &      & \textbf{U} &     1.638 &     2.324 &     1.275 &     0.996 \\
\bottomrule
\end{tabular}

\end{table}

\paragraph{Presence of safe nodes}

In \cref{thm:submodularity,cor:apxguarantee}, 
we established that if a graph $\graph$ has at least $\maxunsafeoutdegree$ safe nodes, 
where a node $i$ is safe if its exposure $\unitvector_i^T\fundamental\labelingvector$ is $0$
and $\maxunsafeoutdegree$ is the maximum degree of an unsafe node in $\graph$,
then we can approximate $\maxobjectivef$ up to a factor of $(1-\nicefrac{1}{e})$. 
In \cref{fig:safety}, we demonstrate that under \emph{all} our cost functions,
this applies to \emph{all} \yt graphs and roughly \emph{two thirds} of the \nf graphs, 
with the notable exception of news articles on the topic of January~6 
(i.e., content reporting on the Capitol riot).
Hence, our theoretical approximation guarantee mostly holds also in practice.

\clearpage

\begin{table}[t]
	\centering
	\caption{%
		Number of news outlets $\cardinality{C}$ and number of videos $\cardinality{\nodes}$ in each of our NELA-GT-2021 datasets, 
		for each unique combination of assigned costs.
	}
	\label{tab:nela-cost-counts}
	\begin{tabular}{llllrrrr}
\toprule
  &   &     &     &  & \textsc{NF-Jan6} & \textsc{NF-Cov19} & \textsc{NF-All} \\
$c_{B1}$ & $c_{B2}$ & $c_{R1}$ & $c_{R2}$ &   $|C|$    &            $|V|$      &       $|V|$            &         $|V|$        \\
\midrule
0 & 0 & 0.0 & 0.0 &     4 &              147 &               631 &             993 \\
0 & 0 & 0.2 & 0.0 &    69 &             3188 &             17794 &           29021 \\
0 & 0 & 0.4 & 0.0 &    18 &             1463 &              3986 &            6549 \\
0 & 1 & 0.6 & 0.0 &     1 &                0 &                 0 &               0 \\
0 & 1 & 0.6 & 1.0 &    40 &             3920 &             19398 &           32303 \\
1 & 1 & 0.6 & 0.5 &    73 &             1533 &              8804 &           15549 \\
1 & 1 & 0.8 & 0.5 &   112 &             1497 &              6436 &           14337 \\
1 & 1 & 1.0 & 0.5 &    24 &              237 &               983 &            1987 \\
\bottomrule
\end{tabular}

\end{table}

\begin{table}[b]
	\centering
	\caption{%
		Initial expected total exposure to harm  $\objective{\graph}$, 
		as well as 
		total segregation and maximum segregation from \citeauthor{fabbri2022rewiring} \cite{fabbri2022rewiring}, 
		on each of our \nf graphs with $\pabsorption = 0.05$,  $\probabilityshape=\textbf{U}$, 
		and $\roundingthreshold = 1.0$ for segregation computations under $\labeling_{R1}$ and $\labeling_{R2}$.
	}
	\label{tab:nela-initial-exposures}
	\begin{tabular}{lllrrr}
 \toprule
 &   $\outregulardegree$ &  $\labeling$        &   $\objective{\graph}$ & Total Seg. & Max. Seg. \\
 \midrule
\multirow[c]{12}{*}{\nelaone} & \multirow[c]{4}{*}{5} & $c_{B1}$ & 51\,940 & 6\,896 & 19.99\\
 &  & $c_{B2}$ & 118\,506 & 25\,038 & 19.99\\
 &  & $c_{R1}$ & 104\,948 & 259 & 2.04\\
 &  & $c_{R2}$ & 92\,536 & 8\,555 & 19.99\\
 & \multirow[c]{4}{*}{10} & $c_{B1}$ & 50\,682 & 5\,342 & 5.79\\
 &  & $c_{B2}$ & 114\,141 & 18\,151 & 19.99\\
 &  & $c_{R1}$ & 103\,482 & 250 & 1.47\\
 &  & $c_{R2}$ & 88\,800 & 6\,172 & 10.33\\
 & \multirow[c]{4}{*}{20} & $c_{B1}$ & 50\,319 & 4\,889 & 4.05\\
 &  & $c_{B2}$ & 113\,738 & 16\,194 & 6.50\\
 &  & $c_{R1}$ & 103\,445 & 246 & 1.21\\
 &  & $c_{R2}$ & 88\,579 & 5\,782 & 2.93\\\midrule
\multirow[c]{12}{*}{\nelatwo} & \multirow[c]{4}{*}{5} & $c_{B1}$ & 264\,781 & 54\,699 & 19.99\\
 &  & $c_{B2}$ & 635\,867 & 201\,579 & 19.99\\
 &  & $c_{R1}$ & 520\,763 & 1\,075 & 2.19\\
 &  & $c_{R2}$ & 503\,477 & 81\,903 & 19.99\\
 & \multirow[c]{4}{*}{10} & $c_{B1}$ & 252\,294 & 44\,313 & 19.99\\
 &  & $c_{B2}$ & 618\,645 & 157\,105 & 19.99\\
 &  & $c_{R1}$ & 513\,982 & 1\,038 & 1.86\\
 &  & $c_{R2}$ & 492\,498 & 65\,199 & 19.99\\
 & \multirow[c]{4}{*}{20} & $c_{B1}$ & 248\,708 & 38\,597 & 19.99\\
 &  & $c_{B2}$ & 617\,014 & 134\,998 & 19.99\\
 &  & $c_{R1}$ & 513\,113 & 1\,020 & 1.48\\
 &  & $c_{R2}$ & 492\,660 & 56\,696 & 19.05\\\midrule
\multirow[c]{12}{*}{\nelathree} & \multirow[c]{4}{*}{5} & $c_{B1}$ & 520\,092 & 128\,388 & 19.99\\
 &  & $c_{B2}$ & 1\,111\,742 & 383\,640 & 19.99\\
 &  & $c_{R1}$ & 890\,383 & 3\,013 & 19.99\\
 &  & $c_{R2}$ & 851\,697 & 136\,356 & 19.99\\
 & \multirow[c]{4}{*}{10} & $c_{B1}$ & 496\,667 & 103\,825 & 19.99\\
 &  & $c_{B2}$ & 1\,089\,690 & 307\,444 & 19.99\\
 &  & $c_{R1}$ & 880\,536 & 2\,666 & 15.10\\
 &  & $c_{R2}$ & 841\,357 & 111\,298 & 19.99\\
 & \multirow[c]{4}{*}{20} & $c_{B1}$ & 480\,186 & 88\,983 & 19.99\\
 &  & $c_{B2}$ & 1\,076\,287 & 260\,998 & 19.99\\
 &  & $c_{R1}$ & 874\,785 & 2\,187 & 6.90\\
 &  & $c_{R2}$ & 836\,194 & 96\,895 & 19.99\\
 \bottomrule
\end{tabular}

\end{table}

\begin{figure}[t]
	\centering
	\includegraphics[height=3cm]{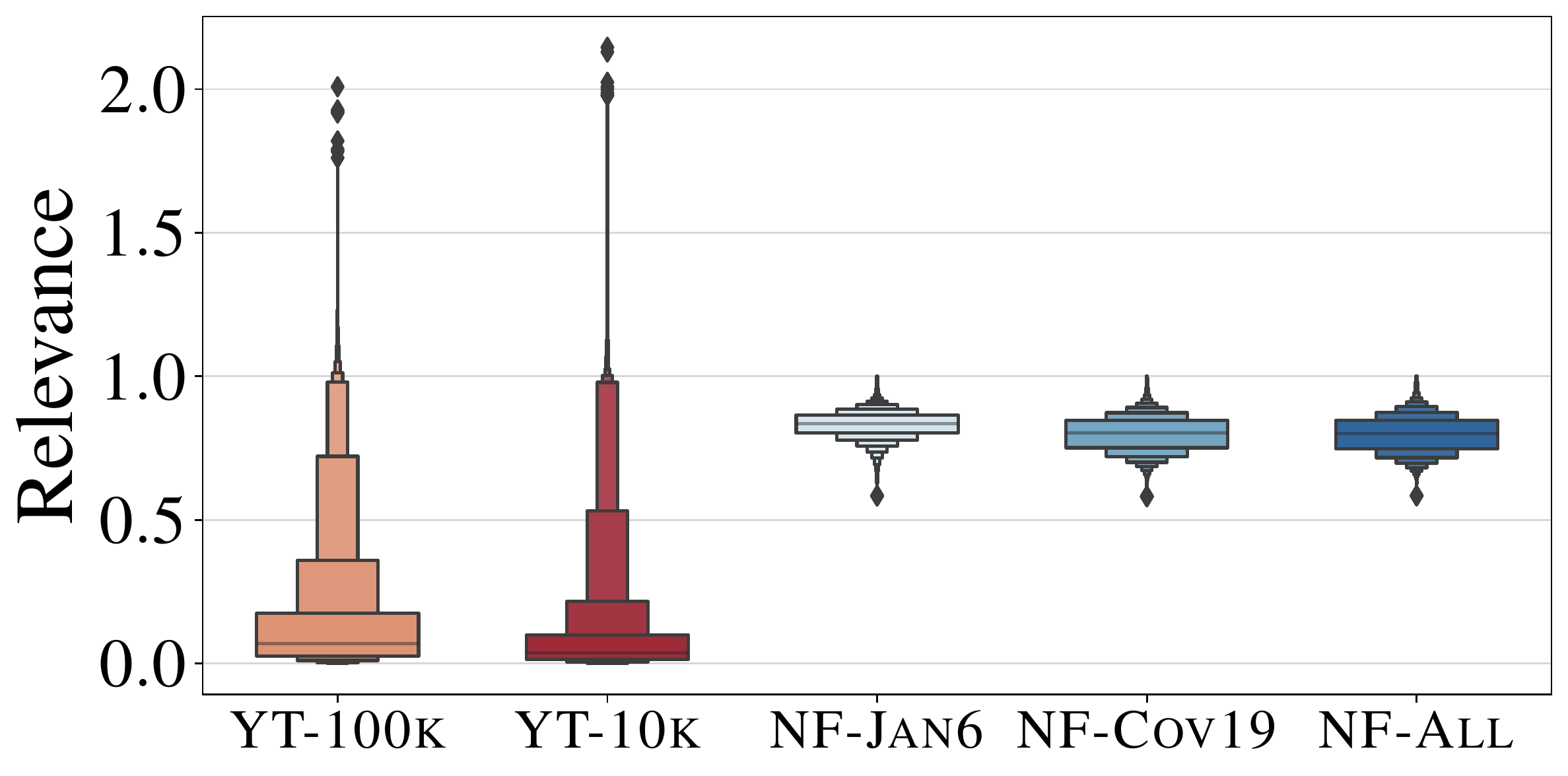}
	\caption{%
		Distributions of relevance scores $\relevancematrix[i,j]$ for each source node $i\in\nodes$ and the top $100$ nodes $j$ considered as potential rewiring targets, 
		for each of our real-world datasets in the \ourproblemtwo setting. 
		The relevance distributions are much more skewed in the \yt datasets than in the \nf datasets.
	}
	\label{fig:relevancescores}
\end{figure}

\begin{figure}[b]
	\centering
	\begin{subfigure}{0.478\linewidth}
		\centering
		\includegraphics[width=\linewidth]{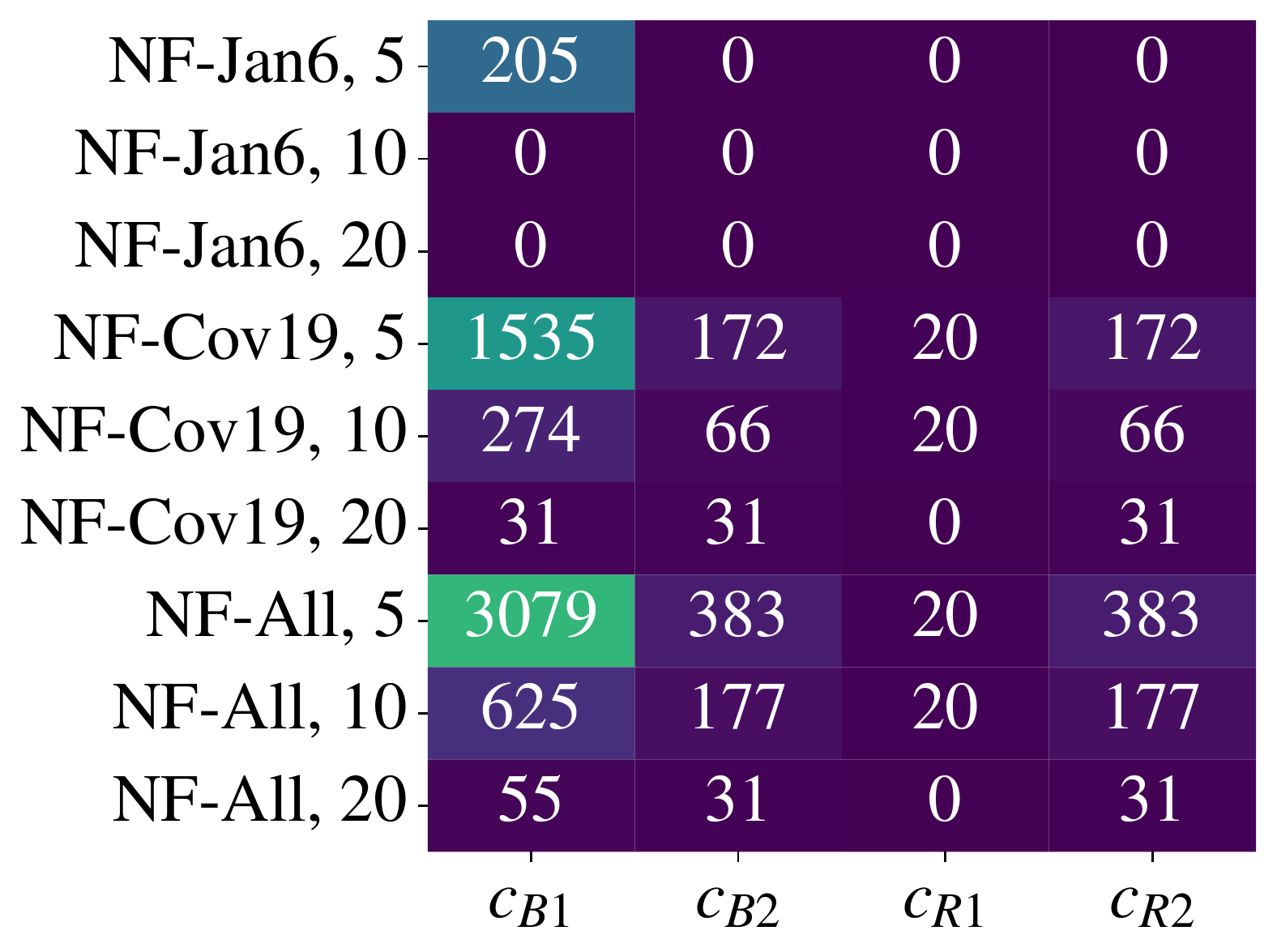}
		\subcaption{\nf}
	\end{subfigure}%
	\begin{subfigure}{0.5\linewidth}
		\centering
		\includegraphics[width=\linewidth]{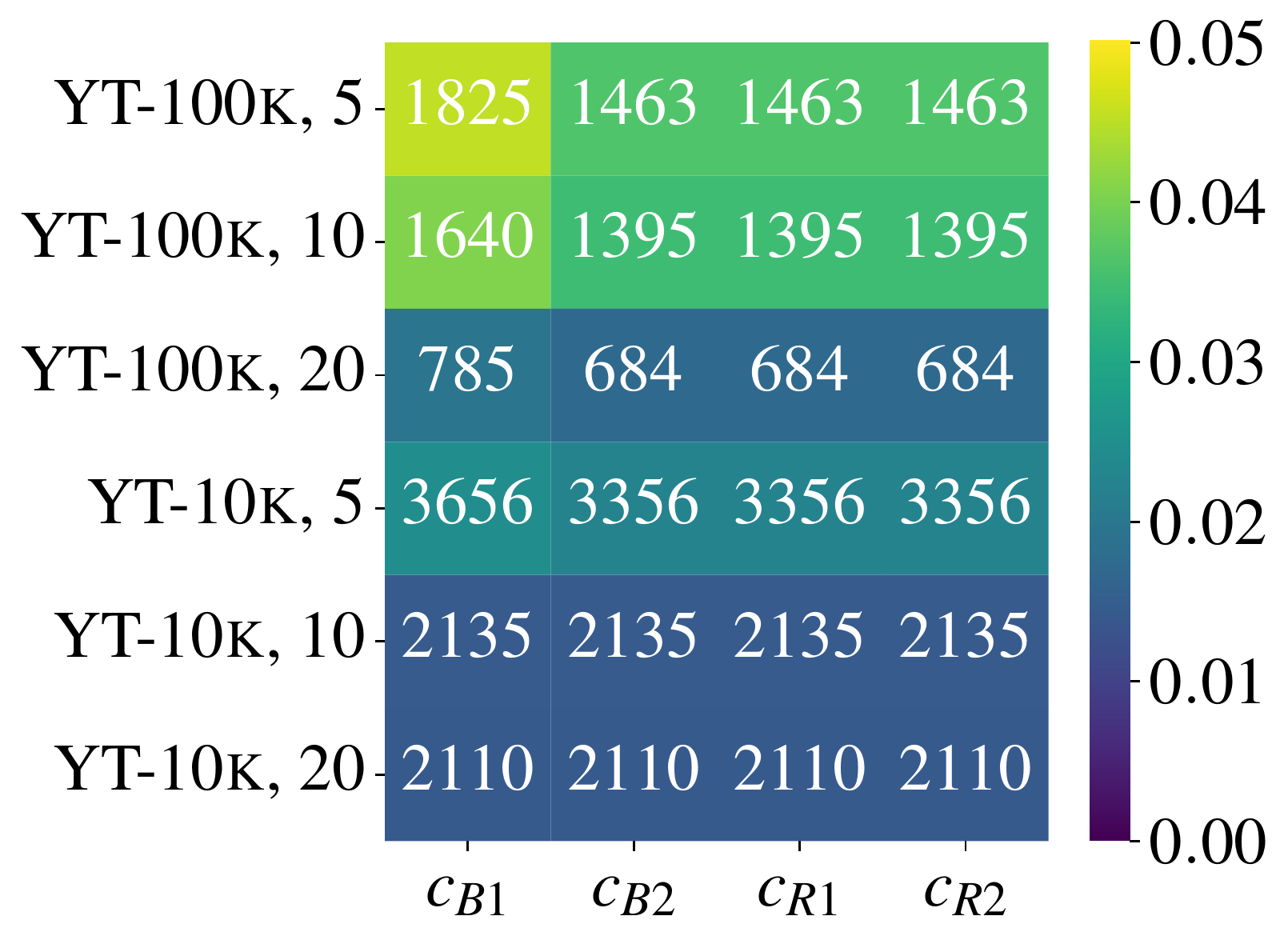}
		\subcaption{\yt}
	\end{subfigure}
	\caption{Fraction (color) and number (annotation) of safe nodes in each of our real-world graphs with degrees $\outregulardegree\in\{5,10,20\}$, under our four cost functions $\labeling\in\{\labeling_{B1},\labeling_{B2},\labeling_{R1},\labeling_{R2}\}$.
		In all \yt graphs and roughly two thirds of the \nf graphs, the precondition of \cref{thm:submodularity} holds, such that the greedy algorithm can approximate $\maxobjectivef$ up to a factor of $(1-\frac{1}{e})$.
	}\label{fig:safety}
\end{figure}

\paragraph{Impact of cost-function noise}

To see how errors in harmfulness assessment might impact our ability to rewire edges effectively,
we investigate the behavior of our exposure objective under noise in the cost function.
In particular, we assess how the distribution of node exposures shifts when we change the original cost vector $\labelingvector$ to a cost vector $\labelingvector'$ by
either swapping the cost of a randomly chosen harmful node with that of a randomly chosen benign node (\emph{cost swaps}), 
or setting the cost $\labeling_i$ of a randomly chosen node $i$ to its opposite, i.e., $1-\labeling_i$ (\emph{cost flips}). 
Illustrating the results on the \ytone dataset in \cref{fig:noise}, 
we observe that as expected---and \emph{by} \emph{construction}---,
node exposures are generally sensitive to individual cost assignments.
However, the median impact of moderate cost-function noise on node exposure levels is close to zero, 
and the most extreme cost fluctuations occur for nodes whose \emph{observed} exposure decreases as compared to their \emph{actual} exposure. 
The latter might lead \ourmethod to undervalue some highly exposed nodes in its rewiring considerations, 
but this risk is unavoidable when dealing with noisy data. 
In contrast to prior work, 
\ourmethod uses an exposure objective that depends on the cost assignments of \emph{all} nodes in the graph. 
Overall, our experiments with node-level cost-function noise demonstrate that this objective decays rather smoothly---%
not only in theory but also in practice.

\clearpage

\begin{figure}[t]
	\centering
	\begin{subfigure}{0.5\linewidth}
		\centering
		\includegraphics[width=\linewidth]{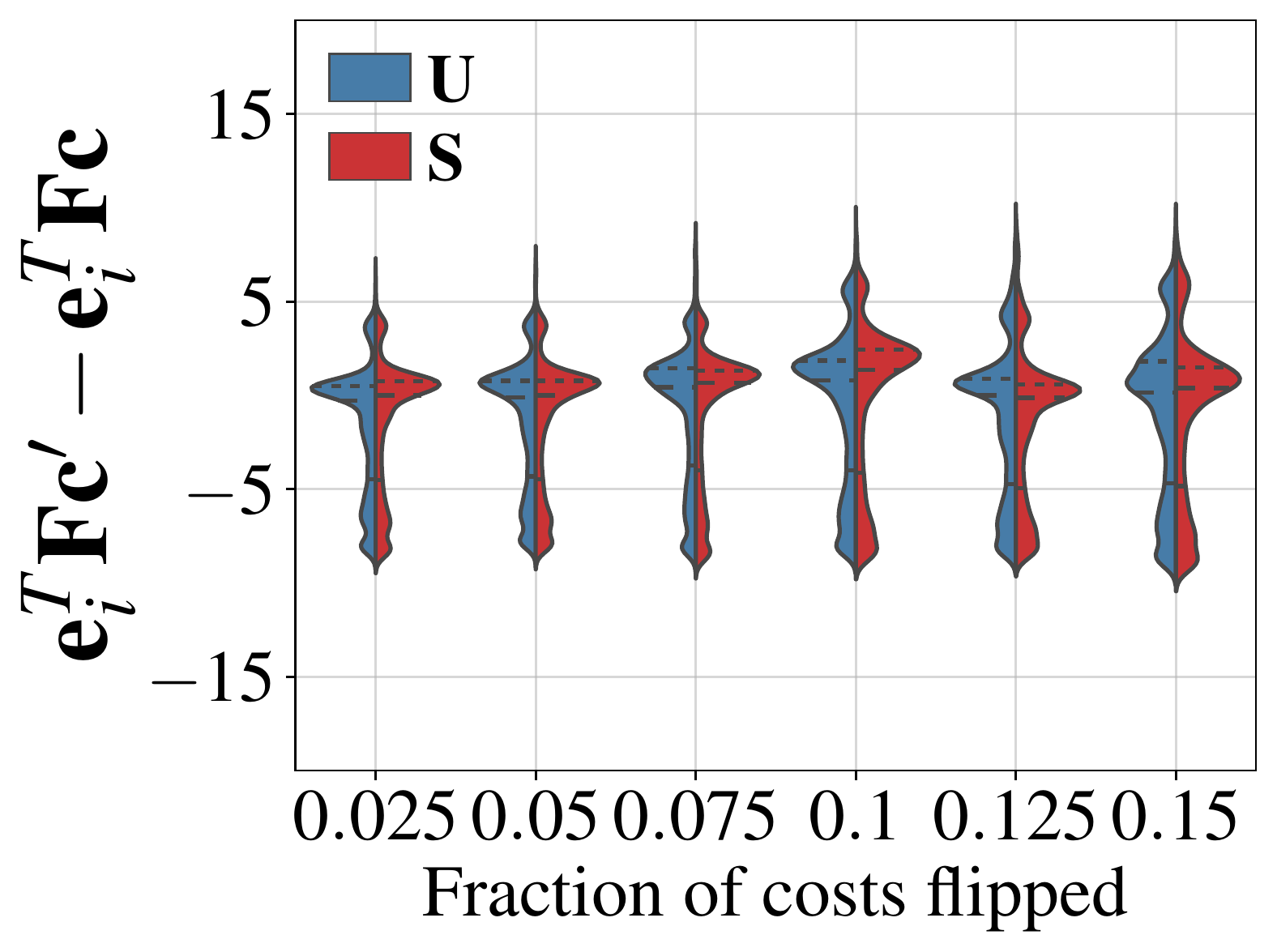}
		\subcaption{Cost flips}
	\end{subfigure}%
	\begin{subfigure}{0.5\linewidth}
		\centering
		\includegraphics[width=\linewidth]{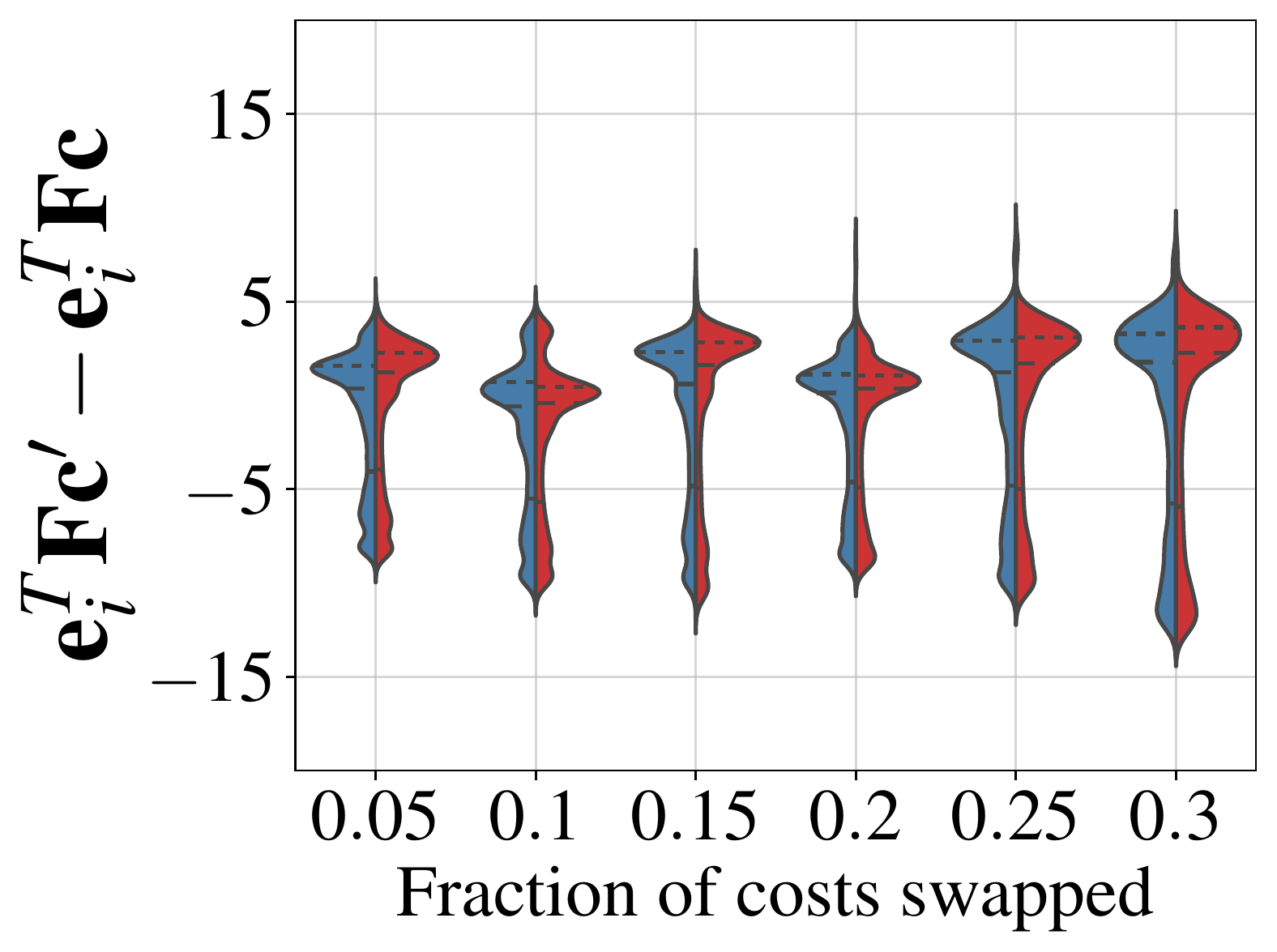}
		\subcaption{Cost swaps}
	\end{subfigure}
	\caption{%
		Distribution of differences between node exposures before and after the introduction of noise,
		shown for the \ytone dataset with $\outregulardegree=5$, $\pabsorption=0.05$,
		and $\probabilityshape\in\{\mathbf{U},\mathbf{S}\}$,
		as measured under $\labeling_{R1}$. 
		\emph{Negative} values signal that adding noise \emph{decreased} the exposure to harm of a particular node.
	}\label{fig:noise}
\end{figure}

\section{Further Experiments}
\label{apx:experiments}

\subsection{Impact of Modeling Choices}
\label{apx:exp:modelingchoices}

In the main text, we only demonstrated the impact of the quality threshold $\qualitythreshold$ on the performance of \ourmethod.
Here, we further discuss the performance impact of 
the regular out-degree $\outregulardegree$, 
the absorption probability $\pabsorption$, 
the shape of the probability distribution over out-edges $\probabilityshape$, 
and the cost function $\labeling$.

\paragraph{Impact of regular out-degree $\outregulardegree$}
Since the impact of individual edges on the objective function decreases as $\outregulardegree$ increases, 
for a given budget $\budget$,
we expect \ourmethod to reduce our objective more strongly for smaller values of $\outregulardegree$.
This is exactly what we find, as illustrated in \cref{fig:outregulardegree}, 
and the pattern persists across absorption probabilities $\pabsorption$, probability shapes $\probabilityshape$, quality thresholds $\qualitythreshold$, and cost functions $\labeling$.

\paragraph{Impact of absorption probability $\pabsorption$ and out-edge probability distribution shape $\probabilityshape$}
For smaller random-walk absorption probabilities $\pabsorption$, 
we obtain longer random walks and thus higher exposure to harmful content, 
and for $\probabilityshape = \mathbf{S}$, 
some edges are traversed particularly often.
Thus, given a constant budget $\budget$, 
we expect \ourmethod to achieve a \emph{larger} decrease of $\objectivef$ for smaller $\pabsorption$, 
and an initially \emph{faster} decrease on graphs with skewed out-edge probability distributions.
Again, this is what we find,
as depicted in \cref{fig:absorptionshape}.

\paragraph{Impact of cost function $\labeling$}
As the binary cost function $c_{B1}$ (used also in \cite{fabbri2022rewiring} on a prior version of the data from \cite{ribeiro2020auditing}) 
labels only videos from Alt-Right, Alt-Lite, and Intellectual Dark Web (IDW) channels as harmful ($\labeling_{B1} = 1$) and all other videos as benign ($\labeling_{B1} = 0$), 
whereas all other cost functions also assign positive cost to videos from anti-feminist channels (Incel, MGTOW, MRA, and PUA) (cf. \cref{tab:yt-costs}),
we expect \ourmethod to perform strongest under $\labeling_{B1}$. 
As exemplified in \cref{fig:cost_functions}, 
this is exactly what we observe, 
and the pattern persists across regular out-degrees $\outregulardegree$, 
absorption probabilities $\pabsorption$, 
distribution shapes $\probabilityshape$, 
and quality thresholds $\qualitythreshold$.  
Interestingly, we also consistently observe that \ourmethod is roughly equally strong under the binary cost function $\labeling_{B2}$ and the real-valued cost function $\labeling_{R1}$, 
and weakest under the real-valued cost function $\labeling_{R2}$.
As $\labeling_{R1}$ and $\labeling_{R2}$ differ only in how they assign costs to videos from IDW and anti-feminist channels, 
with $\labeling_{R1}$ ($\labeling_{R2}$) placing IDW to the \emph{right} (\emph{left}) of anti-feminist channels,
this means that reducing the exposure to harm is \emph{harder} when we consider the IDW more benign than anti-feminist communities, 
\emph{even though} there are more IDW videos in \ytone than videos from all anti-feminist communities combined. 

\begin{figure}[b]
	\centering
	\begin{subfigure}{0.5\linewidth}
		\centering
		\includegraphics[width=\linewidth]{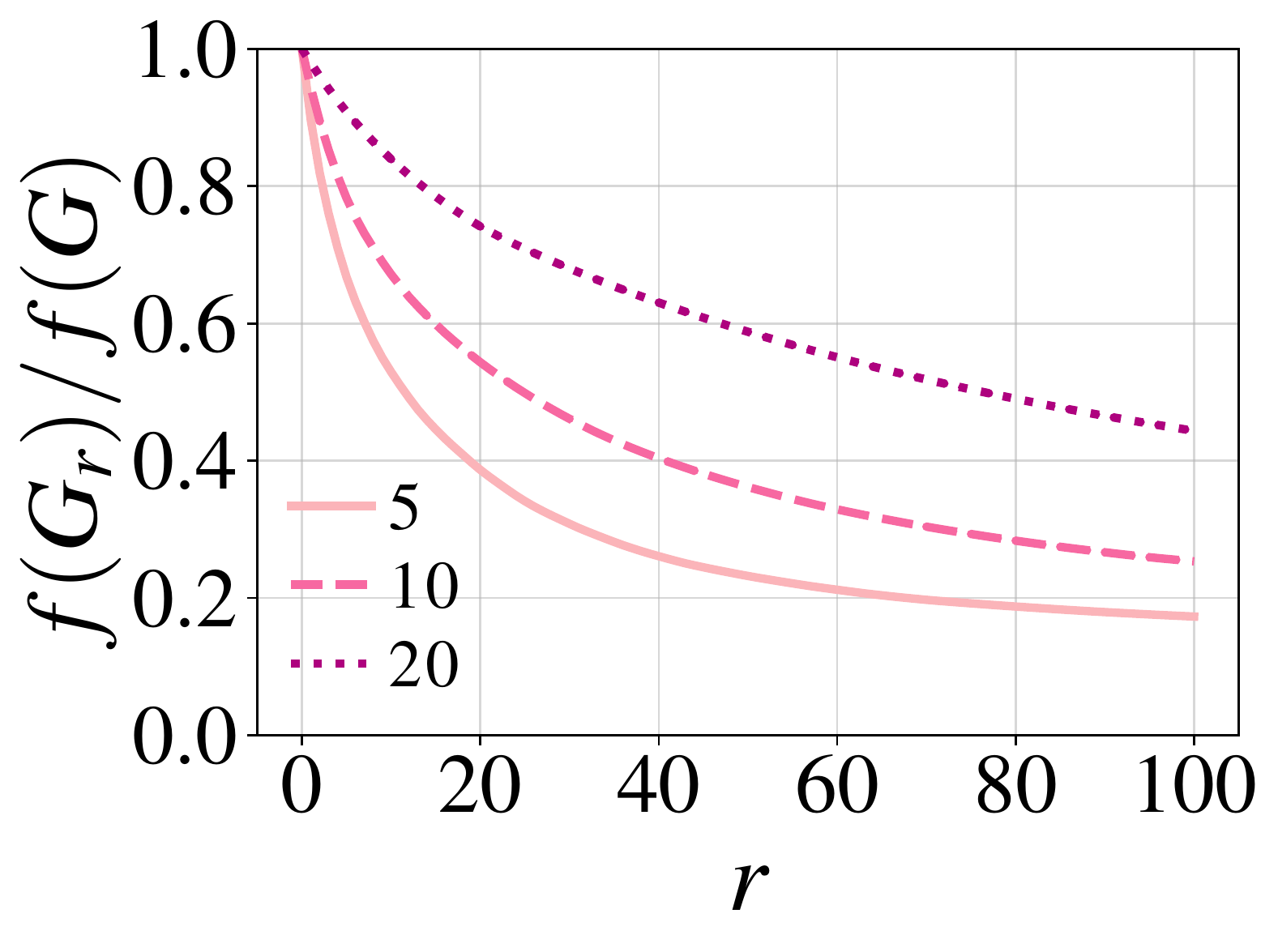}
		\subcaption{Absorption probability $\pabsorption = 0.05$}
	\end{subfigure}~%
	\begin{subfigure}{0.5\linewidth}
		\centering
		\includegraphics[width=\linewidth]{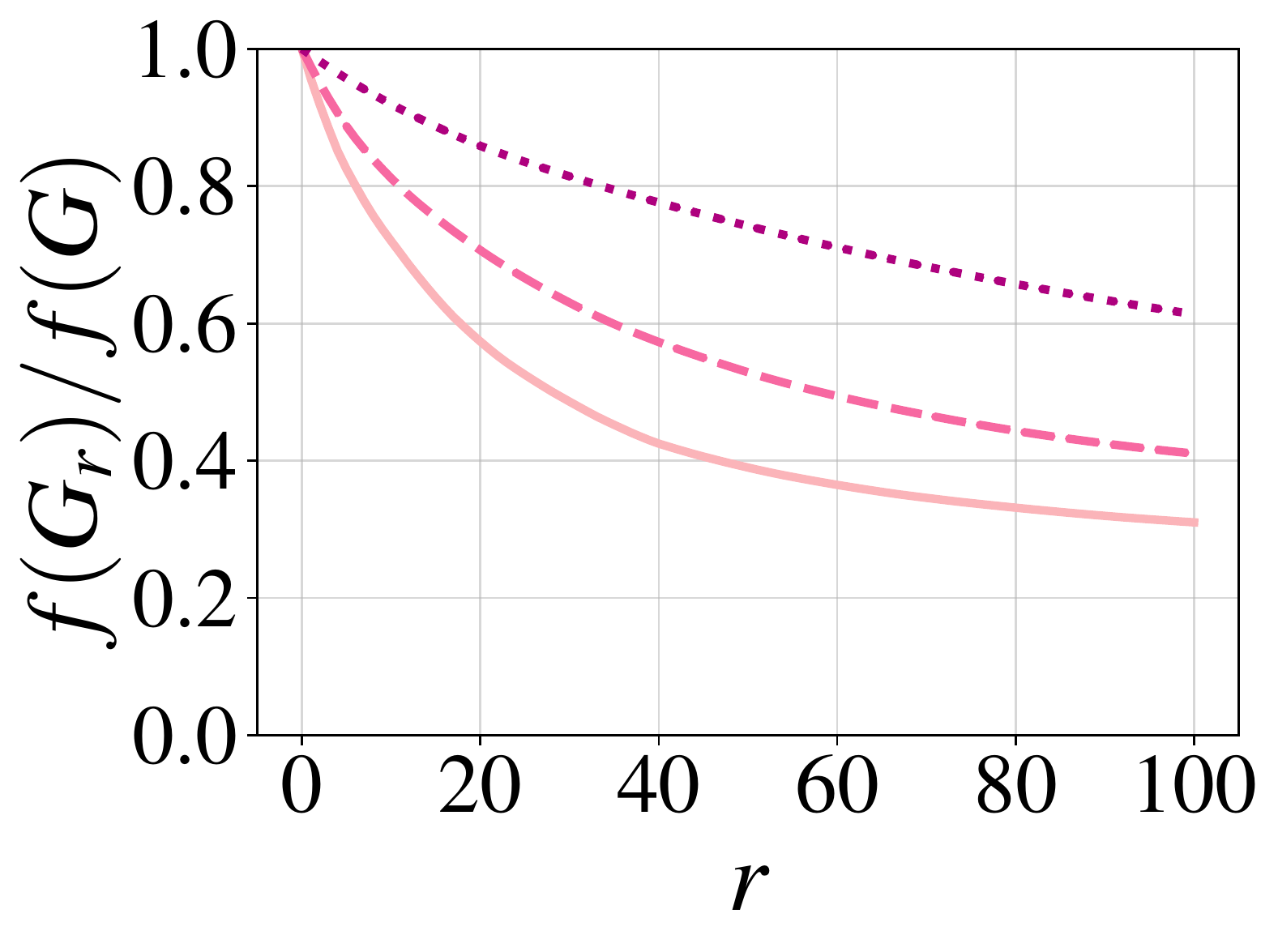}
		\subcaption{Absorption probability $\pabsorption = 0.1$}
	\end{subfigure}
	\caption{%
		Performance of \ourmethod for out-regular degrees $\outregulardegree\in\{5,10,20\}$, 
		run with $\qualitythreshold = 0.0$ under $c_{B1}$ on \ytone with $\probabilityshape=\mathbf{U}$.
		The smaller the out-degree, the stronger \ourmethod.
	} 
	\label{fig:outregulardegree}
\end{figure}
\begin{figure}[b]
	\centering
	\begin{subfigure}{0.5\linewidth}
		\centering
		\includegraphics[width=\linewidth]{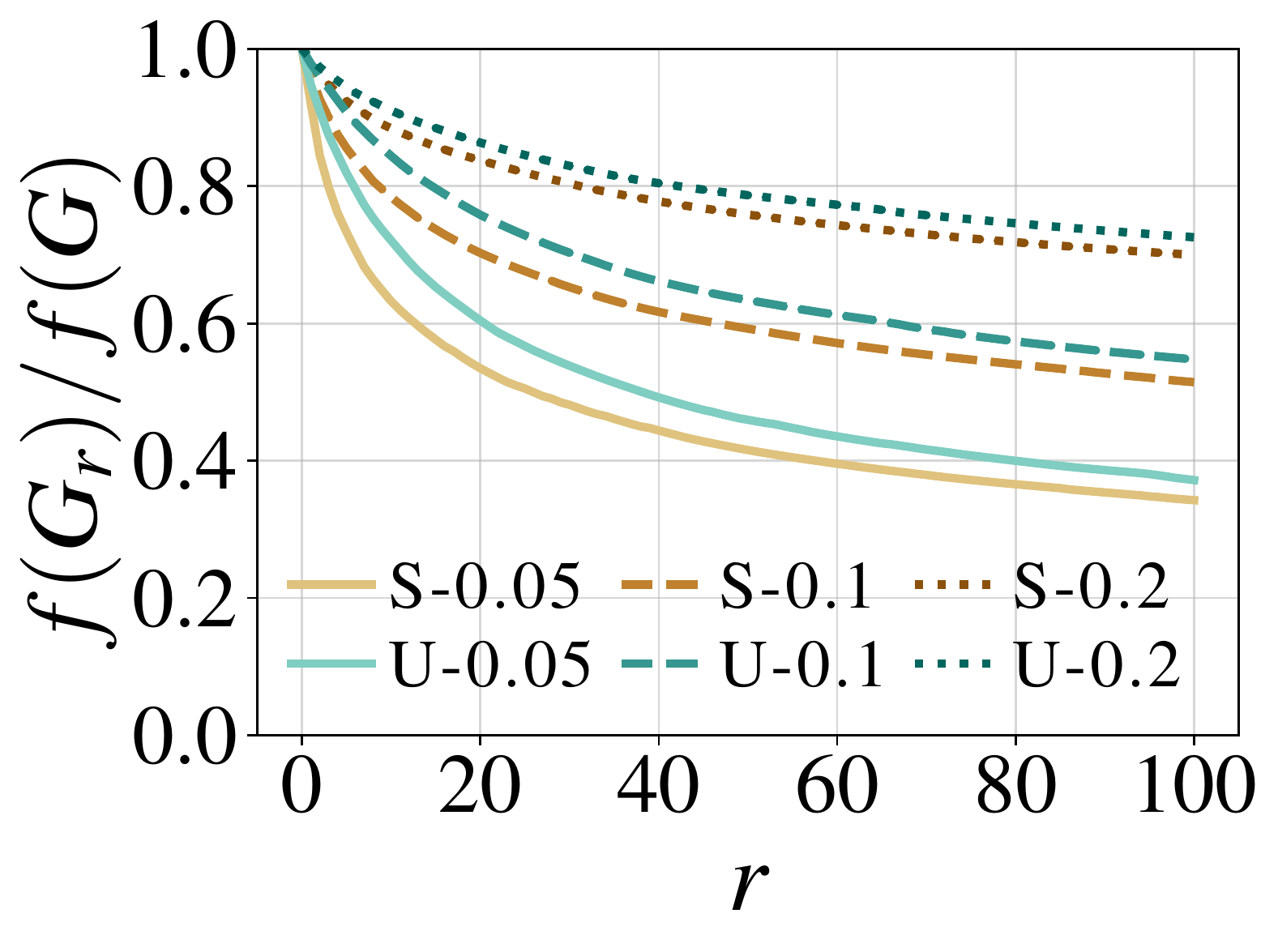}
		\subcaption{Cost function $\labeling_{R1}$}
	\end{subfigure}~%
	\begin{subfigure}{0.5\linewidth}
		\centering
		\includegraphics[width=\linewidth]{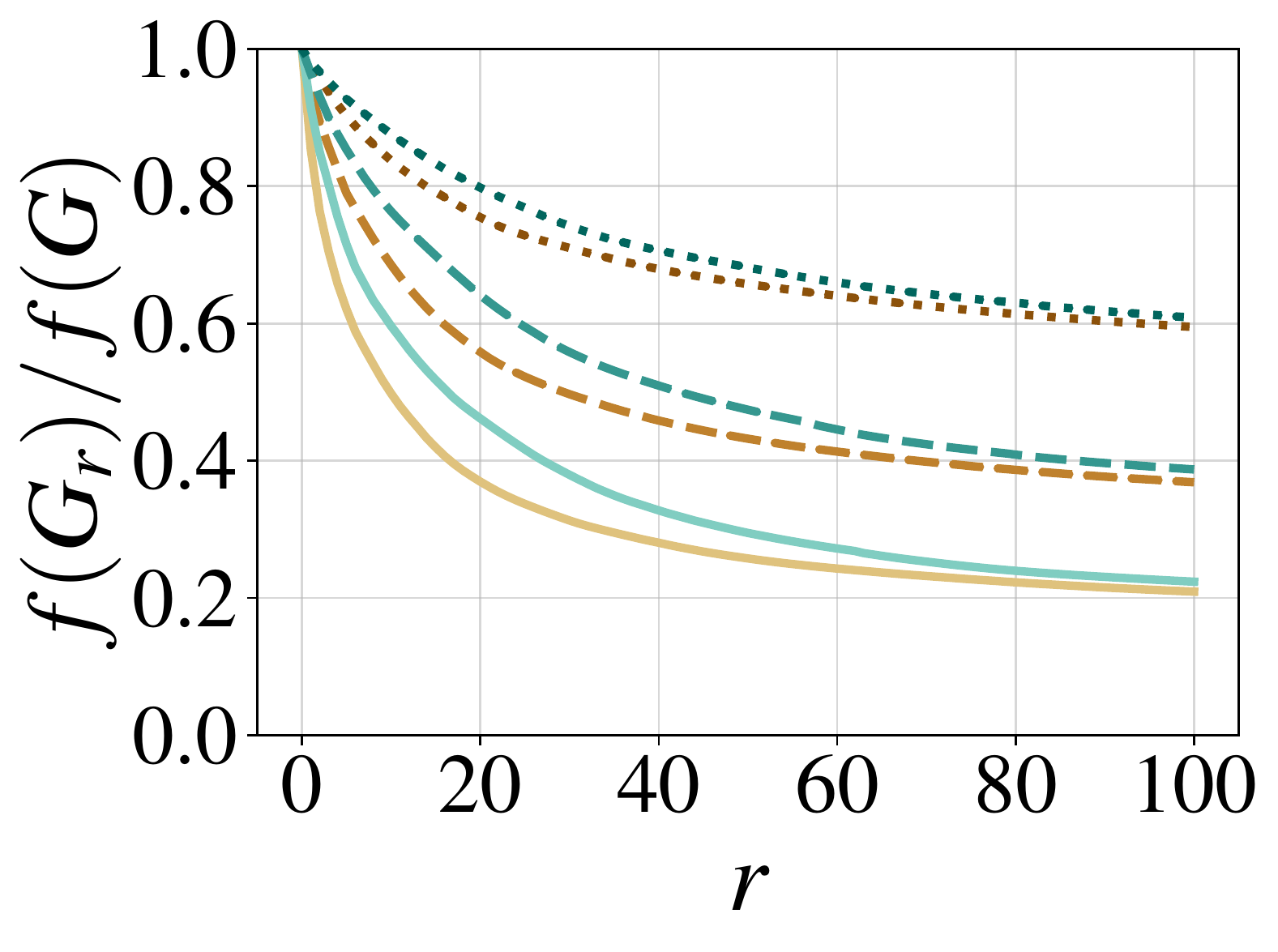}
		\subcaption{Cost function $\labeling_{B1}$}
	\end{subfigure}
	\caption{%
		Performance of \ourmethod
		for absorption probabilities $\pabsorption \in \{0.05,0.1,0.2\}$ and out-edge probability distribution shapes $\probabilityshape \in \{\mathbf{U}, \mathbf{S}\}$,
		run with $\qualitythreshold = 0.5$ on \ytone with $\outregulardegree = 5$. 
		The smaller the absorption probability, the stronger the performance of \ourmethod, 
		and our objective function drops faster when the out-edge probability distribution is skewed.
	}
	\label{fig:absorptionshape}
\end{figure}

\begin{figure}[t]
	\centering
	\begin{subfigure}{0.5\linewidth}
		\centering
		\includegraphics[width=\linewidth]{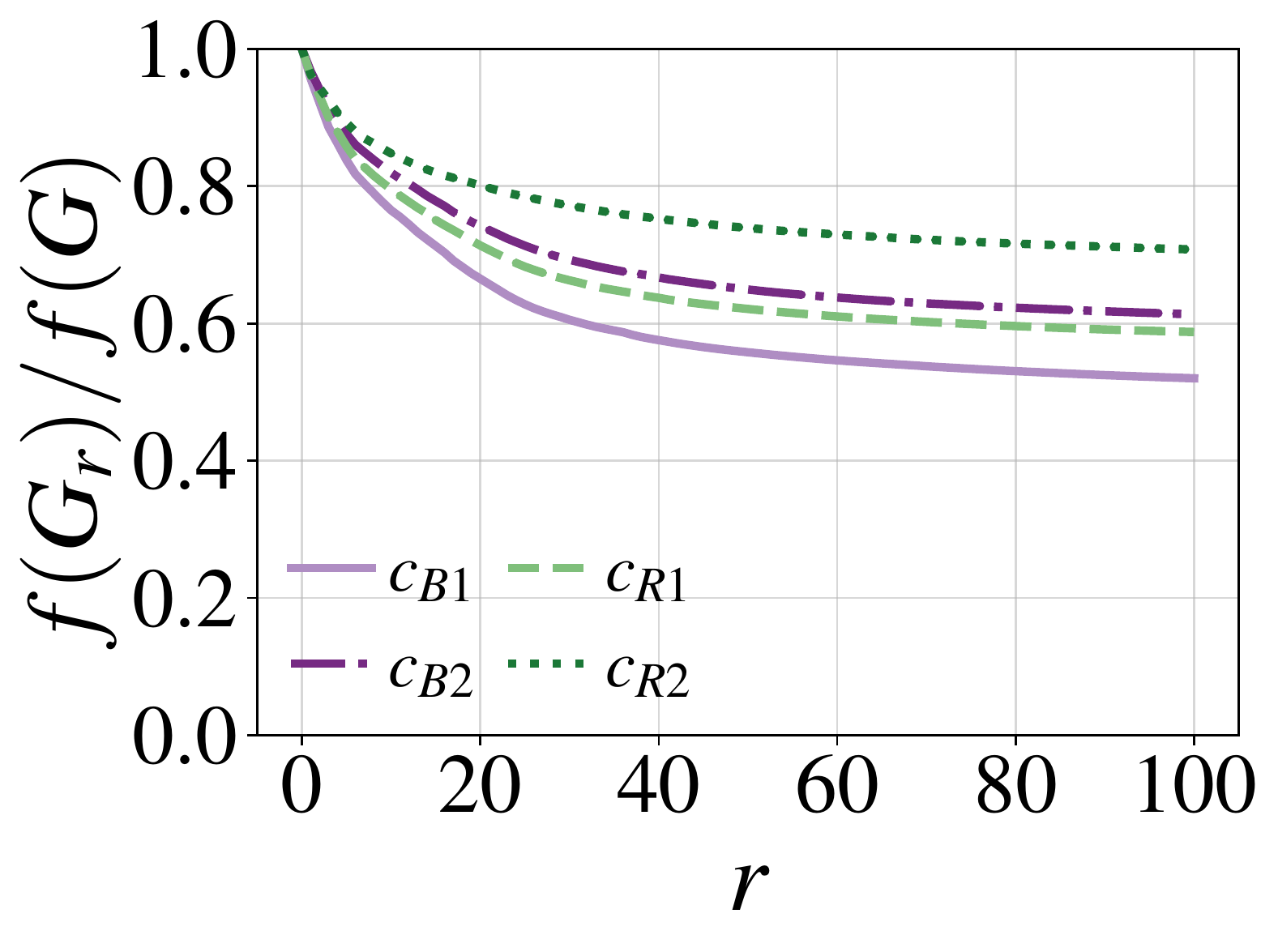}
		\subcaption{Quality threshold $\qualitythreshold=0.99$}
	\end{subfigure}~%
	\begin{subfigure}{0.5\linewidth}
		\centering
		\includegraphics[width=\linewidth]{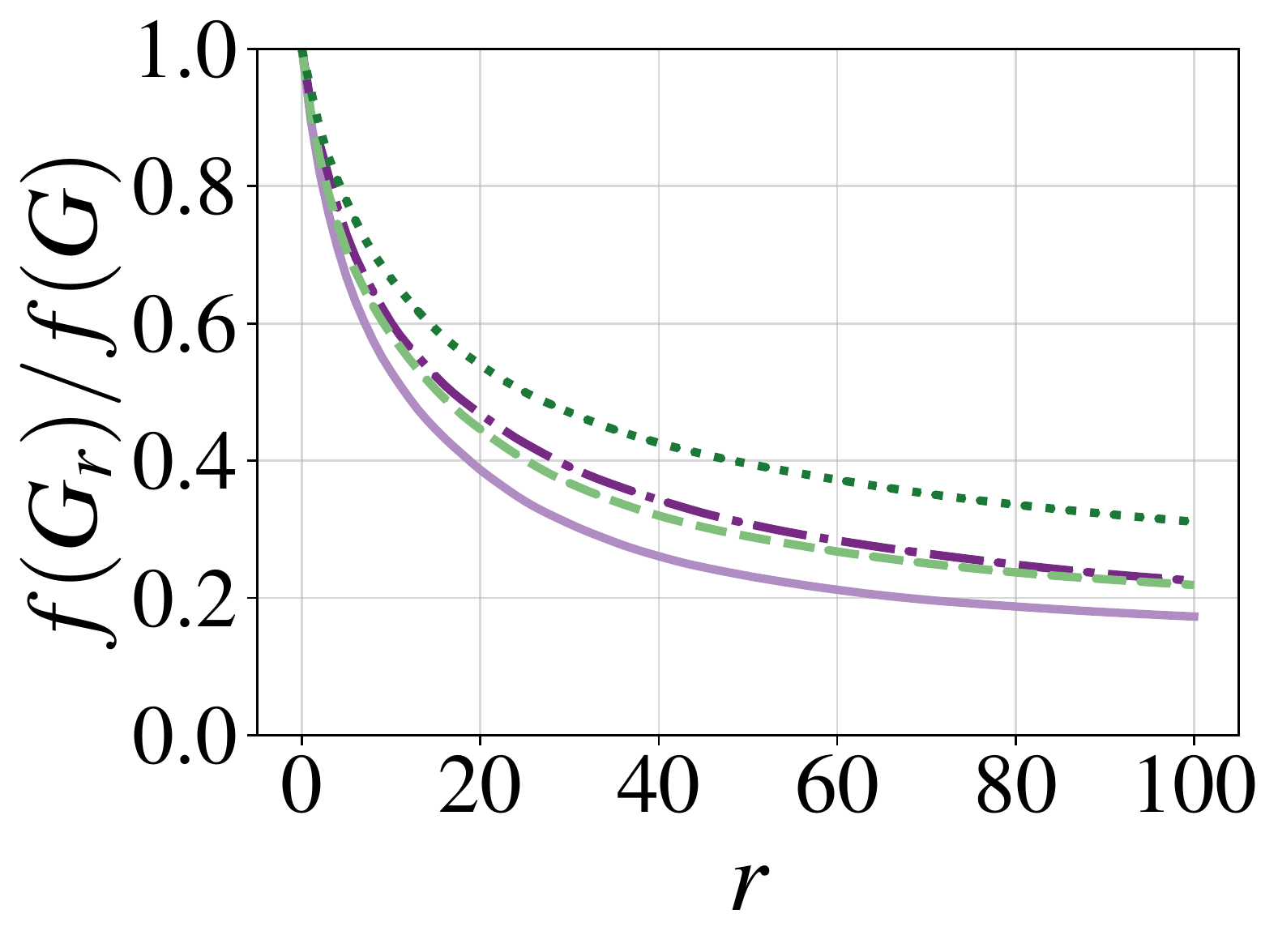}
		\subcaption{Quality threshold $\qualitythreshold=0.0$}
	\end{subfigure} 
	\caption{%
		Performance of \ourmethod under cost functions $c\in\{c_{B1},c_{B2},c_{R1},c_{R2}\}$, 
		run on \ytone with $\outregulardegree=5$, $\pabsorption=0.05$, and $\probabilityshape = \mathbf{U}$.
		\ourmethod is strongest under the binary cost function $\labeling_{B1}$, 
		weakest under the real-valued cost function $\labeling_{R2}$, 
		and roughly equally strong under the binary $\labeling_{B2}$ and the real-valued $\labeling_{R1}$.
	}\label{fig:cost_functions}
\end{figure}

\subsection{Scalability}
\label{apx:exp:scalability}

\subsubsection{Precomputations}

In \cref{fig:scalability} in the main text, we showed that \ourmethod's individual edge rewirings scale approximately linearly in practice, 
whereas \fabbrialg's individual edge rewirings scale quadratically.
In \cref{fig:precomputations}, we additionally show that precomputations add approximately linear overhead for \ourmethod and somewhat unpredictable, at times quadratic overhead for \fabbrialg. 
This could be due to two factors.
First, the relevance precomputations for \fabbrialg are slightly more complicated than for \ourmethod. 
Second, one part of \fabbrialg's precomputations not present in \ourmethod is a matrix inverse approximation via power iteration. 
This computation is quadratic in the number of \emph{harmful} nodes, 
as \fabbrialg considers only these nodes as transient states.

\begin{figure}[t]
	\centering
	\begin{subfigure}{0.5\linewidth}
		\centering
		\includegraphics[height=3.025cm]{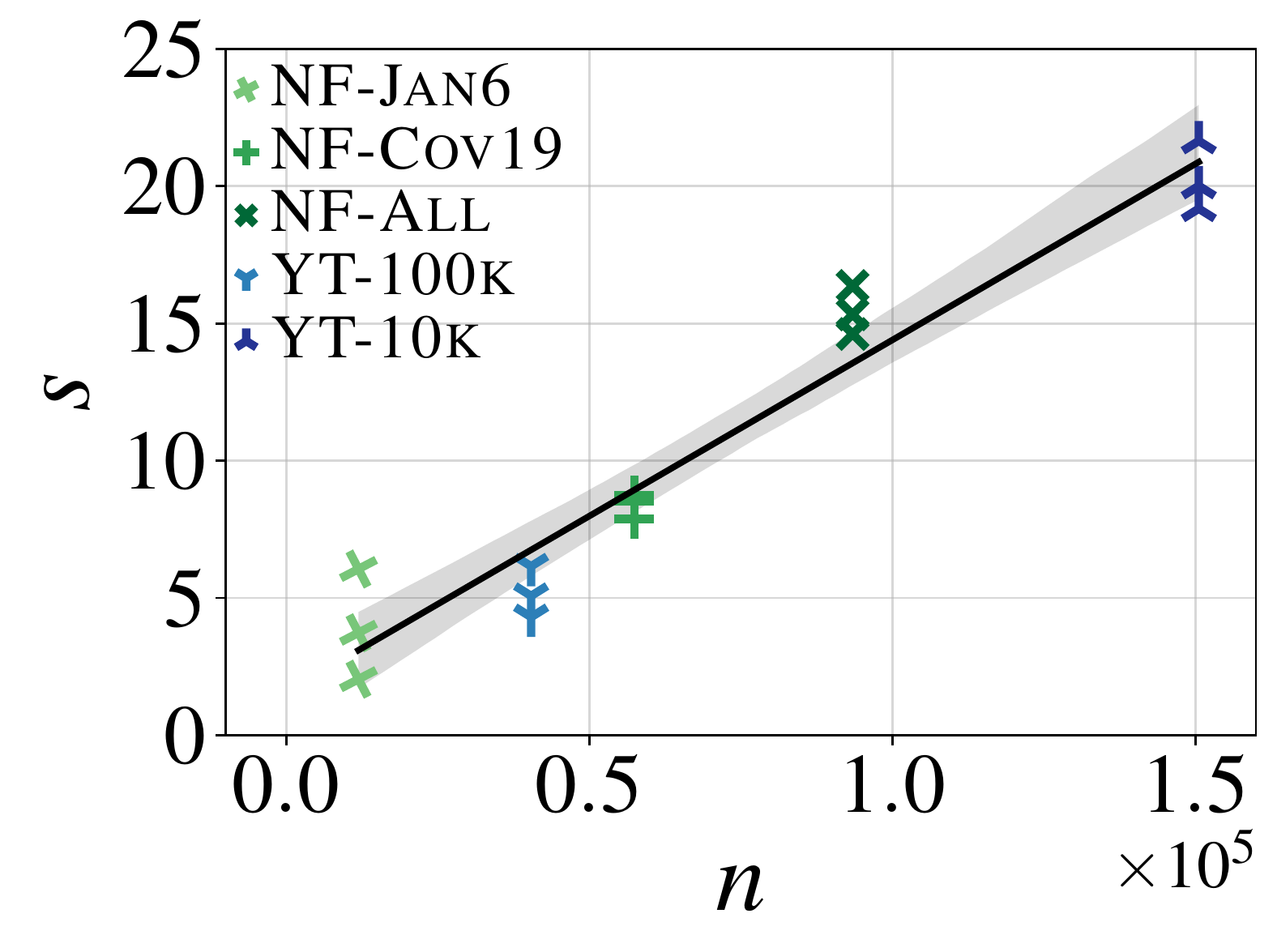}\vspace*{1pt}
		\subcaption{\ourmethod on \ourproblem}
	\end{subfigure}~%
	\begin{subfigure}{0.5\linewidth}
		\centering
		\includegraphics[width=\linewidth]{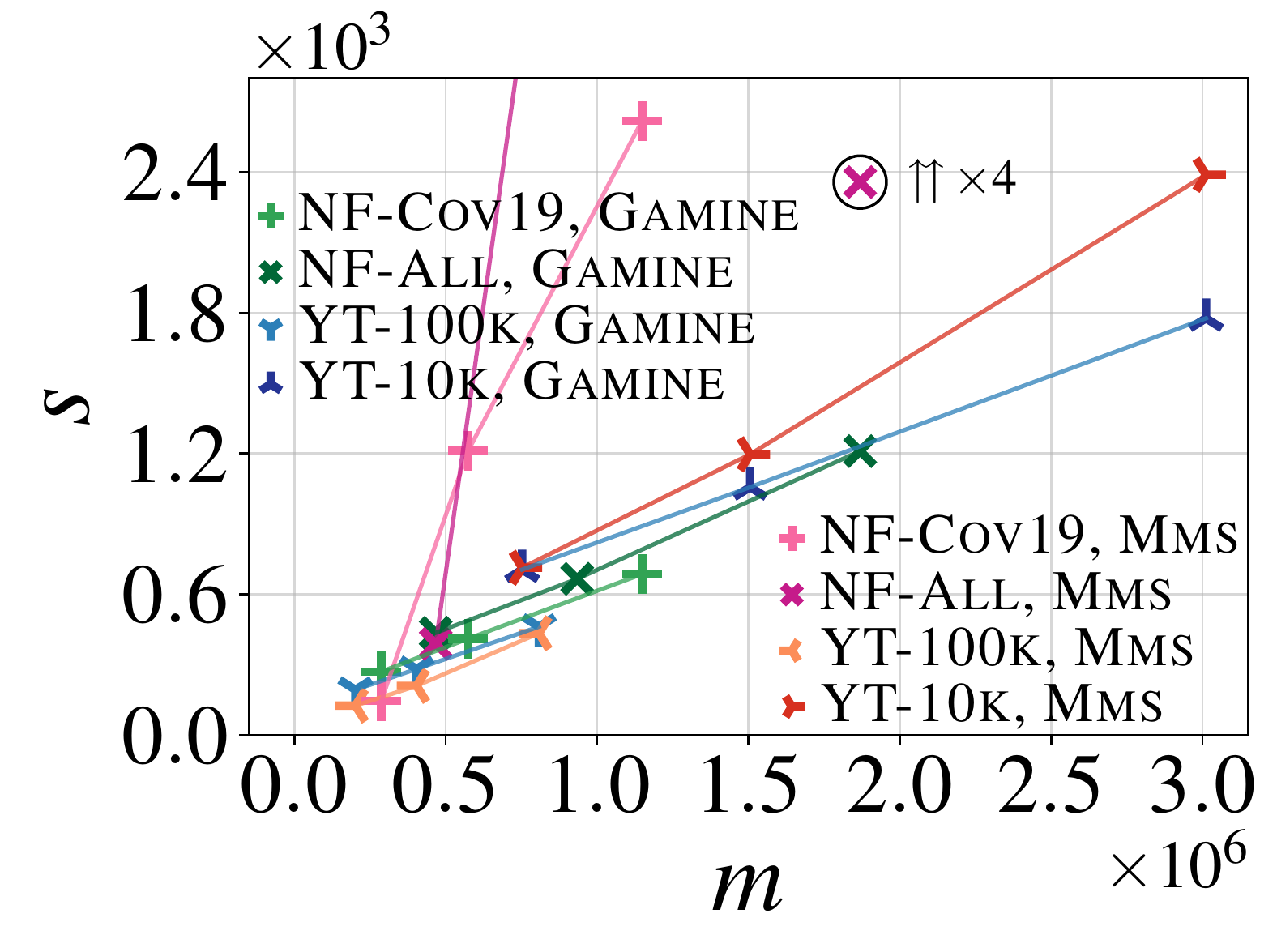}
		\subcaption{\ourmethod and \fabbrialg on \ourproblemtwo}
	\end{subfigure}
	\caption{%
		Empirical scaling of precomputations for \ourmethod and \fabbrialg, computed with $\pabsorption=0.05$, $\probabilityshape=\mathbf{U}$, $\labeling_{B1}$, and $\qualitythreshold=0.0$ (\ourproblem) resp. $0.99$ (\ourproblemtwo).
		We depict scaling for \ourproblem as a function of $\nnodes$, 
		with a linear regression fitted across datasets, 
		and scaling for \ourproblemtwo as a function of $\nedges$, 
		where we connect the observations stemming from different regular out-degrees of the same dataset for \nelatwo, \nelathree, and \yttwo. 
		Precomputations for \ourproblem are much faster than for \ourproblemtwo, 
		and while \ourmethod scales approximately linearly, \fabbrialg scales somewhat unpredictably.
	}
	\label{fig:precomputations}
\end{figure}

\begin{figure}[t]
	\centering
	\includegraphics[width=0.5\linewidth]{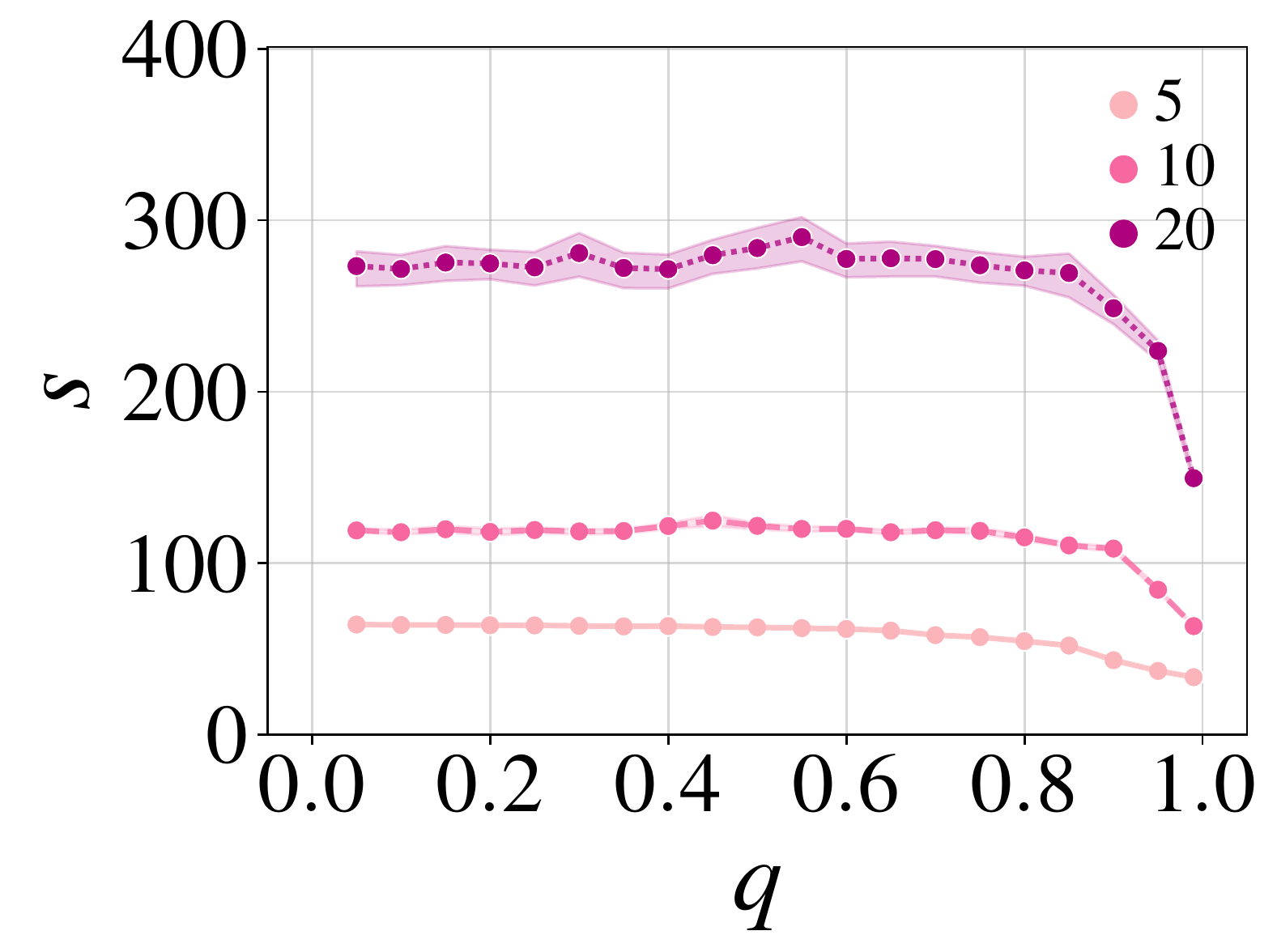}
	\caption{%
		Empirical scaling of \ourmethod as a function of the quality threshold $\qualitythreshold$, 
		run on \ytone with $\pabsorption=0.05$, $\probabilityshape=\mathbf{U}$, and $\labeling_{B1}$, for $\outregulardegree\in\{5,10,20\}$.
		The larger $\qualitythreshold$, the faster \ourmethod.
	}
	\label{fig:scalability:quality}
\end{figure}

\subsubsection{Impact of quality threshold}

In addition to \ourmethod's scaling behavior as a function of $\nnodes$ and $\nedges$, 
for \ourproblemtwo, 
we would like to understand how the scaling behavior of our method depends on the quality threshold $\qualitythreshold$. 
To this end, we run \ourmethod on each of our \ytone datasets with $\qualitythreshold \in \{\nicefrac{x}{100}\mid 0 < x < 100, x\mod 5 = 0\} \cup \{0.99\}$. 
Since increasing $\qualitythreshold$ eliminates rewiring candidates, 
we hope to see the running time decrease as $\qualitythreshold$ increases, 
and we expect a larger acceleration on graphs with higher (regular) out-degrees. 
As reported in \cref{fig:scalability:quality}, 
this is precisely what we find---%
and the dependence on $\qualitythreshold$ is particularly small for our sparser \ytone datasets. 

\begin{figure}[b]
	\centering
	\begin{subfigure}{0.33\linewidth}
		\centering
		\includegraphics[width=\linewidth]{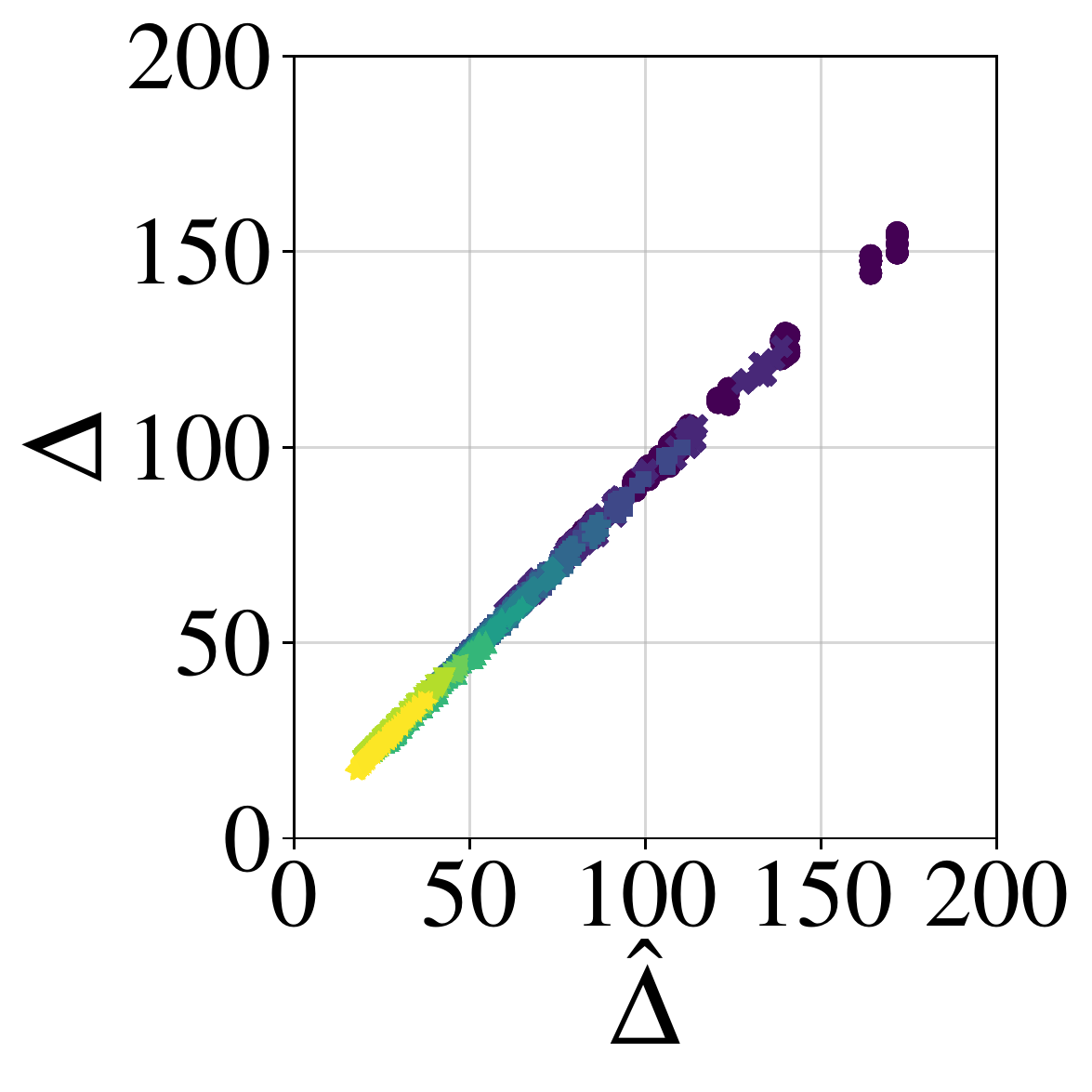}
		\subcaption{\synhom, $\nnodes = 100$}
	\end{subfigure}~%
	\begin{subfigure}{0.33\linewidth}
		\centering
		\includegraphics[width=\linewidth]{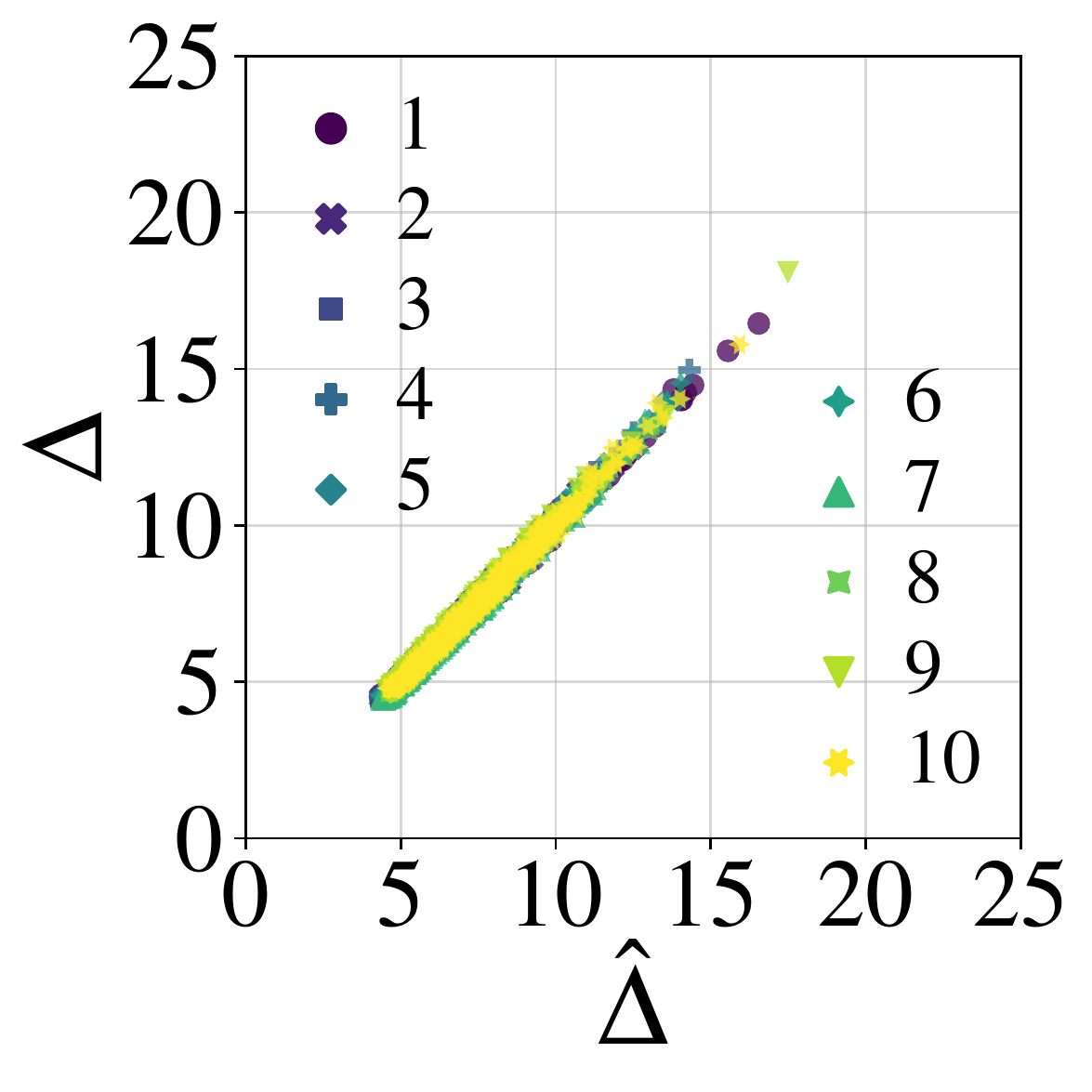}
		\subcaption{\synuni, $\nnodes = 100$}
	\end{subfigure}~%
	\begin{subfigure}{0.33\linewidth}
		\centering
		\includegraphics[width=\linewidth]{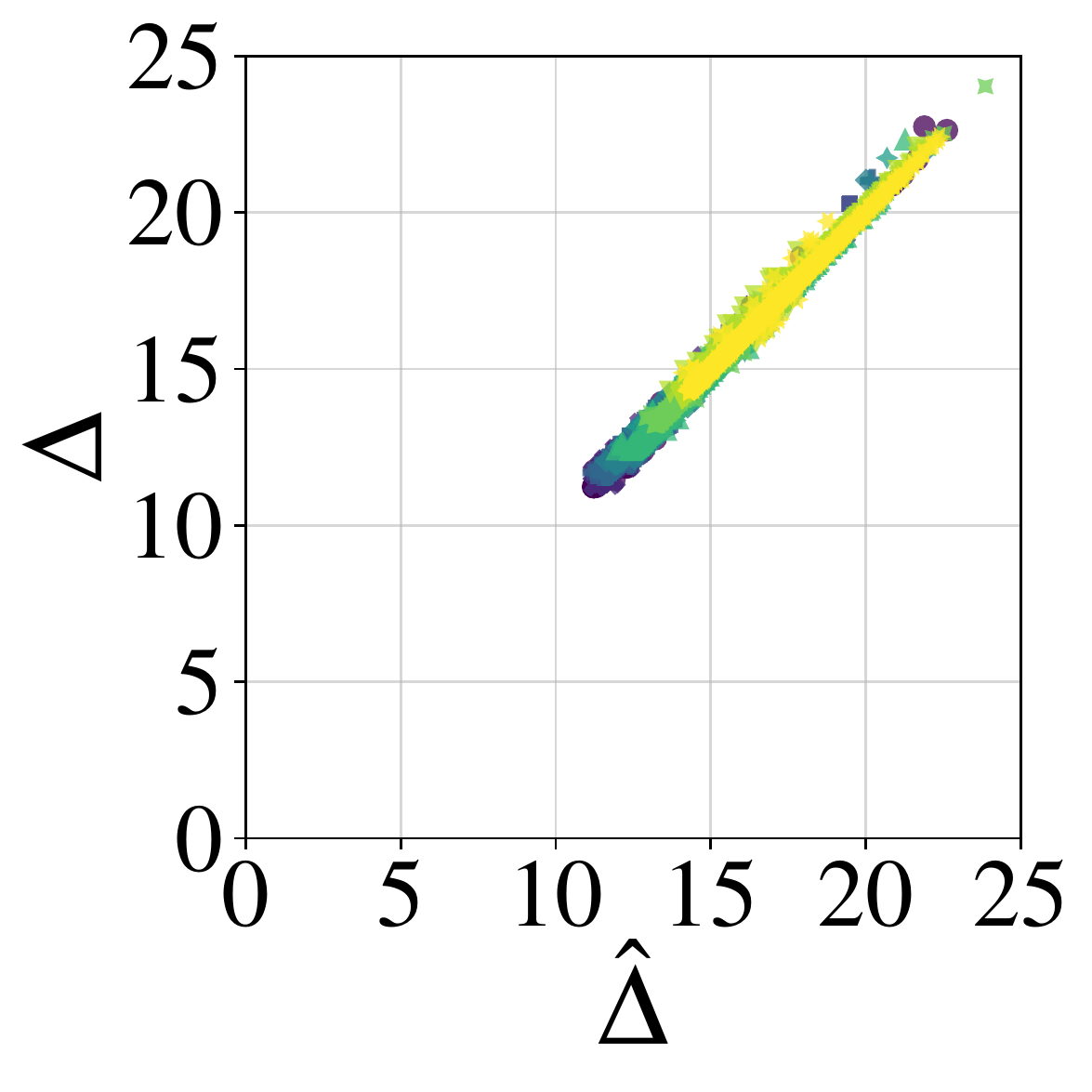}
		\subcaption{\synuni, $\nnodes = 1\,000$}
	\end{subfigure}
	\caption{%
		Correlation of the $\heuristic$ and $\Delta$ values for the $100$ candidates $(i,j,k)$ with the largest $\Delta$, 
		in $10$ rewiring rounds on synthetic graphs with $\pabsorption=0.05$, $\fractionbad = 0.7$, and $\probabilityshape=\textbf{U}$, under our binary costs $\labeling_B$. 
		$\Delta$ and $\heuristic$ are almost perfectly correlated.
	}
	\label{fig:heuristic}
\end{figure}
\begin{figure}[b]
	\centering
	\begin{subfigure}{0.5\linewidth}
		\centering
		\includegraphics[height=3.065cm]{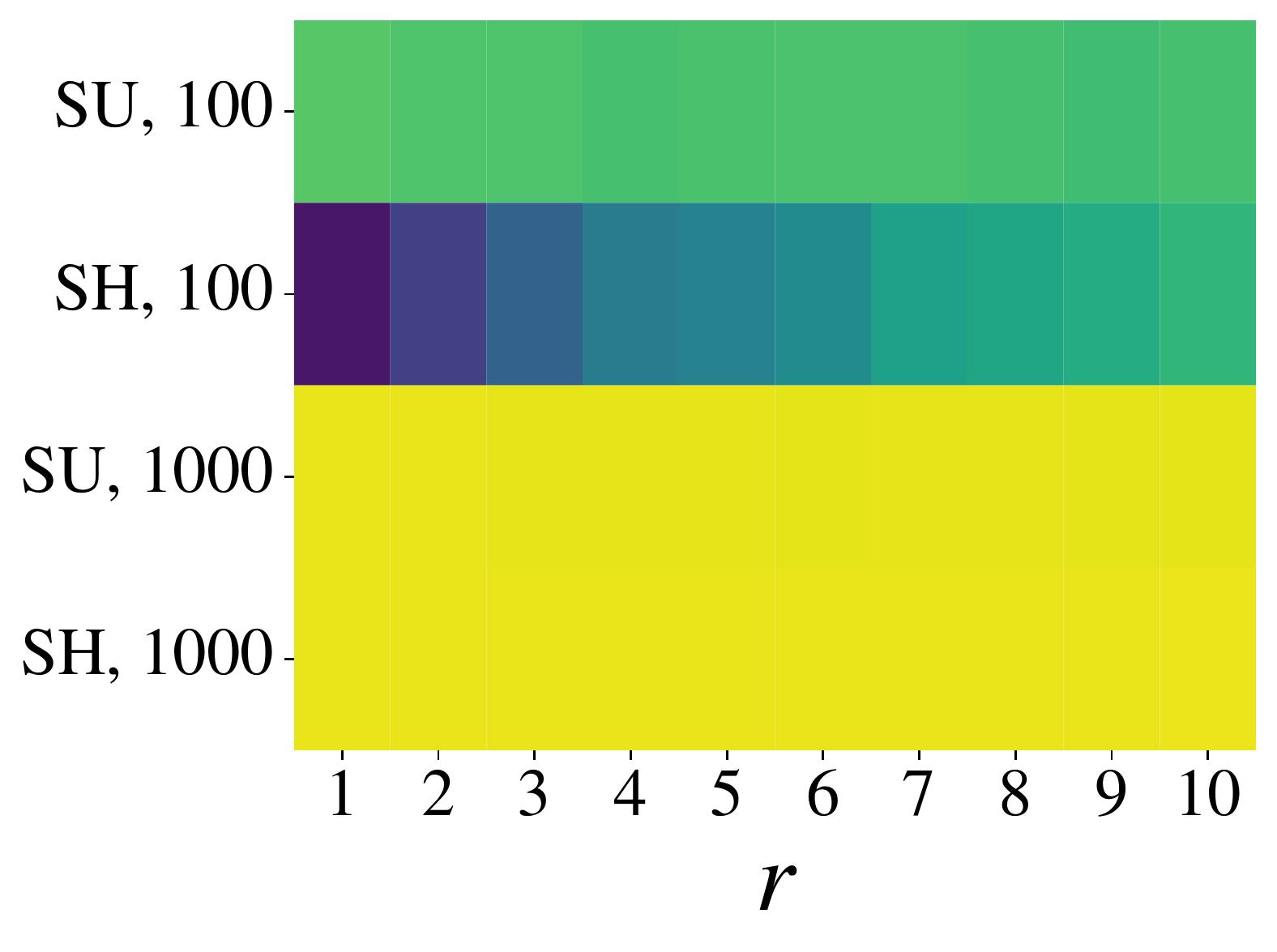}\vspace*{0.5pt}
		\subcaption{Pearson's Correlation}
	\end{subfigure}~%
	\begin{subfigure}{0.5\linewidth}
		\centering
		\includegraphics[width=\linewidth]{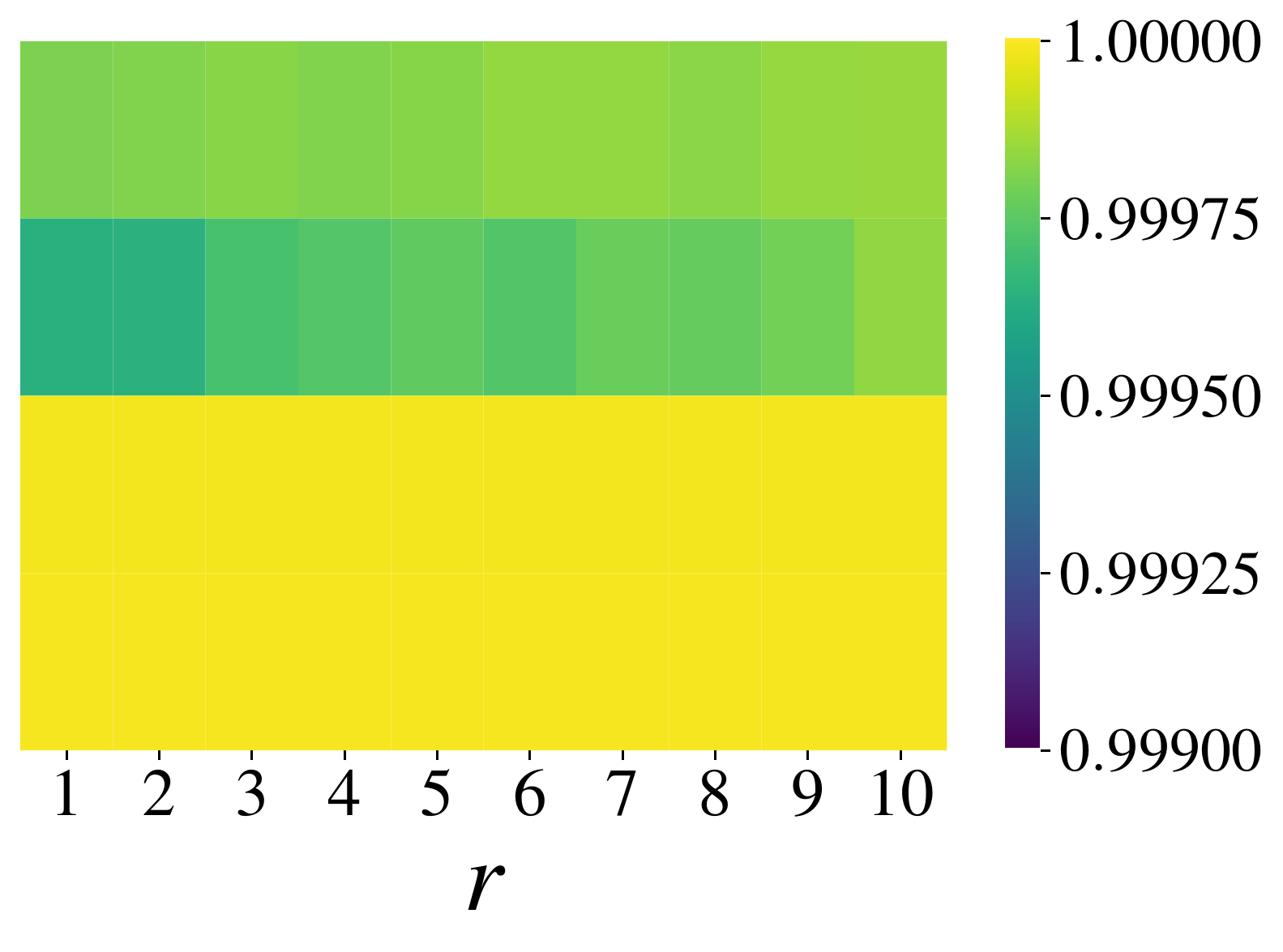}
		\subcaption{Spearman's Correlation}
	\end{subfigure}~%
	\caption{Pearsons' product-moment correlation and Spearman's rank correlation between $\Delta$ and $\heuristic$, 
		in $10$ rewiring rounds on synthetic graphs with $\pabsorption=0.05$, $\fractionbad = 0.7$, and $\probabilityshape=\textbf{U}$, under binary costs. 
		Both correlations are almost perfect across all rounds, 
		and they are even closer to $1$ for the larger synthetic graphs than for the smaller ones.
	}
	\label{fig:correlations}
\end{figure}

\subsection[Impact of Using Delta-Hat Instead of Delta]{Impact of Using $\heuristic$ Instead of $\Delta$}
\label{apx:exp:heuristic}

Having confirmed in the main text that \ourmethod scales linearly not only in theory but also in practice (cf. \cref{fig:scalability}), 
we would like to ensure that moving from $\Delta$ to $\heuristic$, which enables this scalability, 
has little impact on the quality of our results. 
To this end, we investigate the relationship between $\Delta$ and $\heuristic$ on the smallest instances of our synthetic graphs, \synuni and \synhom. 
As illustrated in \cref{fig:heuristic}, $\Delta$ and $\heuristic$ are almost perfectly correlated,
and \cref{fig:correlations} shows that this holds not only for the top-ranked candidates but for all candidates, 
under both product-moment correlation and, more importantly, rank correlation. 
Thus, we are confident that our reliance on $\heuristic$, rather than $\Delta$, 
to select greedy rewirings hardly degrades our results.

\subsection{Performance on the NELA-GT Datasets}
\label{apx:exp:nela}

Whereas on the \yt datasets,
$100$ rewirings with \ourmethod reduce the expected total exposure to harm by $50\%$ while guaranteeing recommendations still $95\%$ as relevant as the original recommendations (\cref{fig:quality_threshold}),
the reduction we achieve on the \nf datasets is more moderate.
As illustrated in \cref{fig:nela-performance}, 
our best result here is a reduction of the expected total exposure to harm by about $30\%$, again under a $95\%$ quality guarantee.
Notably, changing the quality threshold $\qualitythreshold$ has a smaller impact on the \nf than on the \yt datasets, 
and sometimes it has no performance impact at all.
In fact, for the \nelaone graphs involved in \cref{fig:nela-performance}, 
$\qualitythreshold = 0.95 \equiv 0.9 \equiv 0.5$ (and hence, we only draw the line for $\qualitythreshold = 0.95$).
This indicates that unlike on the \yt datasets, 
on the \nf datasets, 
\ourmethod is actually affected by the restriction of rewirings to the $100$ most relevant candidates, 
which we implement for all real-world datasets (cf.~\cref{exp:alg} and also \cref{apx:datasets:real}, \cref{fig:relevancescores}).


\begin{figure}[t]
	\centering
\begin{subfigure}{0.5\linewidth}
	\centering
	\includegraphics[width=\linewidth]{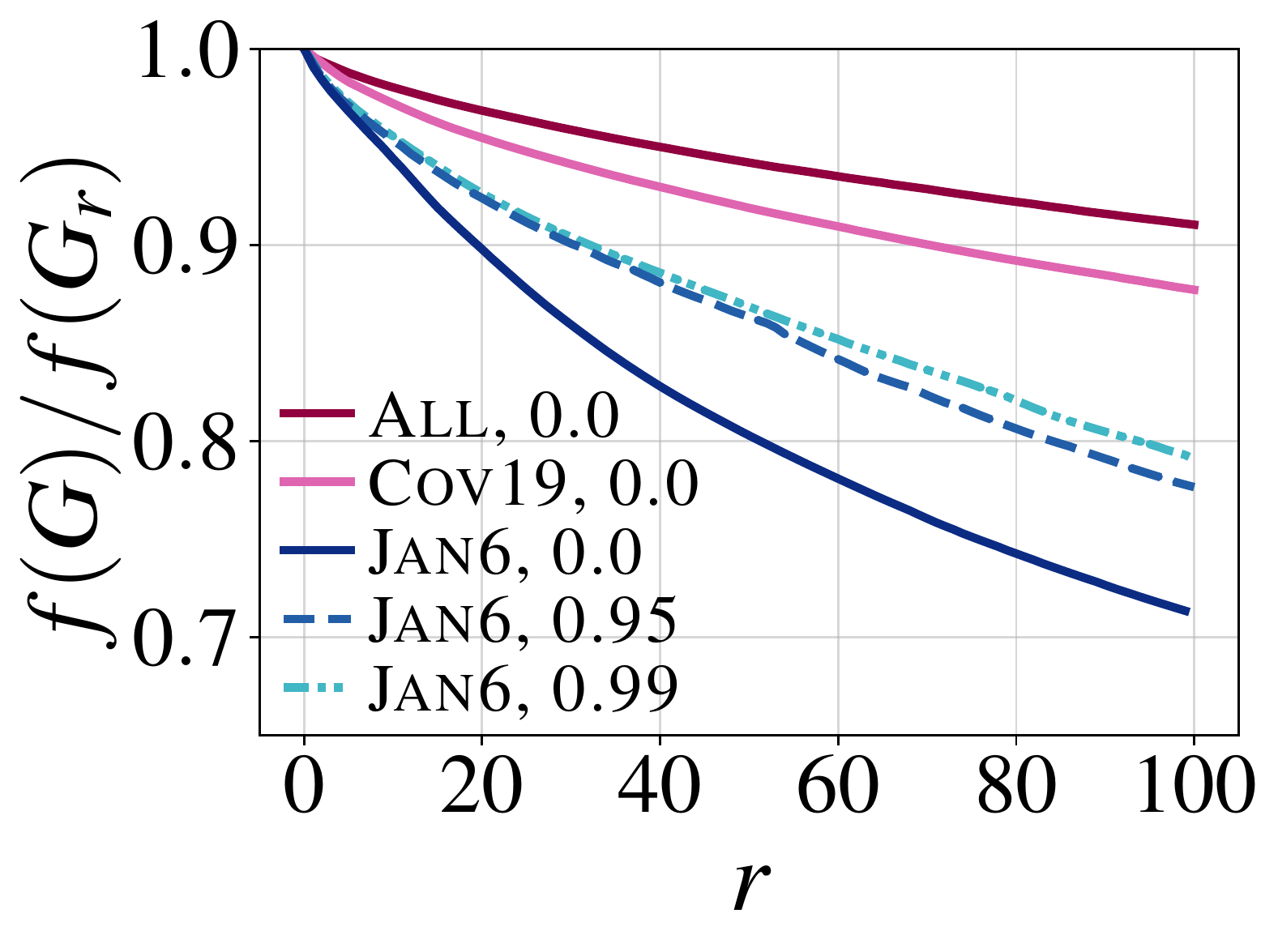}
	\subcaption{Cost function $\labeling_{R2}$}
\end{subfigure}~%
\begin{subfigure}{0.5\linewidth}
	\centering
	\includegraphics[width=\linewidth]{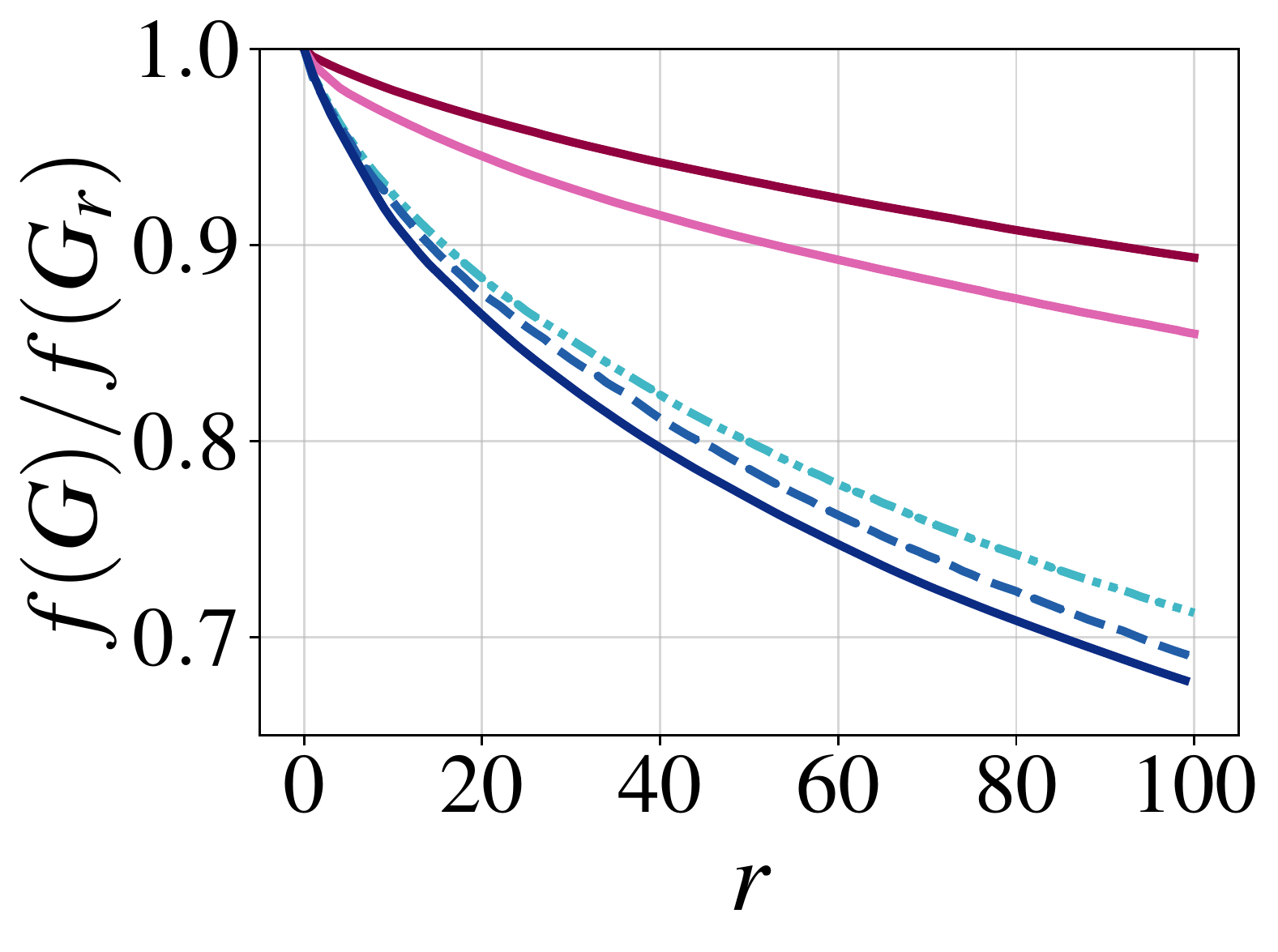}
	\subcaption{Cost function $\labeling_{B1}$}
\end{subfigure}
	\caption{%
		Performance of \ourmethod on our \nf datasets with $\outregulardegree = 5$, $\pabsorption=0.05$, and $\probabilityshape=\textbf{U}$.
		Varying $\qualitythreshold$ on the \nelaone dataset has a much smaller performance impact than what we observed on the \ytone dataset, 
		and the relative reduction in the expected total exposure after $100$ rewirings is at the level of what we achieve after $10$ rewirings on the \ytone dataset (cf.~\cref{fig:quality_threshold}).
	}
	\label{fig:nela-performance}
\end{figure}

\begin{figure}[b]
	\centering
	\begin{subfigure}{0.5\linewidth}
		\centering
		\includegraphics[width=\linewidth]{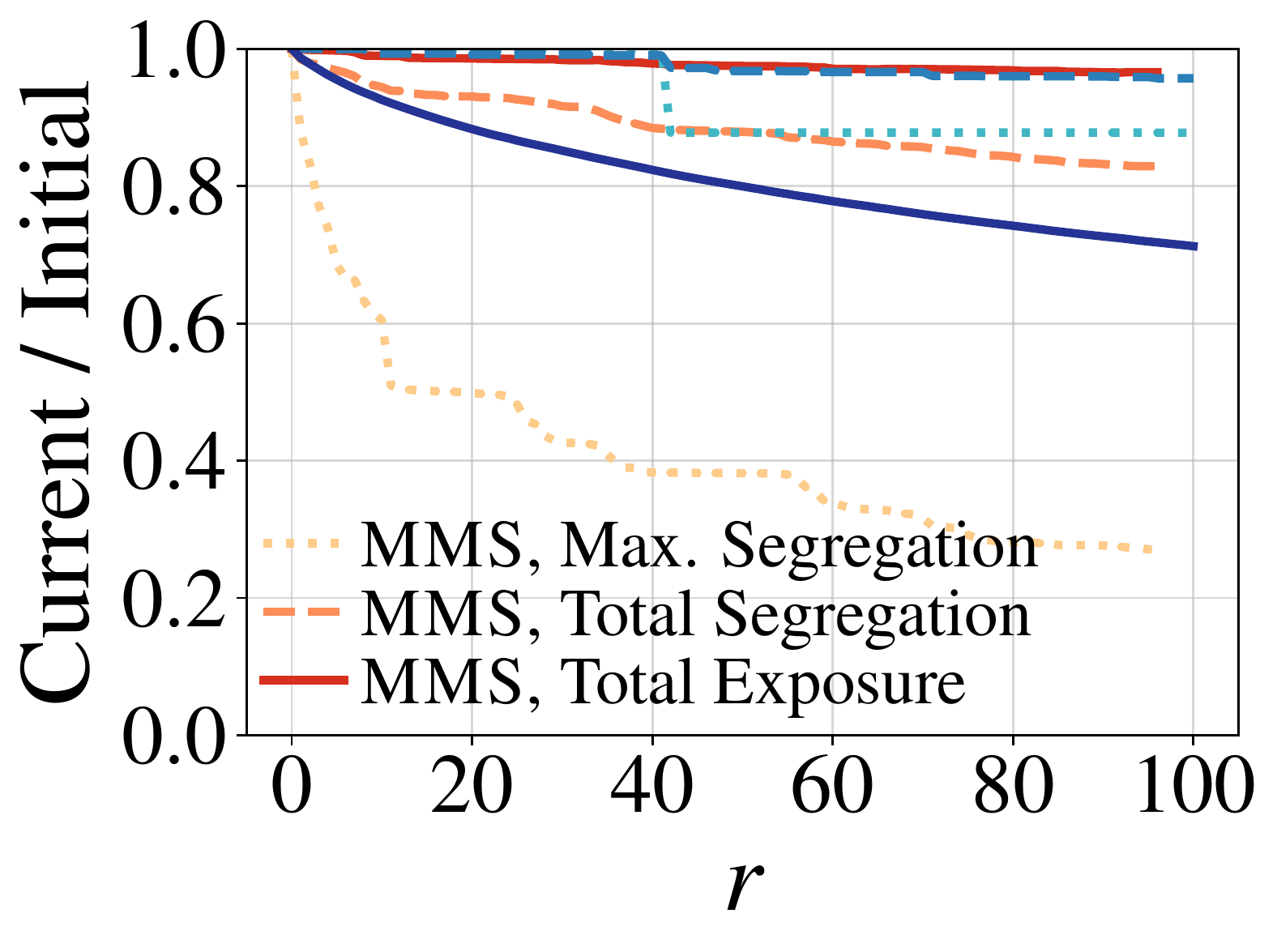}
		\subcaption{Quality threshold $\qualitythreshold = 0.99$}
	\end{subfigure}~%
	\begin{subfigure}{0.5\linewidth}
		\centering
		\includegraphics[width=\linewidth]{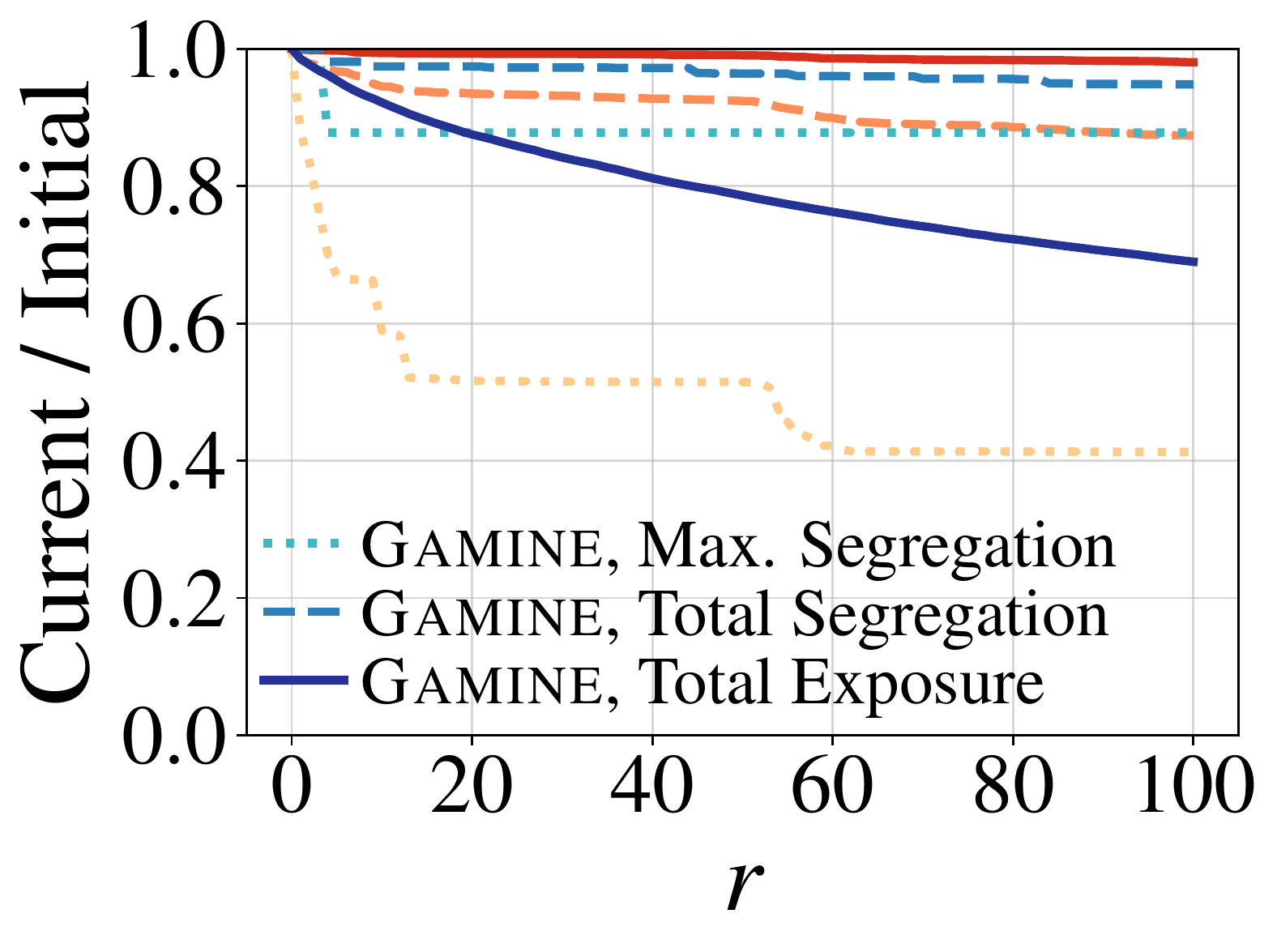}
		\subcaption{Quality threshold $\qualitythreshold = 0.5$}
	\end{subfigure}
	\caption{%
		Performance of \ourmethod and \fabbrialg 
		when measured under $c_{B1}$
		by the maximum segregation or the total segregation from \citeauthor{fabbri2022rewiring} \cite{fabbri2022rewiring}, 
		or by the total exposure as defined in \cref{eq:harm},
		run on \nelaone with $\outregulardegree = 5$, $\pabsorption=0.05$, and $\probabilityshape=\mathbf{U}$. 
		\ourmethod outperforms \fabbrialg on the exposure objective, 
		but \fabbrialg reduces its segregation objective more strongly. 
		Counterintuitively, \fabbrialg achieves a \emph{stronger} reduction of its segregation objective under a \emph{stricter} quality threshold.
	}
\label{fig:nela-mms}
\end{figure}

As illustrated in \cref{fig:nela-mms}, 
on the \nelaone dataset under $\labeling_{B1}$,
\ourmethod still outperforms \fabbrialg on our exposure objective, 
but \fabbrialg achieves a much stronger relative reduction of its segregation objective.
However, \fabbrialg counterintuitively reduces its objective function more strongly under a \emph{stricter} quality threshold---%
a behavior we never observe with \ourmethod under our exposure objective.
As given the same recommendation sequence $\outseq{i}$ at node $i$, 
a rewiring $\rewiring{i}{j}{k}$ 
that is  $\qualitythreshold$-permissible under $\qualitythreshold=0.99$ is also $\qualitythreshold$-permissible under $\qualitythreshold=0.5$,
this suggests that \fabbrialg is highly dependent on its trajectory and sometimes requires greedily suboptimal choices to obtain the best possible result after $\budget$ rewirings.

Moreover, the promising performance we observe for \fabbrialg on \nelaone under $\labeling_{B1}$ does not carry over to \nelatwo and \nelathree, or even to other cost functions on \nelaone:
On \nelatwo and \nelathree under $\labeling_{B1}$, 
and on \nelaone under $\labeling_{B2}$ or $\labeling_{R2}$ with binarization threshold $\roundingthreshold=1.0$, 
\fabbrialg cannot reduce its segregation objective at all, 
\emph{even though} the starting value of the maximum segregation is exactly the same as for \nelaone under $\labeling_{B1}$ (cf.~\cref{apx:datasets:real}, \cref{tab:nela-initial-exposures}).
On \nelaone under $\labeling_{R2}$ with binarization threshold $\roundingthreshold=1.0$, 
\fabbrialg stops after four rewirings with a reduction of $25\%$, 
but the maximum segregation is already miniscule from the start.
Thus, our experiments on \nf data confirm our impression from the main paper that \fabbrialg less robust than \ourmethod.

\subsection{Edge Statistics of Rewired Edges}
\label{apx:experiments:edgestats}

In the main text, we reported that \ourmethod frequently rewires edges $\edge{i}{j}$ with a comparatively large sum of in-degrees.
To corroborate this claim, in \cref{fig:indegree-sums}, 
we compare the distribution of normalized in-degree sums for edges rewired by \ourmethod to that of all edges. 
We find that the distribution of in-degree sums for edges rewired by \ourmethod has a higher median than the distribution of in-degree sums for all edges, 
and that the former is generally shifted toward higher in-degree sums as compared to the latter.
\balance

\begin{figure}[b]
	\centering
	\begin{subfigure}{0.5\linewidth}
		\centering
		\includegraphics[width=\linewidth]{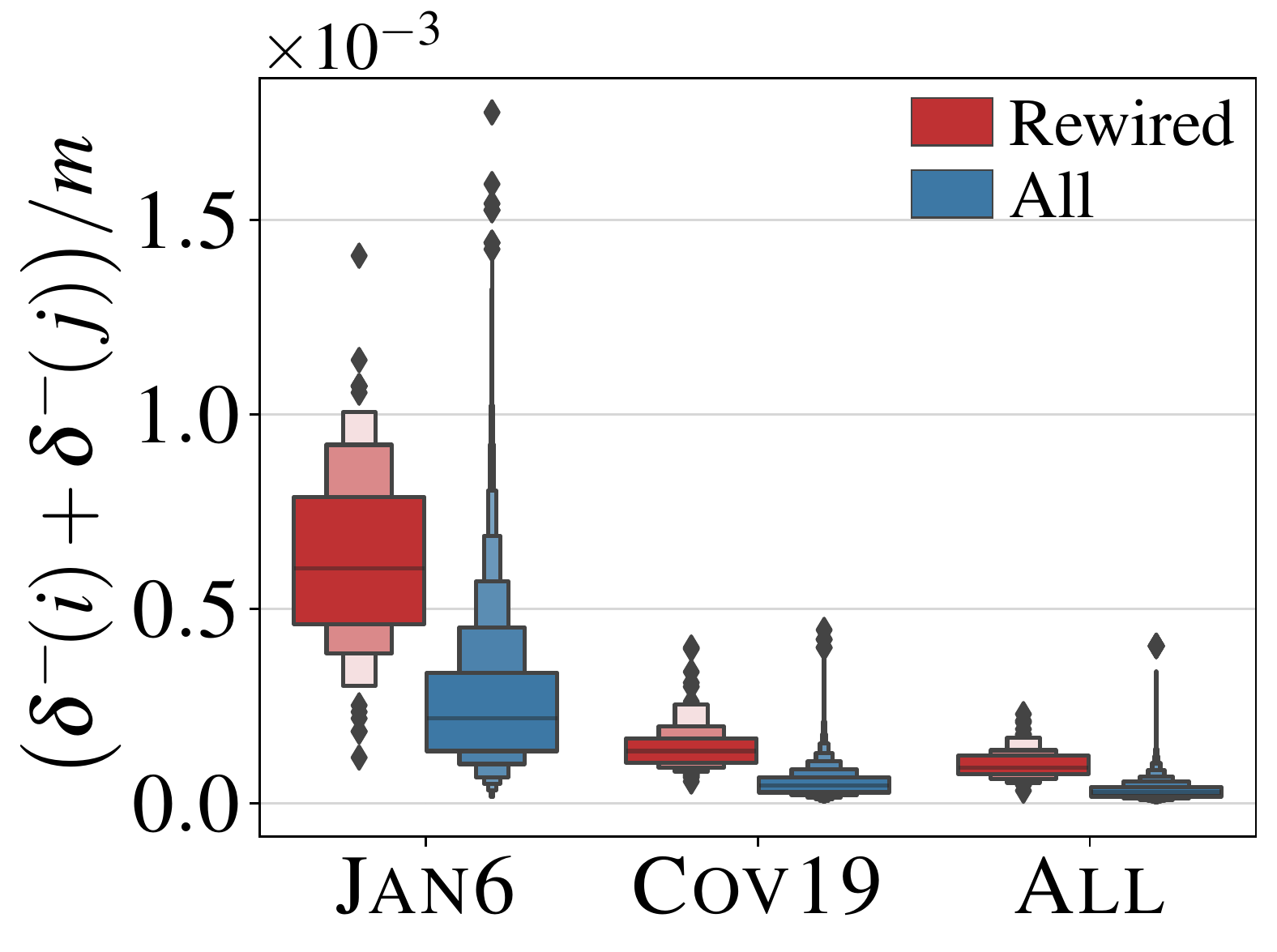}
		\subcaption{\nf, $\qualitythreshold=0.0$}
	\end{subfigure}~%
	\begin{subfigure}{0.5\linewidth}
		\centering
		\includegraphics[width=\linewidth]{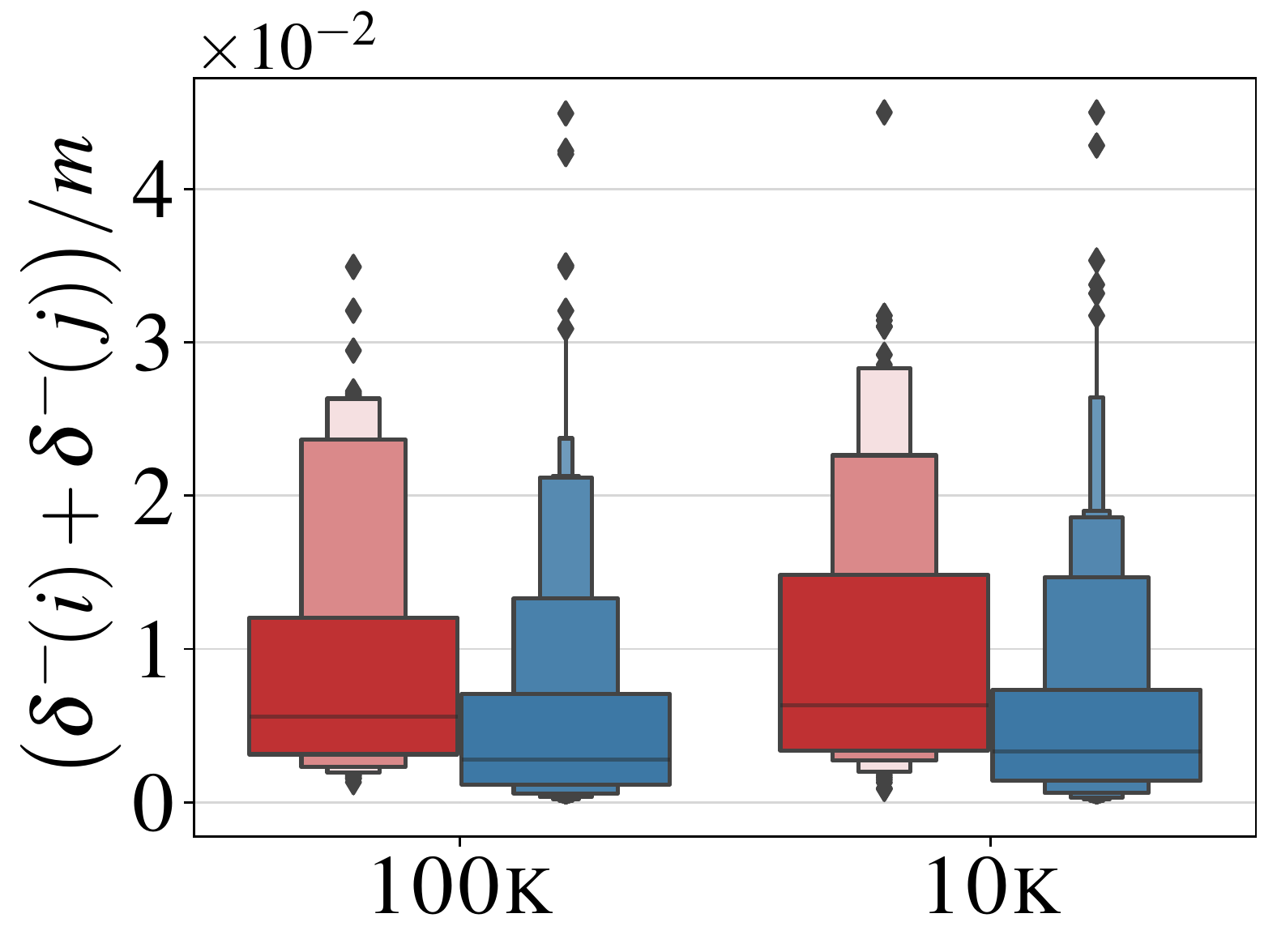}
		\subcaption{\yt, $\qualitythreshold=0.5$}
	\end{subfigure}
	\caption{%
		Distributions of normalized in-degree sums $\frac{\indegree{i} + \indegree{j}}{\nedges}$ for edges rewired by \ourmethod under $\labeling_{B1}$
		vs. all edges, 
		on real-world graphs with $\outregulardegree=5$, 
		$\pabsorption=0.05$,
		and $\probabilityshape=\mathbf{U}$. 
		\ourmethod tends to rewire edges with larger in-degree sums.
	}
	\label{fig:indegree-sums}
\end{figure}

\end{document}